\pdfoutput=1
\RequirePackage{ifpdf}
\ifpdf 
\documentclass[pdftex]{sigma}
\else
\documentclass{sigma}
\fi

\usepackage{tikz,pict2e}

\numberwithin{equation}{section}

\newtheorem{Theorem}{Theorem}[section]
\newtheorem{Lemma}[Theorem]{Lemma}
\newtheorem{Proposition}[Theorem]{Proposition}
 { \theoremstyle{definition}
\newtheorem{Remark}[Theorem]{Remark} }

\begin{document}

\allowdisplaybreaks

\newcommand{\arXivNumber}{1708.02519}

\renewcommand{\thefootnote}{}

\renewcommand{\PaperNumber}{018}

\FirstPageHeading

\ShortArticleName{Asymptotics for Hankel Determinants Associated to a Hermite Weight}

\ArticleName{Asymptotics for Hankel Determinants Associated\\ to a Hermite Weight with a Varying Discontinuity\footnote{This paper is a~contribution to the Special Issue on Orthogonal Polynomials, Special Functions and Applications (OPSFA14). The full collection is available at \href{https://www.emis.de/journals/SIGMA/OPSFA2017.html}{https://www.emis.de/journals/SIGMA/OPSFA2017.html}}}

\Author{Christophe CHARLIER~$^\dag$ and Alfredo DEA\~{N}O~$^\ddag$}

\AuthorNameForHeading{C.~Charlier and A.~Dea\~{n}o}

\Address{$^\dag$~Department of Mathematics, KTH Royal Institute of Technology,\\
\hphantom{$^\dag$}~Lindstedtsv\"{a}gen 25, SE-114 28 Stockholm, Sweden}
\EmailD{\href{mailto:cchar@kth.se}{cchar@kth.se}}

\Address{$^\ddag$~School of Mathematics, Statistics and Actuarial Science, University of Kent,\\
\hphantom{$^\ddag$}~Canterbury CT2 7FS, UK}
\EmailD{\href{mailto:A.Deano-Cabrera@kent.ac.uk}{A.Deano-Cabrera@kent.ac.uk}}

\ArticleDates{Received November 02, 2017, in f\/inal form February 27, 2018; Published online March 07, 2018}

\Abstract{We study $n\times n$ Hankel determinants constructed with moments of a Hermite weight with a Fisher--Hartwig singularity on the real line. We consider the case when the singularity is in the bulk and is both of root-type and jump-type. We obtain large $n$ asymptotics for these Hankel determinants, and we observe a critical transition when the size of the jumps varies with $n$. These determinants arise in the thinning of the generalised Gaussian unitary ensembles and in the construction of special function solutions of the Painlev\'e IV equation.}


\Classification{30E15; 35Q15; 15B52; 33E17}

\renewcommand{\thefootnote}{\arabic{footnote}}
\setcounter{footnote}{0}

\section{Introduction and motivation}

We consider the Hankel determinant
\begin{gather}\label{Hn}
H_{n}(v,s,\alpha)=\det\left( \int_{\mathbb{R}} x^{j+k} w(x;v,s,\alpha){\rm d}x \right)_{j,k=0}^{n-1}, \qquad n \in \mathbb{N},
\end{gather}
with a Gaussian weight on the real line of the form
\begin{gather}\label{wHermite}
w(x;v,s,\alpha) = e^{-x^2} |x-v|^{\alpha} \begin{cases}
s, & \mbox{if } x<v,\\
1, & \mbox{if } x>v,
\end{cases}
\end{gather}
where $v \in \mathbb{R}$, $s\in [0,1]$ and $\alpha \in (-1,\infty)$. There is a root-type Fisher--Hartwig (FH) singularity if $\alpha \neq 0$. The piecewise constant factor in \eqref{wHermite} is a jump-type FH singularity only if $s \neq 0$ and $s\neq 1$. If $s=1$ there is no jump, and if $s = 0$ the weight is supported on the interval $[v,\infty)$. By Heine's formula, $H_{n}(v,s,\alpha)$ admits the following $n$-fold integral representation:
\begin{gather}\label{integral representation for H_n}
H_{n}(v,s,\alpha) = \frac{1}{n!}\int_{\mathbb{R}^{n}} \Delta(x)^{2} \prod_{i=1}^{n}w(x_{i};v,s,\alpha){\rm d}x_{i}, \qquad \Delta(x) = \prod_{1\leq i <j \leq n} (x_{j}-x_{i}).
\end{gather}
In this paper we are interested in large $n$ asymptotics for $H_{n}(v,s,\alpha)$, uniformly in $s \in [0,1]$. An analogous case was analysed in \cite{BotDeiItsKra} for Fredholm determinant associated with the sine kernel and in~\cite{ChCl1,ChCl2} for Toeplitz determinants with a weight def\/ined on the unit circle. We brief\/ly summarize here some known results for particular values of the parameters and we present some applications.

The Hankel determinant $H_{n}(v,s,\alpha)$ arises in random matrix theory. Consider the set of $n \times n$ Hermitian matrices $M$ endowed with the probability measure
\begin{gather*}
\frac{1}{\widehat{Z}_{n}(v,\alpha)}|\det (M)-vI|^{\alpha} e^{-\operatorname{Tr}M^{2}}{\rm d}M, \qquad {\rm d}M = \prod_{i=1}^{n} {\rm d}M_{ii} \prod_{1\leq i<j\leq n} {\rm d}\Re M_{ij} {\rm d}\Im M_{ij},
\end{gather*}
where $\widehat{Z}_{n}(v,\alpha)$ is the normalisation constant. We will refer to this random matrix ensemble as the \textit{generalised Gaussian unitary ensemble}, which we denote by ${\rm GUE}(v,\alpha)$. Such a measure of matrices $M$ induces a probability measure on the eigenvalues $x_{1},\dots,x_{n}$ of $M$ which is of the form
\begin{gather}\label{lol 23}
\frac{1}{n! Z_{n}(v,\alpha)} \Delta(x)^{2} \prod_{i=1}^{n} w(x_{i};v,1,\alpha){\rm d}x_{i},
\end{gather}
where $Z_{n}(v,\alpha)$ is called the partition function of ${\rm GUE}(v,\alpha)$. From~\eqref{lol 23} and the integral representation for Hankel determinants given by~\eqref{integral representation for H_n}, we have the relation
\begin{gather*}
Z_n(v,\alpha) = H_n(v,1,\alpha).
\end{gather*}
The special case of ${\rm GUE}(0,0)$ (note that if $\alpha = 0$, the parameter $v$ is irrelevant) is called the Gaussian unitary ensemble~(GUE), and has already been widely studied (see, e.g.,~\cite{Mehta}). The partition function of the GUE is a Selberg integral and is explicitly known (see, e.g., \cite{Mehta}). Its exact expression and its large $n$ asymptotics are given by
\begin{gather}
 \log H_{n}(0,1,0) = -\frac{n^{2}}{2}\log 2+ \frac{n}{2}\log(2\pi)+\sum_{j=1}^{n-1}\log (j!) \label{lol 23b}\\
\hphantom{\log H_{n}(0,1,0)}{} = \frac{n^{2}}{2} \log \left( \frac{n}{2} \right) -\frac{3}{4}n^{2} + n \log(2\pi)-\frac{\log n}{12} + \zeta^{\prime}(-1)+\mathcal{O}\big(n^{-1}\big), \quad \mbox{as} \ n \to \infty, \nonumber
\end{gather}
where $\zeta$ is Riemann's zeta-function.

When $\alpha \neq 0$ but $v=0$, the partition function of ${\rm GUE}(0,\alpha)$ is again a Selberg integral and is also known exactly for f\/inite~$n$, see~\cite{MehtaNormand} or \cite[equation~(A.8)]{DS}. This is not true for $v \neq 0$. In \cite{Krasovsky}, Krasovsky obtained via the Riemann--Hilbert method that large $n$ asymptotics of $H_{n}(\sqrt{2n}t,1,\alpha)$, when $t$ is in a compact subset of $(-1,1)$, are given by
\begin{gather}
 \log \frac{H_{n}\big(\sqrt{2n}t,1,\alpha\big)}{H_{n}(0,1,0)} = \frac{\alpha}{2}n\log n - \frac{\alpha}{2}\big(1-2t^{2}+\log 2\big)n + \frac{\alpha^{2}}{4}\log n + \frac{\alpha^{2}}{4}\log \big(2\sqrt{1-t^{2}}\big)\nonumber \\
 \hphantom{\log \frac{H_{n}\big(\sqrt{2n}t,1,\alpha\big)}{H_{n}(0,1,0)} =}{}+ \log \frac{G\big(1+\frac{\alpha}{2}\big)^{2}}{G(1+\alpha)} + \mathcal{O}\left( \frac{\log n}{n} \right),\label{lol 22}
\end{gather}
where $G$ is Barnes' $G$-function.

Let us denote by $x_{\min}^{(v,\alpha)}$ and $x_{\max}^{(v,\alpha)}$ the smallest and largest eigenvalue in ${\rm GUE}(v,\alpha)$, respectively. The probability of observing no eigenvalues in $(-\infty,v)$, denoted by $\mathbb{P}\big( x_{\min}^{(v,\alpha)} \geq v \big)$, can be expressed as a ratio of Hankel determinants given by \eqref{Hn} with $s=0$ and $s=1$. From a direct integration of \eqref{lol 23}, and from the symmetry $w(x;v,1,\alpha) = w(-x;-v,1,\alpha)$, we have
\begin{gather}\label{prob x_min}
\mathbb{P}\big( x_{\min}^{(v,\alpha)} \geq v \big) = \frac{H_{n}(v,0,\alpha)}{H_{n}(v,1,\alpha)} = \mathbb{P}\big( x_{\max}^{(-v,\alpha)} \leq -v \big).
\end{gather}
It is well-known that the empirical spectral distribution of the eigenvalues (after proper rescaling) in the GUE converges weakly almost surely to the Wigner semi-circle distribution, i.e.,
\begin{gather*}
\frac{1}{n}\sum_{i=1}^{n}\delta_{\frac{x_{i}}{\sqrt{2n}}}
\rightarrow \frac{2}{\pi}\sqrt{1-x^{2}}{\rm d}x,
\end{gather*}
see for instance \cite[Chapter~2]{AGZ}. It is also known that the smallest eigenvalue $x_{\min}^{(0,0)}$ is usually located near $-\sqrt{2n}$ and the properly rescaled f\/luctuations of $x_{\min}^{(0,0)}$ around $-\sqrt{2n}$ follow the Tracy--Widom distribution. The ratio~\eqref{prob x_min} with $v = -\sqrt{2n} \big( 1+\frac{w}{2n^{2/3}} \big)$ and~$w$ in a compact subset of~$\mathbb{R}$, i.e., when~$v$ is near the edge of the spectrum, has been recently studied for $\alpha \neq 0$ in~\cite{WuXuZhao} (including general~$s>0$).

In this paper we investigate the probability of a large deviation of $x_{\min}^{(v,\alpha)}$, i.e., the probability~\eqref{prob x_min} when $v$ is suf\/f\/iciently far from~$-\sqrt{2n}$. In the particular case of $v = 0$, \eqref{prob x_min} is the probability that a matrix~$M$ drawn from the ${\rm GUE}(0,\alpha)$ is positive def\/inite. Large $n$ asymptotics for this probability have been obtained in \cite{DS}:
\begin{gather}\label{lol 38}
\log \frac{H_{n}(0,0,\alpha)}{H_{n}(0,1,\alpha)} = - \frac{\log 3}{2} n^{2} - \frac{\alpha \log 3}{2}n + \left( \frac{\alpha^{2}}{4}-\frac{1}{12} \right) \log n + c_{0}+ \mathcal{O}\big(n^{-1}\big),
\end{gather}
where
\begin{gather*}
c_{0} = \frac{\alpha}{2}\log(2\pi) + \left( \frac{\alpha^{2}}{4}-\frac{1}{6} \right) \log 2 + \left( \frac{1}{8} - \frac{\alpha^{2}}{2} \right) \log 3 + \zeta^{\prime}(-1) - \log \left[G\left( 1+\frac{\alpha}{2} \right)^{2}\right].
\end{gather*}
Note that the denominator $H_{n}(0,1,\alpha)$ can be obtained from \eqref{lol 23b} and \eqref{lol 22} with $t = 0$. One of the goals of the present paper is to generalize this result for $v \neq 0$, i.e., to obtain strong large~$n$ asymptotics of~$H_{n}\big(\sqrt{2n}t,0,\alpha\big)$, when~$t$ is in a~compact subset of~$(-1,\infty)$. In particular, it is possible to deduce from our result the probability~\eqref{prob x_min}.

Assume we thin the eigenvalues $x_{1},\dots,x_{n}$ from \eqref{lol 23} by removing each of them independently with a certain probability $s \in [0,1]$. The resulting point process is called the \textit{thinned ${\rm GUE}(v,\alpha)$}, whose spectrum is denoted by $y_{1},\dots,y_{m}$, where $m$ is itself a random variable following the Binomial distribution Bin$(n,1-s)$. Thinning was introduced in random matrix theory by Bohigas and Pato~\cite{BohigasPato1}, and we refer to~\cite{Charlier, ChCl2} for analogous situations and an overview of the theory of thinning. Let us denote by $y_{\min}^{(v,s,\alpha)}$ for the smallest thinned eigenvalue (note that we always have $y_{\min}^{(v,s,\alpha)} \geq x_{\min}^{(v,\alpha)}=y_{\min}^{(v,0,\alpha)}$). The ratio $\frac{H_{n}(v,s,\alpha)}{H_{n}(v,1,\alpha)}$ is a one-parameter generalisation of~\eqref{prob x_min} and represents the probability of observing no eigenvalue of the thinned spectrum in~$(-\infty,v)$, i.e., we have
\begin{gather}\label{prob x_min thinned}
\mathbb{P}\big( y_{\min}^{(v,s,\alpha)} \geq v \big) = \frac{H_{n}(v,s,\alpha)}{H_{n}(v,1,\alpha)}.
\end{gather}
In the regime when $s$ is in a compact subset of $(0,1]$ and $v = \sqrt{2n}t$ with $t$ in a compact subset of $(-1,1)$, asymptotics of this probability have been obtained rigorously in~\cite{ItsKrasovsky} for $\alpha = 0$ and in \cite{Charlier} for $\alpha \neq 0$. Note that asymptotics for integer $\alpha$ were obtained previously in~\cite{Garoni}, based on the work~\cite{BH}. As $n \to \infty$, they are given by
\begin{gather}
\log \frac{H_{n}\big(\sqrt{2n}t,s,\alpha\big)}{H_{n}\big(\sqrt{2n}t,1,\alpha\big)} = n \int_{-1}^{t} \frac{2}{\pi}\sqrt{1-x^{2}}{\rm d}x \log s + \frac{(\log s)^{2}}{4\pi^{2}} \log n + \tilde{c}_{0} + \mathcal{O} \left( \frac{\log n}{n} \right)\nonumber\\
\qquad{} =\left[2\arcsin t+2t\sqrt{1-t^2}+\pi\right]\frac{\log s}{2\pi}\, n+ \frac{(\log s)^{2}}{4\pi^{2}} \log n + \tilde{c}_{0} + \mathcal{O} \left( \frac{\log n}{n} \right),\label{lol 39}
\end{gather}
where the constant $\tilde{c}_{0}=\tilde{c}_{0}(t,s,\alpha)$ is explicit:
\begin{gather*}
\tilde{c}_{0} = \frac{3(\log s)^{2}}{4\pi^{2}} \log\big(2\sqrt{1-t^{2}}\big) + \alpha \frac{\log s}{2\pi}\arcsin t + \log \frac{G\big( 1+\frac{\alpha}{2}+\frac{\log s}{2\pi i} \big)G\big( 1+\frac{\alpha}{2}-\frac{\log s}{2\pi i} \big)}{G\big(1+\frac{\alpha}{2}\big)^{2}}.
\end{gather*}
The second goal of the present paper is to obtain large $n$ asymptotics of $H_{n}\big(\sqrt{2n}t,s,\alpha\big)$ when $s = s(n) \to 0$ as $n \to \infty$, and to observe a transition in the large $n$ asymptotics between~\eqref{lol 39}, where $s$ is bounded away from $0$, and the case $H_{n}\big(\sqrt{2n}t,0,\alpha\big)$.

Another motivation for the study of the Hankel determinant \eqref{Hn} comes from solutions of the Painlev\'{e} IV dif\/ferential equation ($\textrm{P}_{\rm IV}$). The $\mathrm{P_{IV}}$ equation f\/inds interesting applications in many dif\/ferent areas of physics, such as non-linear optics, dispersive long-wave equations, f\/luid dynamics and plasma physics (see, e.g., \cite[Section~10.1]{Clarkson}, \cite{Winter}). This equation depends on two parameters $A,B \in\mathbb{C}$ and is given by
\begin{gather*}
q^{\prime\prime}(z)=\frac{1}{2q(z)} q^{\prime}(z)^2+\frac{3}{2}q(z)^3+4zq(z)^2+2\big(z^2-A\big)q(z)+\frac{B}{q(z)}.
\end{gather*}
In this paper, we focus on special function solutions of~$\mathrm{P_{IV}}$ when the parameters~$A$ and~$B$ satisfy suitable constraints. It is known (see, e.g., \cite[Theorem~3.4]{CJ} and \cite[Theorem~25.2]{GLS}) that $\mathrm{P_{IV}}$ has solutions expressible in terms of parabolic cylinder functions $U(a,z)$ (these functions are def\/ined in \cite[Chapter~12]{NIST}) if and only if either
\begin{gather}\label{lol 44}
B=-2(2n+1+\varepsilon A)^2 \qquad \textrm{or} \qquad B=-2n^2,
\end{gather}
with $\varepsilon = \pm 1$ and $n\in\mathbb{Z}$. In this situation, we are considering the so-called special function solutions of $\mathrm{P_{IV}}$.

In general, the standard method to derive special function solutions is by considering associated Riccati equations. We refer the reader to \cite[Section~7]{Clarkson} for the general theory or \cite[Section 32.10]{NIST} for a summary of special function solutions. For $n=0$ in \eqref{lol 44}, the Riccati equation of $\textrm{P}_{\rm IV}$ is
\begin{gather*}
q^{\prime}(z)=\varepsilon q(z)^2+2\varepsilon zq(z)+2\nu,
\end{gather*}
with $\nu=-(1+\varepsilon A)$.
If we set $q(z)=-\varepsilon \varphi_{\nu}'(z)/\varphi_{\nu}(z)$ to linearise this equation, then $\varphi_{\nu}$ satisf\/ies
\begin{gather}\label{ODE_Weber}
\varphi_{\nu}^{\prime\prime}(z)-2\varepsilon z\varphi_{\nu}^{\prime}(z)+2\varepsilon \nu\varphi_{\nu}(z)=0,
\end{gather}
whose general solutions can be written in terms of the parabolic cylinder function $U(a,z)$, see \cite[Section 12.2]{NIST}. If $\nu\notin \mathbb{Z}$, then $\varphi_{\nu}(z) = \varphi_{\nu}(z;\varepsilon)$ can be written in the form
\begin{gather}\label{phiU}
\varphi_{\nu}(z;\varepsilon)=
\begin{cases}
\big\{C_1 U\big({-}\nu-\tfrac{1}{2},\sqrt{2} z\big)+
C_2 U\big({-}\nu-\tfrac{1}{2},-\sqrt{2} z\big)\big\}\exp\big(\tfrac{1}{2}z^2\big), & \mbox{if } \varepsilon=1,\\
\big\{C_1 U\big(\nu+\tfrac{1}{2},\sqrt{2} z\big)+
C_2 U\big(\nu+\tfrac{1}{2},-\sqrt{2} z\big)\big\}\exp\big({-}\tfrac{1}{2}z^2\big), & \mbox{if } \varepsilon=-1,
\end{cases}\!\!\!\!
\end{gather}
where $C_{1},C_{2} \in \mathbb{C}$. We comment on the case $\nu \in \mathbb{Z}$ at the end of this section. For general $n \in \mathbb{N}$, it is a remarkable fact that the special function solutions can be constructed explicitly in terms of the following Wronskian determinants (see \cite{FW, Okamoto} and \cite[Theorem~3.5]{CJ}):
\begin{gather*}
\tau_{n,\nu}(z;\varepsilon)=
\mathcal{W}\left(\varphi_{\nu}(z;\varepsilon),\frac{{\rm d}\varphi_{\nu}(z;\varepsilon)}{{\rm d}z}, \ldots,
\frac{{\rm d}^{n-1}\varphi_{\nu}(z;\varepsilon)}{{\rm d}z^{n-1}}\right), \qquad n\geq 1,
\end{gather*}
with $\tau_{0,\nu}(z;\varepsilon)=1$. This Wronskian determinant can in turn be written in terms of the Hankel determinant $H_{n}(v,s,\alpha)$, we obtain for $n \geq 0$
\begin{gather}
H_n(v,s,\alpha) = \Gamma(1+\alpha)^n 2^{-n^2-\frac{n(\alpha-1)}{2}}\tau_{n,\alpha}(v,-1)\nonumber \\
\hphantom{H_n(v,s,\alpha)}{} = e^{-nv^2}\Gamma(1+\alpha)^n 2^{-n^2-\frac{n(\alpha-1)}{2}}\tau_{n,-\alpha-1}(v,1),\label{tau in terms of Hn}
\end{gather}
where $C_1=1$ and $C_2=s$. This relation relies on the integral representation for the parabolic cylinder function \cite[equation~(12.5.1)]{NIST}:
\begin{gather}\label{intformula_U}
U(a,z)=\frac{e^{-\tfrac{z^2}{4}}}{\Gamma(a+\frac{1}{2})}
\int_0^{\infty} x^{a-\frac{1}{2}} e^{-\frac{x^2}{2}-zx}{\rm d} x, \qquad \Re \, a>-\frac{1}{2}.
\end{gather}
By a change of variables, we can write the moment of order $0$ of the weight function $w(x) = w(x;v,s,\alpha)$ given by \eqref{wHermite} in terms of $U$ as follows:
\begin{gather*}
\int_{-\infty}^{\infty} w(x){\rm d}x = \displaystyle 2^{-\frac{1+\alpha}{2}}e^{-v^{2}} \left[ s \int_{0}^{\infty}x^{\alpha}e^{-\frac{x^{2}}{2}+\sqrt{2}vx}{\rm d}x + \int_{0}^{\infty}x^{\alpha}e^{-\frac{x^{2}}{2}-\sqrt{2}vx}{\rm d}x \right] \\
 \hphantom{\int_{-\infty}^{\infty} w(x){\rm d}x}{} = 2^{-\frac{1+\alpha}{2}}e^{-\frac{v^{2}}{2}}\Gamma(1+\alpha) \left[ s U \big( \alpha + \tfrac{1}{2},-\sqrt{2}v \big) + U \big( \alpha + \tfrac{1}{2},\sqrt{2}v \big) \right].
\end{gather*}
Therefore, from \eqref{phiU} with $C_1=1$ and $C_2=s$, we have
\begin{gather}\label{lol 42}
\int_{-\infty}^{\infty} w(x){\rm d}x = 2^{-\frac{1+\alpha}{2}}\Gamma(1+\alpha) \varphi_{\alpha}(v;-1) = 2^{-\frac{1+\alpha}{2}}e^{-v^2}\Gamma(\alpha+1)\varphi_{-\alpha-1}(v;1).
\end{gather}
Dif\/ferentiating $j$ times \eqref{lol 42} with respect to $v$, we get
\begin{gather*}
\partial_{v}^{j}\varphi_{\alpha}(v;-1) = \frac{2^{\frac{1+\alpha }{2}}}{\Gamma(1+\alpha)} \partial_{v}^{j} \left( \int_{-\infty}^{\infty} w(x){\rm d}x \right) = \frac{(-1)^{j}2^{j}2^{\frac{1+\alpha}{2}}}{\Gamma(1+\alpha)} \int_{-\infty}^{\infty} x^{j}w(x){\rm d}x, \\
\partial_{v}^{j}\varphi_{-\alpha-1}(v;1) = \frac{2^{\frac{1+\alpha }{2}}}{\Gamma(1+\alpha)} \partial_{v}^{j} \left( e^{v^{2}} \int_{-\infty}^{\infty} w(x){\rm d}x \right) = \frac{(-1)^{j}2^{j}2^{\frac{1+\alpha}{2}}}{\Gamma(1+\alpha)}e^{v^{2}} \int_{-\infty}^{\infty} (x-v)^{j}w(x){\rm d}x.
\end{gather*}
After taking the determinant, this establishes the formula \eqref{tau in terms of Hn}. Thus, the results presented in Section \ref{Section: main results} imply large $n$ asymptotics for special solutions of $\textrm{P}_{\rm IV}$ expressed in terms of $\tau_{n,\nu}$, uniformly for $C_{2}$ small.

As explained in \cite{CJ}, when the parameter $\nu=m$ is an integer, we need to take a dif\/ferent combination of independent solutions of equation \eqref{ODE_Weber}, since both parabolic cylinder functions in the seed function become Hermite polynomials. This leads to the so-called generalised Hermite polynomials, which are rational solutions of $\rm{P}_{\rm IV}$, obtained as particular cases of the special function solutions, see the parameter plane of $\rm{P}_{\rm IV}$ in \cite[Section~5.5]{Clarkson} or~\cite{GLS}. We refer the reader to the work of Kawijara and Ohta \cite[Def\/inition~3.1]{KO} for more information on the rational solutions of $\rm{P}_{\rm IV}$.

We also note that in the previous discussion, since $\nu=\alpha$, the integral representation \eqref{intformula_U} imposes the restriction $\alpha>-1$. This restriction can be lifted at the price of taking complex integration and a complex weight function. This def\/ines an associated family of orthogonal polynomials only formally. We do not pursue this route in this paper, but we note that it has been used in the literature, for example in the analysis of the asymptotic behavior and pole structure of rational solutions of $\rm{P}_{\rm II}$ by Bertola and Bothner \cite{BB} and more recently of rational solutions of $\rm{P}_{\rm IV}$ by Buckingham \cite{Buckingham}.

\section{Main results}\label{Section: main results}
Observe that if we naively take the limit $s \to 0$ in \eqref{lol 39}, the asymptotics on the right-hand side blow up, due to the presence of $\log s$ terms. Therefore, a critical transition is expected when $n \to \infty$ and simultaneously $s \to 0$ in a suitable double scaling limit. As mentioned earlier, the contribution of this paper is to analyse 1)~the case $s = 0$ (which is known only for $v = 0$, see~\eqref{lol 38}), and 2)~the transition between the situation when $s = 0$ and the situation when $s$ is in a compact subset of $(0,1]$, given by~\eqref{lol 39}. We obtain the following results.

\begin{Theorem}\label{Th1}
Let $\alpha \in (-1,\infty)$ and $t \in (-1,\infty)$. As $n\to\infty$, we have
\begin{gather}\label{asymp_Hn_s0}
\log \frac{H_n\big(\sqrt{2n}t,0,\alpha\big)}{H_n(0,0,\alpha)} = C_{1}(t)n^{2} + C_{2}(t,\alpha)n + C_{3}(t,\alpha) + \mathcal{O}\big(n^{-1}\big),
\end{gather}
where the coefficients are given by
\begin{gather}
C_{1}(t) = -\frac{2t^{3}}{27}\big(\sqrt{3+t^{2}}-t\big) - \left( \frac{4}{3}t^{2} + \frac{5}{9}t\sqrt{3+t^{2}} \right) - \log \left(\frac{t+\sqrt{3+t^{2}}}{\sqrt{3}} \right), \label{C_1} \\
C_{2}(t,\alpha) = \frac{\alpha t}{3} \big( t-\sqrt{3 + t^2}\big) - \alpha\log \left(\frac{t+\sqrt{3+t^{2}}}{\sqrt{3}}\right), \label{C_2} \\
C_{3}(t,\alpha) = \frac{1-3\alpha^{2}}{6}\log \left(\frac{t+\sqrt{3+t^{2}}}{\sqrt{3}} \right) -\frac{1}{48}\log \left( \frac{3+t^{2}}{3} \right) \nonumber\\
\hphantom{C_{3}(t,\alpha) =}{} - \frac{1}{16}\big(1-4\alpha^{2}\big)\log \left( \frac{3+5t^{2}+4t\sqrt{3+t^{2}}}{3} \right).\label{C_3}
\end{gather}
Furthermore, the error term $\mathcal{O}\big(n^{-1}\big)$ is uniform for $t$ in a compact subset of $(-1,\infty)$.
\end{Theorem}
Note that there is no critical transition in \eqref{asymp_Hn_s0} as $t \to 1$.
\begin{Remark}
The probability \eqref{prob x_min} with $v = \sqrt{2n}t$ can be rewritten as
\begin{gather*}
\mathbb{P}\big( x_{\min}^{(\sqrt{2n}t,\alpha)} \geq \sqrt{2n}t \big) = \frac{H_{n}(\sqrt{2n}t,0,\alpha)}{H_{n}(0,0,\alpha)}\frac{H_{n}(0,0,\alpha)}{H_{n}(0,1,\alpha)}\frac{H_{n}(0,1,\alpha)}{H_{n}(\sqrt{2n}t,1,\alpha)}.
\end{gather*}
Large $n$ asymptotics of these three ratios are given up to the constant term by Theorem~\ref{Th1} (for $t>-1$), \eqref{lol 38} and \eqref{lol 22} (for $t \in (-1,1)$), respectively. Putting these asymptotics together, we obtain as $n \to \infty$ and for $t \in (-1,1)$,
\begin{gather}
\log \mathbb{P}\big( x_{\min}^{(\sqrt{2n}t,\alpha)} \geq \sqrt{2n}t \big) = \left( C_{1}(t)-\frac{\log 3}{2} \right)n^{2} + \left( C_{2}(t,\alpha)-\frac{\alpha \log 3}{2}-\alpha t^{2} \right)n \nonumber\\
\qquad{}+ \left( \frac{\alpha^{2}}{4}-\frac{1}{12} \right) \log n + C_{3}(t,\alpha) + c_{0} - \frac{\alpha^{2}}{8}\log\big(1-t^{2}\big) + \mathcal{O}\left( \frac{\log n}{n} \right).\label{asymptotics for the probability xmin}
\end{gather}
It can be checked from \eqref{C_1} that $C_{1}(-1) = \frac{\log 3}{2}$ and
\begin{gather*}
C_{1}^{\prime}(t) = -\frac{8}{27}\big(\sqrt{3+t^{2}}-t\big)\big(3+5t^{2}+4t\sqrt{3+t^{2}}\big) < 0, \qquad \mbox{for} \quad t > -1,
\end{gather*}
which shows that the leading term in \eqref{asymptotics for the probability xmin} is negative. This implies that the above probability decays super exponentially fast as $n \to \infty$ for $t \in (-1,1)$.
\end{Remark}
To observe a transition in the large $n$ asymptotics of $H_{n}\big(\sqrt{2n}t,s,\alpha\big)$ when $s = 0$ and when~$s$ is in a compact subset of $(0,1]$, we couple the parameter $s$ with $n$ in the form
\begin{gather*}
s = e^{-\lambda n }, \qquad \lambda \geq 0.
\end{gather*}
Large $n$ asymptotics of $H_{n}\big(\sqrt{2n}t,e^{-\lambda n},\alpha\big)$ will depend on whether $\lambda$ is greater or smaller than a critical value $\lambda_{c}(t)$, which is explicit and given by
\begin{gather}
\lambda_{c}(t) = \frac{2t}{\sqrt{3}}\sqrt{3+t^{2}+2t\sqrt{3+t^{2}}}\nonumber\\
\hphantom{\lambda_{c}(t) =}{} + 2\log \left( 2+t^{2} + t \sqrt{3+t^{2}} + \frac{\sqrt{3+t^{2}}+t}{\sqrt{3}}\sqrt{3+t^{2}+2t\sqrt{3+t^{2}}} \right).\label{critical lambda}
\end{gather}
\begin{Theorem} \label{thm asymptotics for P with (t,lambda)}
Let $s = e^{-\lambda n}$ with $\lambda \in [0,\infty)$, we have the following asymptotic results.
\begin{itemize}\itemsep=0pt
\item[$(1)$] If $t \in (-1,1)$ and $\lambda \geq \lambda_{c}(t)$, then
\begin{gather} \label{lala10}
\log \frac{H_{n}\big(\sqrt{2n}t,e^{-\lambda n},\alpha\big)}{H_{n}\big(\sqrt{2n}t,0,\alpha\big)}=
\mathcal{O}\big(n^{-1/2}e^{-n(\lambda-\lambda_{c}(t))}\big), \qquad \mbox{as} \quad n \to \infty,
\end{gather}
and large $n$ asymptotics for $\log H_{n}\big(\sqrt{2n}t,0,\alpha\big)$ are given by Theorem~{\rm \ref{Th1}}. Furthermore, the $\mathcal{O}$ term in \eqref{lala10} is uniform for $t$ in a compact subset of $(-1,1)$ and for $\lambda \geq \lambda_{c}(t)$.
\item[$(2)$] If $t \in (-1,1)$ and $0 \leq \lambda \leq \lambda_{c}(t)$ are fixed, then
\begin{gather}\label{lol 40}
\lim_{n\to\infty} \frac{1}{n^{2}} \log \frac{H_{n}\big(\sqrt{2n}t,e^{-\lambda n},\alpha\big)}{H_{n}\big(\sqrt{2n}t,1,\alpha\big)} = - \int_{0}^{\lambda} \Omega\big(t,\tilde\lambda\big){\rm d}\tilde\lambda,
\end{gather}
where
\begin{gather*}
\Omega(t,\lambda) = \int_{a}^{b} \rho(x;t,\lambda){\rm d}x, \qquad \rho(x;t,\lambda) = \frac{2}{\pi} \sqrt{c-x} \sqrt{\frac{x-b}{x-t}} \sqrt{x-a},
\end{gather*}
and $a<b<t<c$, with $a$, $b$ and $c$ depending on $\lambda$ and $t$, are uniquely determined by the following equations:
\begin{gather*}
t = a+b+c, \\ 
2 = a^{2} + b^{2} + c^{2} - t^{2},\\ 
\lambda = 4 \int_{b}^{t} \frac{\sqrt{c-x}}{\sqrt{t-x}}\sqrt{x-b}\sqrt{x-a}{\rm d}x. 
\end{gather*}
\end{itemize}
\end{Theorem}

\begin{Remark}\label{remark: not case t geq 1}
In Theorem \ref{thm asymptotics for P with (t,lambda)}, we restrict ourselves to the case $t \in (-1,1)$. With increasing ef\/fort, this result can be extended for $t \in (-1,\infty)$. If $t \geq 1$, a new region appears in the~$(t,\lambda)$ plane which deserves a separate analysis (which we expect to be straightforward but long). Therefore, we decided not to proceed in this direction.
\end{Remark}
\begin{Remark}
Note that the denominators on the left hand sides of \eqref{lala10} and \eqref{lol 40} are dif\/ferent. We can use Theorem~\ref{thm asymptotics for P with (t,lambda)} to obtain information about the large deviation of the smallest thinned ${\rm GUE}\big(\sqrt{2n}t,\alpha\big)$ eigenvalue as follows.
By \eqref{prob x_min thinned} and \eqref{prob x_min} with $v = \sqrt{2n}t$ and $s = e^{-\lambda n}$, we have
\begin{gather*}
\mathbb{P}\big( y_{\min}^{(\sqrt{2n}t,e^{-\lambda n},\alpha)} \geq \sqrt{2n}t \big) = \frac{H_{n}\big(\sqrt{2n}t,e^{-\lambda n},\alpha\big)}{H_{n}\big(\sqrt{2n}t,0,\alpha\big)}\mathbb{P}\big( x_{\min}^{(\sqrt{2n}t,\alpha)} \geq \sqrt{2n}t \big).
\end{gather*}
Therefore, for $t \in (-1,1)$ and $\lambda \geq \lambda_{c}(t)$, by \eqref{lala10}, as $n \to \infty$ we have
\begin{gather}\label{lol 64}
\log \mathbb{P}\big( y_{\min}^{(\sqrt{2n}t,e^{-\lambda n},\alpha)} \geq \sqrt{2n}t \big) = \log \mathbb{P}\big( x_{\min}^{(\sqrt{2n}t,\alpha)} \geq \sqrt{2n}t \big) + \mathcal{O}\big(n^{-1/2}e^{-n(\lambda-\lambda_{c}(t))}\big),
\end{gather}
and large $n$ asymptotics of $\log \mathbb{P}\big( x_{\min}^{(\sqrt{2n}t,\alpha)} \geq \sqrt{2n}t \big)$ are given by \eqref{asymptotics for the probability xmin}. In the regime $t \in (-1,1)$ and $0 \leq \lambda \leq \lambda_{c}(t)$, \eqref{lol 40} implies
\begin{gather}\label{lol 65}
\log \mathbb{P}\big( y_{\min}^{(\sqrt{2n}t,e^{-\lambda n},\alpha)} \geq \sqrt{2n}t \big) = \left( - \int_{0}^{\lambda} \Omega(t,\tilde\lambda){\rm d}\tilde\lambda \right) n^{2} + o\big(n^{2}\big).
\end{gather}
Since \eqref{lol 64} and \eqref{lol 65} are both valid for $\lambda = \lambda_{c}(t)$, by equalling the leading term, we have
\begin{gather}\label{lol 52}
-\int_{0}^{\lambda_{c}(t)} \Omega(t,\lambda) {\rm d}\lambda = C_{1}(t)-\frac{\log 3}{2}.
\end{gather}
We will give an independent and more direct proof of this formula at the end of Section \ref{Sec_asympHn}.
\end{Remark}
\begin{Remark}
Note that the limit \eqref{lol 40} is independent of $\alpha$. The subleading terms in the large $n$ asymptotics of $\frac{H_{n}(\sqrt{2n}t,e^{-\lambda n},\alpha)}{H_{n}(\sqrt{2n}t,1,\alpha)}$ are expected to depend on $\alpha$ and to be oscillatory and described in terms of elliptic $\theta$-functions. These functions appear in our analysis (see, e.g.,~\eqref{Pinf}). This heuristic is also supported by the analogy of our situation with \cite{BotDeiItsKra}, where the authors obtained $\theta$-functions in the subleading terms.
\end{Remark}
\subsection*{Outline}
The orthogonal polynomials (OPs) with respect to the weight \eqref{wHermite} play a central role in our analysis. In Section \ref{sec_OPs}, we obtain identities for $\partial_{v}H_{n}(v,0,\alpha)$ and for $\partial_{s} H_{n}(v,s,\alpha)$ in terms of these OPs. The Riemann--Hilbert (RH) problem which characterizes these OPs is presented in Section \ref{sec_RH1}. We obtain large $n$ asymptotics for the OPs via a Deift--Zhou steepest descent method on this RH problem. The f\/irst steps of the steepest descent method are the same regardless of the value of the parameter $s \in [0,1)$, and are also presented in Section \ref{sec_RH1}. For the last steps, the analysis will then depend on the speed of convergence of $s$ to $0$. By writing $s = e^{-\lambda n}$, $\lambda \in (0,\infty]$, we distinguish two dif\/ferent regimes in $\lambda$ which are separated by the critical value $\lambda_c(t)>0$. We study the situation $\lambda\geq \lambda_c(t)$ in Section \ref{Section: RH analysis region 1} (the case $s = 0$ corresponds to the special case of $\lambda = +\infty$), and the situation $0<\lambda<\lambda_c(t)$ in Section \ref{Section: RH analysis region 2}. We integrate the dif\/ferential identities and prove Theorem \ref{Th1} and Theorem~\ref{thm asymptotics for P with (t,lambda)} in Section~\ref{Sec_asympHn}.

\section{Orthogonal polynomials and dif\/ferential identities}\label{sec_OPs}
\subsection{Orthogonal polynomials}\label{Subsection: OPs}
We consider the family of orthonormal polynomials $p_{j}$ of degree $j$ with respect to $w$ def\/ined in~\eqref{wHermite}, characterized by the orthogonality conditions
\begin{gather}\label{OP conditions for p}
\int_{\mathbb{R}} p_{j}(x)p_{k}(x)w(x){\rm d}x = \delta_{jk}, \qquad j, k = 0,1,2,\dots,
\end{gather}
and $\kappa_{j}>0$ is the leading coef\/f\/icient of $p_{j}$, that is
$\pi_{j}(x) = \kappa_{j}^{-1}p_{j}(x)$ is the monic orthogonal polynomial of degree $j$ which satisf\/ies
\begin{gather}\label{OP conditions for pi}
\int_{\mathbb{R}} \pi_{j}(x)\pi_{k}(x) w(x){\rm d}x = h_{j}\delta_{jk}, \qquad j, k = 0,1,2,\dots,
\end{gather}
where $h_{j}$ is the squared norm of $\pi_{j}$. From \eqref{OP conditions for p}, we have $h_{j} = \kappa_{j}^{-2}$. It is well-known (see, e.g.,~\cite{Szego-OP}) that these OPs satisfy the recurrence relation
\begin{gather}\label{TTRR_monic}
x\pi_{j}(x)=\pi_{j+1}(x) + \beta_{j} \pi_{j}(x)+\gamma_{j}^2 \pi_{j-1}(x), \qquad j \geq 0, \end{gather}
with $\pi_{-1}(x) := 0$. Note that if we write $\pi_{j}(x)=x^j+\sigma_{j}x^{j-1}+\cdots$, then from \eqref{TTRR_monic} we get the relation
\begin{gather}\label{sigma beta}
\sigma_{j}-\sigma_{j+1}=\beta_{j}, \qquad j \geq 0, \qquad \mbox{where} \qquad \sigma_{0} := 0.
\end{gather}
\subsection[Dif\/ferential identity in $v$ for $s = 0$]{Dif\/ferential identity in $\boldsymbol{v}$ for $\boldsymbol{s = 0}$}
From the determinantal representation for OPs (see, e.g., \cite{Szego-OP}) and \eqref{integral representation for H_n}, the Hankel determinant $H_{n}(v,0,\alpha)$ can be written in terms of the norms of the OPs, one has
\begin{gather}\label{product of the norms}
H_n(v,0,\alpha)=\prod_{j=0}^{n-1} h_j.
\end{gather}
If we dif\/ferentiate with respect to $v$ the relation \eqref{OP conditions for pi} with $k=j$, we obtain
\begin{gather*}
\partial_{v}h_{j} = \partial_{v} \left( \int_{v}^{\infty} \pi_{j}^{2}(x)w(x){\rm d}x \right) = \partial_{v} \left( \int_{0}^{\infty} \pi_j^{2}(x+v)w(x+v){\rm d}x \right).
\end{gather*}
Since $w(x+v) = x^{\alpha}e^{-(x+v)^{2}}$, we get
\begin{gather*}
\partial_{v}h_{j} = 2 \int_{0}^{\infty} \pi_{j}(x+v) \partial_{v}(\pi_{j}(x+v))w(x+v){\rm d}x - 2 \int_{0}^{\infty} (x+v) \pi_{j}^{2}(x+v) w(x+v){\rm d}x \\
 \hphantom{\partial_{v}h_{j}}{} = -2 \int_{v}^{\infty} x \pi_{j}^{2}(x)w(x){\rm d}x,
\end{gather*}
where we have used the orthogonality \eqref{OP conditions for pi} and the fact that $\partial_{v}(\pi_{j}(x+v))$ is a polynomial of degree at most $j-1$. From the recurrence relation \eqref{TTRR_monic}, we obtain
\begin{gather*}
\partial_{v}h_{j} = -2 \beta_{j} h_{j}.
\end{gather*}
As a consequence of this, by taking the $\log$ in \eqref{product of the norms} and dif\/ferentiating it with respect to $v$, we have
\begin{gather*}
\partial_{v}\log H_n(v,0,\alpha)=\sum_{j=0}^{n-1}\partial_{v}\log h_{j}=-2\sum_{j=0}^{n-1}\beta_{j},
\end{gather*}
This can be simplif\/ied by using \eqref{sigma beta}, and gives
\begin{gather}\label{diff identity for s=0}
\partial_{v}\log H_n(v,0,\alpha)=2\sigma_{n}(v),
\end{gather}
where $\sigma_n$ is the subleading coef\/f\/icient of the polynomial $\pi_n(x)$, def\/ined after (\ref{TTRR_monic}), and we have explicitly written the dependence of $\sigma_{n}$ on $v$.
\subsection[Dif\/ferential identity in $s$]{Dif\/ferential identity in $\boldsymbol{s}$}
Suppose that the thinned eigenvalues $y_{1},\dots,y_{m}$ are observed and that $\sharp\{ y_{i} \colon y_{i} <v \} = 0$. From Bayes' formula, using \eqref{lol 23} and \eqref{prob x_min thinned}, the distribution of the whole spectrum $x_{1},\dots,x_{n}$ conditionally on this event is given by
\begin{gather*}
\frac{1}{n!H_{n}(v,s,\alpha)} \Delta(x)^{2} \prod_{i=1}^{n} w(x_{i};v,s,\alpha){\rm d}x_{i}.
\end{gather*}
Such point processes are called conditional, and were f\/irst considered in \cite{ChCl2} on the unit circle and then on the real line in~\cite{Charlier}. This point process is determinantal~\cite{Deift}, and its correlation kernel is given by
\begin{gather}\label{correlation kernel 1}
K_{n}(x,y) = \begin{cases}
\displaystyle \sqrt{w(x)w(y)} \frac{\kappa_{n-1}}{ \kappa_{n}}\frac{ p_{n-1}(y) p_{n}(x)- p_{n-1}(x) p_{n}(y)}{x-y}, & \mbox{if } x \neq y, \\
\displaystyle w(x)\frac{\kappa_{n-1}}{ \kappa_{n}} \big(p_{n}^{\prime}(x)p_{n-1}(x)- p_{n}(x) p_{n-1}^{\prime}(x)\big), & \mbox{if } x = y,
\end{cases}
\end{gather}
where the OPs $p_{j}$ are orthonormal with respect to $w(x;v,s,\alpha)$ and are def\/ined in~\eqref{OP conditions for p}. The expected number of points on $(-\infty,v)$ in this point process is denoted by $\mathcal{E}_{n}(v,s,\alpha)$. It is also known~\cite{Deift} that $\mathcal{E}_{n}(v,s,\alpha)$ can be expressed in terms of the one-point correlation function $K_{n}(x,x)$, we have
\begin{gather}\label{E in terms of Kn}
\mathcal{E}_{n}(v,s,\alpha) = \int_{-\infty}^{v} K_{n}(x,x){\rm d}x.
\end{gather}
The quantity $\mathcal{E}_{n}(v,s,\alpha)$ can also be expressed in terms of the logarithmic derivative of $H_{n}(v,s,\alpha)$ with respect to $s$. Consider the following partition of $\mathbb{R}^{n}$:
\begin{gather*}
A_{k} = \big\{ (x_{1},\dots,x_{n}) \in \mathbb{R}^{n} \colon \sharp \{ x_{i} \colon x_{i} < v \} = k \big\}, \qquad \bigsqcup_{k=0}^{n} A_{k} = \mathbb{R}^{n}.
\end{gather*}
By def\/inition of $\mathcal{E}_{n}(v,s,\alpha)$ we have
\begin{gather*}
 \mathcal{E}_{n}(v,s,\alpha) = \sum_{k=0}^{n} \frac{k}{n!H_{n}(v,s,\alpha)} \int_{A_{k}}\Delta(x)^{2} \prod_{i=1}^{n} w(x_{i};v,s,\alpha){\rm d}x_{i} \\
 \hphantom{\mathcal{E}_{n}(v,s,\alpha)}{} = \sum_{k=0}^{n} \frac{k s^{k}}{n!H_{n}(v,s,\alpha)} \int_{A_{k}}\Delta(x)^{2} \prod_{i=1}^{n} |x_{i}-v|^{\alpha}e^{-x_{i}^{2}}{\rm d}x_{i}.
\end{gather*}
Note that the $n$-fold integral \eqref{integral representation for H_n} can be rewritten as
\begin{gather*}
H_{n}(v,s,\alpha) = \sum_{k=0}^{n} \frac{ s^{k}}{n!} \int_{A_{k}}\Delta(x)^{2} \prod_{i=1}^{n} |x_{i}-v|^{\alpha}e^{-x_{i}^{2}}{\rm d}x_{i},
\end{gather*}
and thus we have $\mathcal{E}_{n}(v,s,\alpha) = s \partial_{s}\log H_{n}(v,s,\alpha)$. Putting this together with \eqref{E in terms of Kn}, we obtain the dif\/ferential identity
\begin{gather}\label{diff identity 1}
s \partial_{s}\log H_{n}(v,s,\alpha) = \int_{-\infty}^{v} K_{n}(x,x){\rm d}x.
\end{gather}
We will also use later the well-known (see, e.g., \cite{Deift}) formula for reproducing kernels
\begin{gather}\label{lol 20}
\int_{-\infty}^{\infty} K_{n}(x,x){\rm d}x = n.
\end{gather}

\section[A Riemann--Hilbert problem and renormalization of the problem]{A Riemann--Hilbert problem and renormalization\\ of the problem}\label{sec_RH1}

We will perform the Deift--Zhou \cite{DeiftZhou1992, DeiftZhou} steepest descent method on a Riemann--Hilbert problem to get the large $n$ asymptotics for $p_{n}$. Consider the matrix valued function $Y$, def\/ined by
\begin{gather}\label{Y definition}
Y(z) = \begin{pmatrix}
\kappa_{n}^{-1}p_{n}(z) & \displaystyle \frac{\kappa_{n}^{-1}}{2\pi i} \int_{\mathbb{R}} \frac{p_{n}(x)w(x)}{x-z}{\rm d}x \vspace{1mm}\\
-2\pi i \kappa_{n-1} p_{n-1}(z) & \displaystyle -\kappa_{n-1} \int_{\mathbb{R}} \frac{p_{n-1}(x)w(x)}{x-z}{\rm d}x
\end{pmatrix}.
\end{gather}
It is well-known \cite{FokasItsKitaev} that $Y$ is the unique solution of the following RH problem.

\subsubsection*{RH problem for $\boldsymbol{Y}$}
\begin{itemize}\itemsep=0pt
\item[(a)] $Y \colon \mathbb{C}\setminus \mathbb{R} \to \mathbb{C}^{2\times 2}$ is analytic.
\item[(b)] The limits of $Y(x\pm i \epsilon)$ as $\epsilon>0$ approaches $0$ exist, are continuous on $\mathbb{R}\setminus \{v\}$ and are denoted by $Y_+$ and $Y_-$ respectively. Furthermore they are related by
\begin{gather}\label{lol 51}
Y_{+}(x) = Y_{-}(x) \begin{pmatrix}
1 & w(x) \\ 0 & 1
\end{pmatrix}, \qquad \mbox{for} \quad x \in \mathbb{R}\setminus \{v\}.
\end{gather}
\item[(c)] As $z \to \infty$, we have $Y(z) = \big(I + Y_{1} z^{-1} + \mathcal{O}\big(z^{-2}\big)\big) z^{n\sigma_{3}}$, where $\sigma_{3} = \left(\begin{smallmatrix}
1 & 0 \\ 0 & -1
\end{smallmatrix}\right)$.
\item[(d)] As $z$ tends to $v$, the behaviour of $Y$ is
\begin{gather*}
\displaystyle Y(z) = \begin{pmatrix}
\mathcal{O}(1) & \mathcal{O}(\log (z-v)) \\
\mathcal{O}(1) & \mathcal{O}(\log (z-v))
\end{pmatrix}, \qquad \mbox{if} \ \ \alpha = 0, \\
\displaystyle Y(z) = \begin{pmatrix}
\mathcal{O}(1) & \mathcal{O}(1)+\mathcal{O}((z-v)^{\alpha}) \\
\mathcal{O}(1) & \mathcal{O}(1)+\mathcal{O}((z-v)^{\alpha})
\end{pmatrix}, \qquad \mbox{if} \ \ \alpha \neq 0.
\end{gather*}
\end{itemize}
If $s = 0$, from condition (b) $Y$ has no jump along $(-\infty,v)$ and thus $Y$ is analytic in $\mathbb{C}\setminus [v,\infty)$. Note also that $Y_{11}(z) = \kappa_{n}^{-1}p_{n}(z) = \pi_{n}(z)$, and thus
\begin{gather}\label{lol 6}
Y_{1,11} = \sigma_{n}(v),
\end{gather}
where $Y_{1,11}$ denotes the $(1,1)$ entry of the matrix $Y_1$.
\subsection{Normalization of the RH problem}
We def\/ine $t = \frac{v}{\sqrt{2n}}$, and we normalize the RH problem for $Y$ with the following transformation
\begin{gather}\label{Y to U transformation}
U(z) = (2n)^{-(\frac{\alpha}{4}+\frac{n}{2})\sigma_{3}} Y\big(\sqrt{2n}z\big)(2n)^{\frac{\alpha}{4}\sigma_{3}}.
\end{gather}
The matrix $U$ satisf\/ies the following RH problem.
\subsubsection*{RH problem for $\boldsymbol{U}$}
\begin{itemize}\itemsep=0pt
\item[(a)] $U \colon \mathbb{C}\setminus \mathbb{R} \to \mathbb{C}^{2\times 2}$ is analytic.
\item[(b)] $U$ has the following jumps:
\begin{gather*}
U_{+}(x) = U_{-}(x) \begin{pmatrix}
1 & \widetilde w(x) \\ 0 & 1
\end{pmatrix}, \qquad \mbox{for} \quad x \in \mathbb{R}\setminus \{t\},
\end{gather*}
where
\begin{gather*}
\widetilde w(x) = (2n)^{-\frac{\alpha}{2}}w(\sqrt{2n}x)= |x-t|^{\alpha} e^{-2 n x^{2}} \begin{cases}
s, & \mbox{if} \ x < t, \\
1, & \mbox{if} \ x > t.
\end{cases}
\end{gather*}
\item[(c)] As $z \to \infty$, we have $U(z) = \big(I + U_{1} z^{-1} + \mathcal{O}\big(z^{-2}\big)\big) z^{n\sigma_{3}}$.
\item[(d)] As $z$ tends to $t$, the behaviour of $U$ is
\begin{gather*}
\displaystyle U(z) = \begin{pmatrix}
\mathcal{O}(1) & \mathcal{O}(\log (z-t)) \\
\mathcal{O}(1) & \mathcal{O}(\log (z-t))
\end{pmatrix}, \qquad \mbox{if} \quad \alpha = 0, \\
\displaystyle U(z) = \begin{pmatrix}
\mathcal{O}(1) & \mathcal{O}(1)+\mathcal{O}((z-t)^{\alpha}) \\
\mathcal{O}(1) & \mathcal{O}(1)+\mathcal{O}((z-t)^{\alpha})
\end{pmatrix}, \qquad \mbox{if} \quad \alpha \neq 0.
\end{gather*}
\end{itemize}
The following lemma translates the dif\/ferential identities \eqref{diff identity for s=0} and \eqref{diff identity 1} in terms of $U$.
\begin{Lemma}
We have the following differential identities
\begin{gather}
\partial_{t} \log H_{n}\big(\sqrt{2n}t,0,\alpha\big) = 4n U_{1,11}, \label{lol 36} \\
s\partial_{s} \log H_{n}\big(\sqrt{2n}t,s,\alpha\big) = \int_{-\infty}^{t}\frac{\widetilde{w}(x)}{2\pi i}\big[ U^{-1}(x)U^{\prime}(x) \big]_{21}{\rm d}x. \label{diff identity 2}
\end{gather}
\end{Lemma}
\begin{proof}
The dif\/ferential identity \eqref{lol 36} is obtained by substituting \eqref{Y to U transformation} and \eqref{lol 6} into \eqref{diff identity for s=0}. Similarly, using \eqref{correlation kernel 1} and \eqref{Y definition}, the dif\/ferential identity \eqref{diff identity 1} can be rewritten as \begin{gather*}
s\partial_{s} \log H_{n}\big(\sqrt{2n}t,s,\alpha\big) = \int_{-\infty}^{\sqrt{2n}t} \frac{w(x)}{2\pi i}\big[ Y^{-1}(x)Y^{\prime}(x) \big]_{21}{\rm d}x,
\end{gather*}
which gives \eqref{diff identity 2} after using \eqref{Y to U transformation} and a change of variables.
\end{proof}

\begin{Remark}\label{remark: notation + and - not needed}
Note that $\left[ Y^{-1}(z)Y^{\prime}(z) \right]_{21}$ only involves the f\/irst column of $Y$, which is entire (see \eqref{Y definition} or equivalently \eqref{lol 51}). Thus $\left[ Y_{+}^{-1}(x)Y_{+}^{\prime}(x) \right]_{21}=\left[ Y_{-}^{-1}(x)Y_{-}^{\prime}(x) \right]_{21}$ for $x \in \mathbb{R}$, and we simply denote it by $\left[ Y^{-1}(x)Y^{\prime}(x) \right]_{21}$ without ambiguity. The same remark holds for $U$.
\end{Remark}

\subsection{Equilibrium measure}
We introduce a new parameter $\lambda \in [0,+\infty]$, def\/ined through $s = e^{-\lambda n}$, which characterizes the speed of convergence of $s$ to $0$ as $n \to \infty$. An essential tool in the RH analysis is the so-called equilibrium measure. In our case, the equilibrium measure $\mu_{V}$ is the unique minimizer of the functional
\begin{gather*}
\iint_{\mathbb{R}^{2}} \log |x-y|^{-1} {\rm d}\mu(x) {\rm d}\mu(y) + \int_{\mathbb{R}} V(x) {\rm d}\mu(x),
\end{gather*}
among all Borel probability measures $\mu$ on $\mathbb{R}$, where the potential $V$ is def\/ined by
\begin{gather*}
V(x) = \begin{cases}
2x^{2} + \lambda, & \mbox{if} \ x < t, \\
2x^{2}, & \mbox{if} \ x \geq t,
\end{cases}
\end{gather*}
and where the parameter $t$ has been def\/ined above \eqref{Y to U transformation}. The equilibrium measure is absolutely continuous with respect to the Lebesgue measure and its density will be denoted by~$\rho(x)$. The equilibrium measure and its support, denoted by~$\mathcal{S}$, are completely determined by the following Euler--Lagrange variational conditions \cite{SaTo}:
\begin{gather}
2 \int_{\mathcal{S}} \log |x-y| \rho(y){\rm d}y = V(x) - \ell, \qquad \mbox{for} \quad x \in \mathcal{S}, \label{var equality} \\
2 \int_{\mathcal{S}} \log |x-y| \rho(y){\rm d}y \leq V(x) - \ell, \qquad \mbox{for} \quad x \in \mathbb{R}\setminus \mathcal{S}, \label{var inequality}
\end{gather}
where $\ell$ is a constant. Proposition~\ref{prop: equilibrium measure} below shows that the equilibrium measure depends crucially on whether $\lambda \geq \lambda_{c}(t)$ or $0< \lambda < \lambda_{c}(t)$. If $\lambda = 0$, the potential is simply $V(x) = 2x^{2}$ and the equilibrium measure is the semicircle law supported on $(-1,1)$ (see, e.g.,~\cite{SaTo}). For convenience we also include it in Proposition~\ref{prop: equilibrium measure} (see case (3)), but without giving a proof of it.
\begin{Proposition}\label{prop: equilibrium measure}\quad
\begin{itemize}\itemsep=0pt
\item[$(1)$] If $t \in (-1,\infty)$ and $\lambda \geq \lambda_{c}(t)$, the density of the equilibrium measure $\rho(x) = \rho(x;t)$ is independent of $\lambda$ and is given by
\begin{gather}\label{lol 35}
\rho(x) = \frac{2}{\pi} ( x - \overline{b}) \frac{\sqrt{\overline{c}-x}}{\sqrt{x-t}},
\end{gather}
supported on $\mathcal{S} = [t, \overline{c}]$, with
\begin{gather*}
\overline{b} = \overline{b}(t) = \frac{t - \sqrt{3+t^{2}}}{3}, \qquad \overline{c} = \overline{c}(t) = \frac{t + 2\sqrt{3+t^{2}}}{3}.
\end{gather*}
The constant $\ell = \ell(t)$ in the variational conditions \eqref{var equality} and \eqref{var inequality} is given by
\begin{gather}\label{EL constant in region 1}
\ell = 1 + \frac{2}{3} t \big(\sqrt{3+t^{2}}+2t\big)+2\log \big( 2\big(t+\sqrt{3+t^{2}}\big) \big).
\end{gather}
Furthermore, the variational inequality \eqref{var inequality} is strict for all $x \in \mathbb{R}\setminus \mathcal{S}$ if $\lambda > \lambda_{c}$, and if $\lambda = \lambda_{c}$ then \eqref{var inequality} is strict for all $x \in \mathbb{R}\setminus (\mathcal{S}\cup \{ \overline{b} \})$ and \eqref{var inequality} is an equality at $x = \overline{b}$.
\item[$(2)$] If $t \in (-1,1)$ and $0< \lambda <\lambda_{c}(t)$, the density of the equilibrium measure $\rho(x) = \rho(x;t,\lambda)$ is given by
\begin{gather}\label{lol 45}
\rho(x) = \frac{2}{\pi} \sqrt{c-x} \sqrt{\frac{x-b}{x-t}} \sqrt{x-a},
\end{gather}
supported on two disjoint intervals
\begin{gather*}
\mathcal{S} = [a,b] \cup [t,c], \qquad a < b < t < c,
\end{gather*}
and $a$, $b$, $c$, depending on $\lambda$ and $t$, are uniquely determined by the following equations:
\begin{gather}
 t = a+b+c,\label{abc 1} \\
 2 = a^{2} + b^{2} + c^{2} - t^{2}, \label{abc 2} \\
 \lambda = 4 \int_{b}^{t} \frac{\sqrt{c-x}}{\sqrt{t-x}}\sqrt{x-b}\sqrt{x-a} {\rm d}x. \label{abc lambda}
\end{gather}
Furthermore, for a fixed $t$, the function $\lambda \mapsto b = b(a(\lambda),c(\lambda),\lambda)$ given by the system \eqref{abc 1}--\eqref{abc lambda} is strictly decreasing from $\lambda \in (0,\lambda_{c}(t))$ to $b \in (\overline{b},t)$. The constant $\ell = \ell(t,\lambda)$ is given by
\begin{gather}
\ell = -2 \int_{S} \log|x-t| \rho(x){\rm d}x + 2t^{2} \nonumber\\
\hphantom{\ell}{} = -2 \log |c| + 2c^{2} + \int_{c}^{\infty}\left( 4x - \frac{2}{x} - 4 \frac{\sqrt{x-c}}{\sqrt{x-t}}\sqrt{x-b}\sqrt{x-a} \right){\rm d}x,\label{Euler-Lagrange constant}
\end{gather}
and \eqref{var inequality} is strict.
\item[$(3)$] If $t\in(-1,1)$ and $\lambda = 0$, then we have the semi-circle law
\begin{gather}\label{sc law}
\rho(x) = \frac{2}{\pi} \sqrt{1-x^{2}}, \qquad \mathcal{S} = [-1,1], \qquad \ell = 1 + \log 4,
\end{gather}
and \eqref{var inequality} is strict for $x \in \mathbb{R}\setminus \mathcal{S}$.
\end{itemize}
\end{Proposition}

\begin{proof} We will start by proving case (2). From \eqref{abc 1} and \eqref{abc 2}, we can express $a$ and $c$ in terms of $b$ and $t$ as follows
\begin{gather}\label{lol 1}
a = -\frac{b}{2} + \frac{t}{2}-\frac{\sqrt{4-3b^{2}+2t b + t^{2}}}{2}, \qquad c = -\frac{b}{2} + \frac{t}{2}+\frac{\sqrt{4-3b^{2}+2t b + t^{2}}}{2}.
\end{gather}
For $t \in (-1,1)$ f\/ixed and $b \in (\overline{b},t)$, a direct check from \eqref{lol 1} shows that $a<b<t<c$, $\partial_{b} c <0$ and $-1 < \partial_{b} a$. This implies from \eqref{abc lambda} that
\begin{gather*}
\partial_{b} \lambda = 4 \int_{b}^{t} \frac{\sqrt{c-x}}{\sqrt{t-x}}\sqrt{x-b}\sqrt{x-a} \left( \frac{\partial_{b}c}{2(c-x)} - \frac{1}{2(x-b)} - \frac{\partial_{b}a}{2(x-a)} \right){\rm d}x < 0.
\end{gather*}
If $b \nearrow t$, equation \eqref{abc lambda} implies $\lambda \to 0$. On the other hand, if $b \searrow \overline{b}$, equations \eqref{abc 1} and \eqref{abc 2} imply $a \to \overline{b}$ and $c \to \overline{c}$. Again from \eqref{abc lambda}, we thus have
\begin{gather}\label{lol 3}
\lambda \to 4 \int_{\overline{b}}^{t} \frac{\sqrt{\overline{c}-x}}{\sqrt{t-x}}(x-\overline{b}){\rm d}x = \lambda_{c}(t), \qquad \mbox{as} \quad b \to \overline{b},
\end{gather}
where $\lambda_{c}(t)$ is given by \eqref{critical lambda}. This proves that the function $\lambda \mapsto b(a(\lambda),c(\lambda),\lambda)$ is a decreasing bijection from $\lambda \in (0,\lambda_{c}(t))$ to $b \in (\overline{b},t)$. In particular, given $t \in (-1,1)$ and $0 < \lambda < \lambda_{c}(t)$, $a$, $b$ and $c$ are uniquely determined by the equations \eqref{abc 1}--\eqref{abc lambda}. Equations~\eqref{abc 1} and~\eqref{abc 2} imply also that~$\rho$ is a density. Indeed, with a contour deformation and a residue calculation at~$\infty$, we obtain
\begin{gather*}
\int_{\mathcal{S}} \rho(x){\rm d}x = \frac{a^{2}+b^{2}+c^{2}-2ab-2ac-2bc + 2(a+b+c)t-3t^{2}}{4} = 1.
\end{gather*}
Now, we def\/ine
\begin{gather*}\label{f in proof}
f(x) = 2 \int_{\mathcal{S}} \log |x-y| \rho(y) {\rm d}y - 2 x^{2},
\end{gather*}
where $\rho$ is given by~\eqref{lol 45}. Its derivative $f^{\prime}(x)$ can be explicitly evaluated. Consider the function $\tilde{\rho}(z) = \frac{2}{\pi} \sqrt{z-c} \frac{\sqrt{z-b}}{\sqrt{z-t}}\sqrt{z-a}$, such that $\tilde{\rho}_{\pm}(x) = \pm i \rho(x)$ for $x \in \mathcal{S}$. From \eqref{f in proof}, we have that
\begin{gather*}
f^{\prime}(x) = \frac{1}{2\pi i} \int_{\Sigma} \frac{-2\pi \tilde{\rho}(w)}{x-w}{\rm d}w - 4x,
\end{gather*}
where $\Sigma$ consists of two circles surrounding $\mathcal{S}$ in the counter-clockwise direction, and if $x \notin \mathcal{S}$, then $x$ does not lie in the interior region of $\Sigma$. Also, note that
\begin{gather*}
\operatorname{Res}\left(\frac{-2\pi \tilde{\rho}(w)}{x-w},w = \infty\right) = -2(a+b+c-t-2x) = 4x,
\end{gather*}
where we have used \eqref{abc 1}. Therefore, $f^{\prime}(x) = 0$ for $x \in \mathcal{S}$ and from a contour deformation and a residue calculation, we obtain
\begin{gather}\label{f prime in proof}
f^{\prime}(x) = \begin{cases}
\displaystyle 4 \frac{\sqrt{c-x}}{\sqrt{t-x}}\sqrt{b-x}\sqrt{a-x}, & x < a, \\
 0, & a \leq x \leq b, \\
\displaystyle -4 \frac{\sqrt{c-x}}{\sqrt{t-x}}\sqrt{x-b}\sqrt{x-a}, & b < x < t, \\
 0, & t \leq x \leq c, \\
\displaystyle -4 \frac{\sqrt{x-c}}{\sqrt{x-t}}\sqrt{x-b}\sqrt{x-a}, & c < x.
\end{cases}
\end{gather}
Note furthermore that from \eqref{var equality}, we have
\begin{gather*}
\lambda = - \int_{b}^{t} f^{\prime}(x){\rm d}x = - (f(t)-f(b)).
\end{gather*}
This implies that \eqref{var equality} and \eqref{var inequality} are satisf\/ied, with a strict inequality in \eqref{var inequality}. From \eqref{var equality}, we have $-\ell = f(x)$, for any $x \in [t,c]$. In particular we have
\begin{gather}\label{EL constant integral form}
\ell = -f(t) = -2 \int_{\mathcal{S}} \log|x-t| \rho(x) {\rm d}x + 2t^{2},
\end{gather}
which is \eqref{Euler-Lagrange constant}. To prove the other expression for $\ell$ in \eqref{Euler-Lagrange constant}, we note that
\begin{gather}\label{f to f tilde}
f(x) = 2\log|x| - 2x^{2} + \tilde f(x),
\end{gather}
where $\tilde f(x) = 2 \int_{\mathcal{S}}\log \left| 1-\frac{y}{x} \right| \rho(y){\rm d}y$. Using the fact that $\lim\limits_{x\to\infty} \tilde f(x) = 0$ and the equa\-tions~\eqref{f prime in proof} and~\eqref{f to f tilde}, we obtain
\begin{gather*}
0 = \tilde f(c) + \int_{c}^{\infty} \tilde{f}^{\prime}(x){\rm d}x = -\ell - 2 \log |c| + 2c^{2} + \int_{c}^{\infty} \left( 4x- \frac{2}{x}+f^{\prime}(x) \right){\rm d}x,
\end{gather*}
from which we f\/ind the second expression in \eqref{Euler-Lagrange constant}. This f\/inishes the proof of case (2).

The case (1) is similar and simpler. In this case, we def\/ine a function $f$ as in \eqref{f in proof}, where $\rho(x) = \frac{2}{\pi} ( x - \overline{b}) \frac{\sqrt{\overline{c}-x}}{\sqrt{x-t}}$. The derivative of $f$ is equal to
\begin{gather}\label{f prime in proof 2}
f^{\prime}(x) = \begin{cases}
\displaystyle -4 \frac{\sqrt{\overline{c}-x}}{\sqrt{t-x}}(x-\overline{b}), & x < t, \\
 0, & t \leq x \leq \overline{c}, \\
\displaystyle -4 \frac{\sqrt{x-\overline{c}}}{\sqrt{x-t}}(x-\overline{b}), & \overline{c} < x.
\end{cases}
\end{gather}
The Euler--Lagrange constant can also be written as in \eqref{EL constant integral form}, but in case~(1) the integral can be explicitly evaluated with a primitive and gives \eqref{EL constant in region 1}. Since $\lambda \geq \lambda_{c}(t)$, from~\eqref{f prime in proof 2} and by the formula for $\lambda_{c}(t)$ given by \eqref{lol 3}, \eqref{var inequality} is strictly satisf\/ied in case~(1) except at $x = \overline{b}$ if $\lambda = \lambda_{c}$. This f\/inishes the proof of Proposition~\ref{prop: equilibrium measure}.
\end{proof}

\begin{Remark}
Here we just comment brief\/ly on what happens for $t \geq 1$ in parts (2) and (3) of Proposition~\ref{prop: equilibrium measure}, even though we will only focus on $t \in (-1,1)$ in the present paper, as mentioned in Remark~\ref{remark: not case t geq 1}. If $t>1$, a critical situation occurs if the endpoints $t$ and $c$ merge together. From equations~\eqref{abc 1} and~\eqref{abc 2}, in this case we have $a = -1$ and $b = 1$. From~\eqref{abc lambda}, this corresponds to
\begin{gather*}
\lambda = 4 \int_{1}^{t} \sqrt{x^{2}-1}{\rm d}x = 2t \sqrt{t^{2}-1}-2\log\big(t + \sqrt{t^{2}-1}\big) =:\epsilon(t).
\end{gather*}
With increasing ef\/fort (by adapting the above proof), it is possible to show that part (2) of Proposition~\ref{prop: equilibrium measure} remains valid for $t \geq 1$, provided that $\lambda \in (\epsilon(t),\lambda_{c}(t))$, and that part (3) remains valid for $t \geq 1$, provided that $\lambda \in [0,\epsilon(t)]$.
\end{Remark}
\begin{Remark}\label{remark : pointwise convergence}
We will prove part 2 of Theorem \ref{thm asymptotics for P with (t,lambda)} (that is, \eqref{lol 40}) by using the dif\/ferential identity \eqref{diff identity 2} (after the change of variables $s = e^{-\lambda n}$). We a priori need large $n$ asymptotics of $\partial_{\lambda}\log H_{n}\big(\sqrt{2n}t,e^{-\lambda n},\alpha\big)$ uniformly in $\lambda \in [0,\lambda_{c}(t)]$. As it can be seen from Proposition \ref{prop: equilibrium measure}, in this case the support of the equilibrium measure consists of two intervals. The region where $\lambda_{c}(t)-\lambda >0$ is small corresponds to the ``birth of a cut", which was studied in \cite{BerLee, Claeys, Mo}, and when $\lambda_{c}(t)-\lambda >0$ becomes larger, $\theta$-functions appear in the analysis. It is a technical task to obtain uniform asymptotics in these regions as $\lambda$ approaches $\lambda_c(t)$. Nevertheless, follo\-wing~\cite{ChCl2}, we will only need pointwise convergence in $\lambda \in (0,\lambda_{c}(t))$ and apply Lebesgue's dominated convergence theorem.
\end{Remark}

\subsection[First transformation: $U \mapsto T$]{First transformation: $\boldsymbol{U \mapsto T}$}
The f\/irst step of the steepest descent analysis consists of normalizing the RH problem at $\infty$, which can be done by using a so-called $g$-function. We def\/ine it by
\begin{gather}\label{g function def}
g(z) = \int_{\mathcal{S}} \log (z-y) \rho(y) {\rm d}y,
\end{gather}
where the principal branch is chosen for the logarithm. The density $\rho$ and its support $\mathcal{S}$ are def\/ined in Proposition~\ref{prop: equilibrium measure}. The $g$-function is analytic in $\mathbb{C}\setminus (-\infty,\sup \mathcal{S}]$ and possesses the following properties
\begin{gather}
g_{+}(x) + g_{-}(x) = 2 \int_{S} \log |x-y| \rho(y){\rm d}y, \qquad x \in \mathbb{R}, \label{g+ + g-} \\
g_{+}(x) - g_{-}(x) = 2 \pi i, \qquad x < \inf \mathcal{S}, \label{g+ - g- 1} \\
g_{+}(x) - g_{-}(x) = 2 \pi i \int_{x}^{\sup \mathcal{S}} \rho(y){\rm d}y, \qquad x \in [\inf \mathcal{S},\sup \mathcal{S}], \label{g+ - g- 2}\\
g_{+}(x) - g_{-}(x) = 0, \qquad \sup \mathcal{S} < x, \label{g+ - g- 3}
\end{gather}
Furthermore, by expanding the $g$-function as $z\to\infty$ in \eqref{g function def}, we have
\begin{gather}
e^{ng(z)} = z^{n}\left(1-\frac{n}{z} \int_{\mathcal{S}}x\rho(x){\rm d}x + \mathcal{O}\big(z^{-2}\big)\right), \qquad z\to\infty.\label{asymptotics g}
\end{gather}
We apply a f\/irst transformation on $U$ by
\begin{gather*}
T(z) = e^{\frac{n\ell}{2}\sigma_{3}}U(z)e^{-ng(z)\sigma_{3}}e^{-\frac{n\ell}{2}\sigma_{3}}.
\end{gather*}
$T$ satisf\/ies the following RH problem.
\subsubsection*{RH problem for $\boldsymbol{T}$}
\begin{itemize}\itemsep=0pt
\item[(a)] $T \colon \mathbb{C}\setminus \mathbb{R} \to \mathbb{C}^{2\times 2}$ is analytic.
\item[(b)] $T$ has the following jumps:
\begin{gather*}
T_{+}(x) = T_{-}(x) J_{T}(x), \qquad \mbox{for} \quad x \in \mathbb{R}\setminus \{t\},
\end{gather*}
where
\begin{gather*}
 J_{T}(x) = \begin{pmatrix}
e^{-n(g_{+}(x)-g_{-}(x))} & |x-t|^{\alpha}e^{n(g_{+}(x)+g_{-}(x)+\ell - V(x))} \\
0 & e^{n(g_{+}(x)-g_{-}(x))}
\end{pmatrix}, \qquad \mbox{for} \quad x \in \mathbb{R}\setminus\{t\}.
\end{gather*}
\item[(c)] As $z \to \infty$, we have $T(z) = I + \mathcal{O}\big(z^{-1}\big)$.
\item[(d)] As $z$ tends to $t$, the behaviour of $T$ is
\begin{gather*}
 T(z) = \begin{pmatrix}
\mathcal{O}(1) & \mathcal{O}(\log (z-t)) \\
\mathcal{O}(1) & \mathcal{O}(\log (z-t))
\end{pmatrix}, \qquad \mbox{if} \quad \alpha = 0, \\
 T(z) = \begin{pmatrix}
\mathcal{O}(1) & \mathcal{O}(1)+\mathcal{O}((z-t)^{\alpha}) \\
\mathcal{O}(1) & \mathcal{O}(1)+\mathcal{O}((z-t)^{\alpha})
\end{pmatrix}, \qquad \mbox{if} \quad \alpha \neq 0.
\end{gather*}
\end{itemize}
From now on, we will separate the analysis into two parts, depending on whether $\lambda \geq \lambda_{c}(t)$ or $0<\lambda < \lambda_{c}(t)$.

\section[RH analysis for $\lambda \geq \lambda_{c}(t)$]{RH analysis for $\boldsymbol{\lambda \geq \lambda_{c}(t)}$}\label{Section: RH analysis region 1}
In this section, $t$ lies in a compact subset of $(-1,\infty)$ and $\lambda$ lies in $[\lambda_{c}(t),\infty]$ as $n \to \infty$. The parameter $\lambda$ is not necessarily bounded, and the case $\lambda = +\infty$ (i.e., $s = 0$) is also included in this analysis. First, we express the jumps for $T$ in terms of $\xi(z)$, which is def\/ined by
\begin{gather}\label{lol 29}
\xi(z) = - \pi \int_{\overline{c}}^{z}\tilde{\rho}(w){\rm d}w,
\end{gather}
where the path of integration lies in $\mathbb{C}\setminus (-\infty,\overline{c})$ and $\tilde{\rho}$ is given by
\begin{gather}\label{lol 48}
\tilde{\rho}(z) = \frac{2}{\pi} ( z - \overline{b}) \frac{\sqrt{z-\overline{c}}}{\sqrt{z-t}},
\end{gather}
where in the above expression the principal branch is chosen for each square root. The func\-tion~$\xi$ is analytic in $\mathbb{C}\setminus (-\infty,\overline{c}]$, and since $\tilde{\rho}_{\pm}(x) = \pm i \rho(x)$ for $x \in \mathcal{S} = [t,\overline{c}]$, we have
\begin{gather} \label{xi +}
2\xi_{\pm}(x) = \pm(g_{+}(x)-g_{-}(x)) = 2g_{\pm}(x) + \ell - 2 x^{2}, \qquad x \in \mathcal{S},
\end{gather}
where we have used \eqref{var equality}, \eqref{g+ + g-} and \eqref{g+ - g- 2}. Thus the function~$\xi(z)-g(z)$ has no jump along~$\mathcal{S}$ and can be analytically continued on $\mathbb{C}$. This implies the following relation
\begin{gather}\label{lol 14}
\xi(z) = g(z) + \frac{\ell}{2} - z^{2}, \qquad \mbox{for all} \quad z \in \mathbb{C}\setminus(-\infty,\overline{c}).
\end{gather}
The jump matrix $J_{T}$ can be rewritten in terms of $\xi$ as follows:
\begin{gather}
J_{T}(x) = \begin{cases}
 \begin{pmatrix}
1 & |x-t|^{\alpha}e^{n(\xi_{+}(x)+\xi_{-}(x)-\lambda)} \\
0 & 1
\end{pmatrix}, & \mbox{if } x < t, \vspace{1mm}\\
\begin{pmatrix}
e^{-2n\xi_{+}(x)} & |x-t|^{\alpha} \\
0 & e^{2n\xi_{+}(x)}
\end{pmatrix}, & \mbox{if } t < x < \overline{c}, \vspace{1mm}\\
 \begin{pmatrix}
1 & |x-t|^{\alpha}e^{2n\xi(x)} \\
0 & 1
\end{pmatrix}, & \mbox{if } \overline{c}<x. \\
\end{cases}
\end{gather}
As $\xi(z)$ appears in the jumps for $T$ (and in the subsequent transformations), it will be useful for us to make the following observations. From \eqref{var inequality} together with Proposition~\ref{prop: equilibrium measure}, and from~\eqref{g+ + g-} and~\eqref{lol 14}, we have that
\begin{gather}
 \xi(x)<0, \qquad \mbox{for} \quad x > \overline{c}, \label{lol 27} \\
\xi_{+}(x)+\xi_{-}(x)-\lambda < 0, \qquad \mbox{for} \quad x < t, \label{lol 28}
\end{gather}
except if $\lambda = \lambda_{c}(t)$ and $x = \overline{b}$, in which case \eqref{lol 28} becomes an equality. Also, if $x \in (t,\overline{c})$, from the def\/inition \eqref{lol 29} we have $\xi_{\pm}(x) \in \pm i \mathbb{R}^{+}$ and by Cauchy-Riemann equations we have
\begin{gather*}
 \partial_{x} \Im \xi_{+}(x) = - \pi \rho(x) = -\partial_{y} \Re \xi(x+iy) \big|_{y=0}, \\
 \partial_{x} \Im \xi_{-}(x) = \pi \rho(x) = \partial_{y} \Re \xi(x-iy) \big|_{y=0}.
\end{gather*}
In particular, this implies that there exists an open neighbourhood $W$ of $(t,\overline{c})$ such that we have
\begin{gather} \label{lol 30}
\Re \xi(z) >0, \qquad \mbox{for} \quad z \in W\setminus (t,\overline{c}).
\end{gather}
\subsection[Second transformation: $T \mapsto S$]{Second transformation: $\boldsymbol{T \mapsto S}$}
We will use the following factorization of $J_{T}(x)$ for $x \in \mathcal{S}$
\begin{gather*}
\begin{pmatrix}
e^{-2n\xi_{+}(x)} & |x-t|^{\alpha} \\
0 & e^{-2n\xi_{-}(x)}
\end{pmatrix} = \begin{pmatrix}
1 & 0 \\ |x-t|^{-\alpha}e^{-2n \xi_{-}(x)} & 1
\end{pmatrix} \\
\hphantom{\begin{pmatrix}
e^{-2n\xi_{+}(x)} & |x-t|^{\alpha} \\
0 & e^{-2n\xi_{-}(x)}
\end{pmatrix} =}{}
\times \begin{pmatrix}
0 & |x-t|^{\alpha} \\ -|x-t|^{-\alpha} & 0
\end{pmatrix} \begin{pmatrix}
1 & 0 \\ |x-t|^{-\alpha}e^{-2n \xi_{+}(x)} & 1
\end{pmatrix}.
\end{gather*}
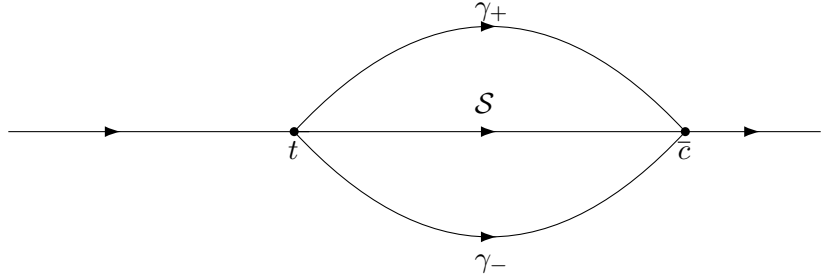
\begin{figure}[t]\centering
 \setlength{\unitlength}{1truemm}
 \begin{picture}(100,48)(28,10)
 \put(58,40){\line(1,0){70}}
 \put(60,40){\line(-1,0){40}}
 \put(85,40){\thicklines\vector(1,0){.0001}}
 \put(82,42.5){$\mathcal{S}$}
 \put(120,40){\thicklines\vector(1,0){.0001}}
 \put(35,40){\thicklines\vector(1,0){.0001}}
 \put(58,40){\thicklines\circle*{1.2}}
 \put(57.2,36.4){$t$}
 \put(110,40){\thicklines\circle*{1.2}}
 \put(109,36.4){$\overline{c}$}
 \qbezier(58,40)(84,68)(110,40)
 \put(85,54){\thicklines\vector(1,0){.0001}}
 \put(82,55.5){$\gamma_{+}$}
 \qbezier(58,40)(84,12)(110,40)
 \put(85,26){\thicklines\vector(1,0){.0001}}
 \put(82,22){$\gamma_{-}$}
 \end{picture}
 \vspace{-13mm}

 \caption{Jump contour for $S$. \label{fig open lens contours}}
\end{figure}
We open the lenses with $\gamma_{+}$ and $\gamma_{-}$ around $S$ as illustrated in Fig.~\ref{fig open lens contours}, such that $\gamma_{+}\cup\gamma_{-} \subset W$ and we def\/ine
\begin{gather*}
S(z) = T(z) \begin{cases}
\begin{pmatrix}
1 & 0 \\ -(z-t)^{-\alpha}e^{-2n \xi(z)} & 1
\end{pmatrix}, & \mbox{if } z \mbox{ is inside the lenses}, \ \Im z > 0, \\
\begin{pmatrix}
1 & 0 \\ (z-t)^{-\alpha}e^{-2n \xi(z)} & 1
\end{pmatrix}, & \mbox{if } z \mbox{ is inside the lenses}, \ \Im z < 0, \\
I, & \mbox{if } z \mbox{ is outside the lenses}, \\
\end{cases}
\end{gather*}
where the principal branch is taken for $(z-t)^{-\alpha}$. $S$ satisf\/ies the following RH problem.
\subsubsection*{RH problem for $\boldsymbol{S}$}
\begin{itemize}\itemsep=0pt
\item[(a)] $S \colon \mathbb{C}\setminus (\mathbb{R} \cup \gamma_{+} \cup \gamma_{-}) \to \mathbb{C}^{2\times 2}$ is analytic, where $\gamma_{+}$ and $\gamma_{-}$ are shown in Fig.~\ref{fig open lens contours}.
\item[(b)] $S$ has the following jumps:
\begin{gather*}
 S_{+}(z) = S_{-}(z)\begin{pmatrix}
1 & |z-t|^{\alpha}e^{n(\xi_{+}(z)+\xi_{-}(z)-\lambda)} \\
0 & 1
\end{pmatrix}, \qquad \mbox{if} \quad z < t, \\
 S_{+}(z) = S_{-}(z)\begin{pmatrix}
1 & |z-t|^{\alpha}e^{2n\xi(z)} \\
0 & 1
\end{pmatrix}, \qquad \mbox{if} \quad \overline{c} < z, \\
 S_{+}(z) = S_{-}(z)\begin{pmatrix}
0 & |z-t|^{\alpha} \\ -|z-t|^{-\alpha} & 0
\end{pmatrix}, \qquad \mbox{if}\quad t < z < \overline{c}, \\
 S_{+}(z) = S_{-}(z)\begin{pmatrix}
1 & 0 \\ (z-t)^{-\alpha}e^{-2n \xi(z)} & 1
\end{pmatrix}, \qquad \mbox{if} \quad z \in \gamma_{+} \cup \gamma_{-}.
\end{gather*}
\item[(c)] As $z \to \infty$, we have $S(z) = I + \mathcal{O}\big(z^{-1}\big)$.
\item[(d)] As $z$ tends to $t$, we have{\samepage
\begin{gather}
 S(z) = \begin{cases}
\begin{pmatrix}
\mathcal{O}(1) & \mathcal{O}\big(\log (z-t)\big) \\
\mathcal{O}(1) & \mathcal{O}\big(\log (z-t)\big)
\end{pmatrix}, & z \ \mbox{outside the lenses}, \\
\begin{pmatrix}
\mathcal{O}\big(\hspace*{-0.05cm}\log (z-t)\big) & \mathcal{O}\big(\hspace*{-0.05cm}\log (z-t)\big) \\
\mathcal{O}\big(\hspace*{-0.05cm}\log (z-t)\big) & \mathcal{O}\big(\hspace*{-0.05cm}\log (z-t)\big)
\end{pmatrix}, & z \ \mbox{inside the lenses},
\end{cases} \qquad \mbox{if} \quad \alpha = 0, \nonumber\\
 S(z) = \begin{cases}
\begin{pmatrix}
\mathcal{O}(1) & \mathcal{O}(1) \\
\mathcal{O}(1) & \mathcal{O}(1)
\end{pmatrix}
, & z \ \mbox{outside the lenses}, \\
\begin{pmatrix}
\mathcal{O}\big((z-t)^{- \alpha }\big) & \mathcal{O}(1) \\
\mathcal{O}\big((z-t)^{- \alpha }\big) & \mathcal{O}(1)
\end{pmatrix}
, & z \ \mbox{inside the lenses},
\end{cases} \qquad \mbox{if}\quad \alpha > 0, \nonumber\\
 S(z) = \begin{pmatrix}
\mathcal{O}(1) & \mathcal{O}\big((z-t)^{ \alpha }\big) \\
\mathcal{O}(1) & \mathcal{O}\big((z-t)^{ \alpha }\big)
\end{pmatrix}
, \qquad \mbox{if} \quad \alpha < 0.\label{behaviour of S near t}
\end{gather}
As $z$ tends to $\overline{c}$, we have $S(z) = \mathcal{O}(1)$.}
\end{itemize}
From \eqref{lol 27}, \eqref{lol 28} and \eqref{lol 30}, we have that the jumps for $S(z)$ on the boundary of the lenses $\gamma_{+}\cup\gamma_{-}$ tend to the identity matrix as $n \to \infty$, and that the $(1,2)$ entry of the jumps on $\mathbb{R}\setminus ([t,\overline{c}]\cup\{\overline{b}\})$ tends to $0$ as $n \to \infty$. This convergence is slower when $z$ approaches $t$ and $\overline{c}$, and also when $z$ approaches $\overline{b}$ if $\lambda = \lambda_{c}(t)$. The jump for $S$ on $(t,\overline{c})$ is independent of $n$ and dif\/ferent from the identity matrix.

\subsection{Global parametrix}
Ignoring the exponentially small terms as $n \to \infty$ in the jumps of $S$ and a small neighbourhood of $\overline{b}$, $t$ and $\overline{c}$, we are left to consider the following RH problem, whose solution $P^{(\infty)}$ is a good approximation of $S$ away from a neighbourhood of $\overline{b}$, $t$ and $\overline{c}$.
\subsubsection*{RH problem for $\boldsymbol{P^{(\infty)}}$}
\begin{itemize}\itemsep=0pt
\item[(a)] $P^{(\infty)} \colon \mathbb{C}\setminus [t,\overline{c}] \to \mathbb{C}^{2\times 2}$ is analytic.
\item[(b)] $P^{(\infty)}$ has the following jumps:
\begin{gather}
 P^{(\infty)}_{+}(z) = P^{(\infty)}_{-}(z)\begin{pmatrix}
0 & |z-t|^{\alpha} \\ -|z-t|^{-\alpha} & 0
\end{pmatrix}, \qquad \mbox{if} \quad t < z < \overline{c}. \label{lol 26}
\end{gather}
\item[(c)] As $z \to \infty$, we have $P^{(\infty)}(z) = I + \mathcal{O}\big(z^{-1}\big)$.
\item[(d)]
As $z$ tends to $t$, we have $P^{(\infty)}(z) = \mathcal{O}\big((z-t)^{-1/4}\big)(z-t)^{-\frac{\alpha}{2}\sigma_{3}}$.

As $z$ tends to $\overline{c}$, we have $P^{(\infty)}(z) = \mathcal{O}\big((z-\overline{c})^{-1/4}\big)$.
\end{itemize}
The construction of the solution of the above RH problem is now standard, and can be done similarly as in \cite{AtkChaZoh, Krasovsky, KMcLVAV}. We def\/ine $\beta(z) = \sqrt[4]{\frac{z-t}{z- \overline{c}}}$, analytic on $\mathbb{C}\setminus [t,\overline{c}]$ and such that $\beta(z) \sim 1$ as $z \to \infty$. It can be checked that the unique solution of the above RH problem is given by
\begin{gather}
P^{(\infty)}(z) =\frac{1}{2} \left(\frac{4}{\overline{c}-t}\right)^{-\frac{\alpha}{2}\sigma_{3}}
\begin{pmatrix}
\beta(z)+\beta^{-1}(z) & i(\beta(z)-\beta^{-1}(z)) \\
-i(\beta(z)-\beta^{-1}(z)) & \beta(z)+\beta^{-1}(z)
\end{pmatrix} \nonumber\\
\hphantom{P^{(\infty)}(z) =}{} \times \varphi \left( \frac{2}{\overline{c}-t} \left(z-\frac{\overline{c}+t}{2}\right) \right)^{\frac{\alpha}{2}\sigma_{3}}(z-t)^{-\frac{\alpha}{2}\sigma_{3}},\label{Pinf Region 1}
\end{gather}
where $\varphi(z) = z + \sqrt{z^{2}-1}$ is analytic in $\mathbb{C}\setminus [-1,1]$ and such that $\varphi(z) \sim 2z$ as $z \to \infty$. We will later need the following expansion as $z\to \infty$:
\begin{gather}\label{asymptotics Pinf}
P_{11}^{(\infty)}(z) = 1 - \frac{\alpha(\overline{c}-t)}{4z} + \mathcal{O}\big(z^{-2}\big).
\end{gather}
\subsection[Local parametrix near $t$]{Local parametrix near $\boldsymbol{t}$}
Note that the assumption at the beginning of the section, i.e., that $t$ lies in a compact subset of $(-1,\infty)$ and $\lambda \geq \lambda_{c}(t)$ as $n \to \infty$, implies from Proposition \ref{prop: equilibrium measure} that there exists a constant $\delta >0$ independent of $n$ such that
\begin{gather*}
\delta < \min (t-\overline{b},\overline{c}-t).
\end{gather*}
Inside a disk $D_{t}$ around $t$, of radius f\/ixed but smaller than $\delta/3$, we want the local parametrix $P$ to satisfy exactly the same jumps as $S$ and to have the same behaviour as $S$ near $t$. Furthermore, the local parametrix $P$ should be close to the global parametrix on the boundary of the disk.
\subsubsection*{RH problem for $\boldsymbol{P}$}
\begin{itemize}\itemsep=0pt
\item[(a)] $P \colon D_{t}\setminus (\mathbb{R} \cup \gamma_{+} \cup \gamma_{-}) \to \mathbb{C}^{2\times 2}$ is analytic.
\item[(b)] $P$ has the following jumps:
\begin{gather}
 P_{+}(z) = P_{-}(z)\begin{pmatrix}
1 & |z-t|^{\alpha}e^{n(\xi_{+}(z)+\xi_{-}(z)-\lambda)} \\
0 & 1
\end{pmatrix}, \qquad \mbox{if} \quad z \in (-\infty,t)\cap D_{t}, \label{P1}\\
 P_{+}(z) = P_{-}(z)\begin{pmatrix}
0 & |z-t|^{\alpha} \\ -|z-t|^{-\alpha} & 0
\end{pmatrix}, \qquad \mbox{if} \quad z \in (t,\infty)\cap D_{t}, \nonumber\\
 P_{+}(z) = P_{-}(z)\begin{pmatrix}
1 & 0 \\ (z-t)^{-\alpha}e^{-2n \xi(z)} & 1
\end{pmatrix}, \qquad \mbox{if}\quad z \in (\gamma_{+} \cup \gamma_{-}) \cap D_{t}.\nonumber
\end{gather}
\item[(c)] As $n \to \infty$, we have $P(z) = \big(I + \mathcal{O}\big(n^{-1}\big)\big)P^{(\infty)}(z)$ uniformly for $z \in \partial D_{t}$.
\item[(d)] As $z$ tends to $t$, we have $S(z)P(z)^{-1} = \mathcal{O}(1)$.
\end{itemize}
The construction of a local parametrix associated with a FH singularity has been studied in \cite{ItsKrasovsky} when the singularity is a pure jump, and then in~\cite{FouMarSou} and~\cite{DIK} for the general case, and involves hypergeometric functions. On the other hand, the construction of a local parametrix associated to a pure root-type FH singularity involves Bessel functions \cite{KMcLVAV}. In the present case, we are in a presence of a FH singularity of both root-type and jump-type, but the signif\/icant dif\/ference is that the parameter $s$ (which parametrizes the jump) is exponentially small as $n\to\infty$. The solution of the present local parametrix will be expressed in terms of Bessel functions, exactly as for a pure root-type singularity. Nevertheless, as the $(1,2)$ element of the jump matrix for~$P$ in~\eqref{P1} is not zero (if $\lambda \neq +\infty$, i.e., $s \neq 0$), the construction of the solution of the above RH problem is not standard. It was done in \cite{ChCl1} for the case $\alpha = 0$. We generalize here the construction for a general $\alpha >-1$.

We will need a modif\/ied version of the Bessel model RH problem $P_{\mathrm{Be}}$, which is presented in Appendix~\ref{ApB}. We search for a matrix function $\widehat{P}_{\mathrm{Be}}$ which satisf\/ies the same jumps as $P_{\mathrm{Be}}$, see~\eqref{Jump for P_Be}, and an extra jump on $\mathbb{R}^{+}$ given by
\begin{gather*}
\widehat{P}_{\mathrm{Be}}(z)_{+} = \widehat{P}_{\mathrm{Be}}(z)_{-}\begin{pmatrix}
1 & e^{-\lambda n} \\ 0 & 1
\end{pmatrix}, \qquad z \in (0,\infty),
\end{gather*}
where the orientation of $(0,\infty)$ is taken from $0$ to $\infty$.
\subsubsection*{Modif\/ied Bessel model RH problem}
We def\/ine
\begin{gather}\label{def of F}
F(z) = P_{\mathrm{Be}}(z) K(z)^{-1} \begin{pmatrix}
1 & -h(z) \\ 0 & 1
\end{pmatrix}z^{-\frac{\alpha}{2}\sigma_{3}},
\end{gather}
where
\begin{gather*}
K(z) = \begin{cases}
I, & | \arg z | < \frac{2\pi}{3}, \\
\begin{pmatrix}
1 & 0 \\ -e^{\pi i \alpha} & 1
\end{pmatrix}, & \frac{2\pi}{3}< \arg z < \pi, \vspace{1mm}\\
\begin{pmatrix}
1 & 0 \\ e^{-\pi i \alpha} & 1
\end{pmatrix}, & -\pi < \arg z < -\frac{2\pi }{3},
\end{cases}\qquad h(z) = \begin{cases}
\displaystyle \frac{1}{2i \sin(\pi \alpha)}, & \mbox{if} \ \alpha \notin \mathbb{N}, \vspace{1mm}\\
\displaystyle \frac{(-1)^{\alpha}}{2\pi i} \log z, & \mbox{if} \ \alpha \in \mathbb{N}.
\end{cases}
\end{gather*}
From the jumps for $P_{\mathrm{Be}}$, given by \eqref{Jump for P_Be}, it can be checked that $F$ has no jumps at all on $\mathbb{C}$. Also, the behaviour of $P_{\mathrm{Be}}(z)$ as $z\to 0$, given by \eqref{local behaviour near 0 of P_Be}, implies that $0$ is a removable singularity of $F$, and thus $F$ is an entire function. We def\/ine $\widehat{P}_{\mathrm{Be}}$ by
\begin{gather*}
\widehat{P}_{\mathrm{Be}}(z) = (I+A(z))P_{\mathrm{Be}}(z),
\end{gather*}
where
\begin{gather}\label{Az}
A(z) = -e^{-\lambda n}h(-z) (-z)^{\alpha} F(z) \begin{pmatrix}
0 & 1 \\ 0 & 0
\end{pmatrix}F^{-1}(z),
\end{gather}
and if $\alpha \notin \mathbb{Z}$, $(-z)^{\alpha}$ is chosen with a branch cut on $[0,\infty)$ such that $(-z)^{\alpha} >0$ for $z < 0$.
Since~$F$ is entire, $A$ is analytic on $\mathbb{C}\setminus \mathbb{R}^{+}$. Therefore, it can be checked that $\widehat{P}_{\mathrm{Be}}$ is the solution of the following RH problem.
\subsubsection*{RH problem for $\boldsymbol{\widehat{P}_{\mathrm{Be}}}$}
\begin{itemize}\itemsep=0pt
\item[(a)] $\widehat{P}_{\mathrm{Be}}\colon \mathbb{C} \setminus (\Sigma_{B}\cup(0,\infty)) \to \mathbb{C}^{2\times 2}$ is analytic, where the orientation of $(0,\infty)$ is from $0$ to\-wards~$\infty$ and $\Sigma_{B}$ is the jump contour for $P_{\mathrm{Be}}$, shown in Fig.~\ref{figBessel}.
\item[(b)] $\widehat{P}_{\mathrm{Be}}$ satisf\/ies the jump conditions
\begin{gather*}
 \widehat{P}_{\mathrm{Be},+}(z) = \widehat{P}_{\mathrm{Be},-}(z) \begin{pmatrix}
0 & 1 \\ -1 & 0
\end{pmatrix}, \qquad z \in \mathbb{R}^{-}, \\
 \widehat{P}_{\mathrm{Be},+}(z) = \widehat{P}_{\mathrm{Be},-}(z) \begin{pmatrix}
1 & 0 \\ e^{\pi i \alpha} & 1
\end{pmatrix},\qquad z \in e^{\frac{2\pi i }{3}}\mathbb{R}^{+}, \\
 \widehat{P}_{\mathrm{Be},+}(z) = \widehat{P}_{\mathrm{Be},-}(z) \begin{pmatrix}
1 & 0 \\ e^{-\pi i \alpha} & 1
\end{pmatrix}, \qquad z \in e^{-\frac{2\pi i }{3}}\mathbb{R}^{+}, \\
\widehat{P}_{\mathrm{Be},+}(z) = \widehat{P}_{\mathrm{Be},-}(z) \begin{pmatrix}
1 & e^{-\lambda n} \\ 0 & 1
\end{pmatrix}, \qquad z \in \mathbb{R}^{+}.
\end{gather*}
\item[(c)] As $z \to \infty$, $z \notin \Sigma_{B}\cup(0,\infty)$, we have
\begin{gather*}
\widehat{P}_{\mathrm{Be}}(z) = (I+A(z))\big( 2\pi z^{\frac{1}{2}} \big)^{-\frac{\sigma_{3}}{2}} N \big(
I+\mathcal{O} \big(z^{-\frac{1}{2}}\big)\big) e^{2z^{\frac{1}{2}}\sigma_{3}},
\end{gather*}
where $N = \frac{1}{\sqrt{2}}\left(\begin{smallmatrix}
1 & i \\ i & 1
\end{smallmatrix}\right)$.
\item[(d)]
As $z$ tends to 0, the behaviour of $\widehat{P}_{\mathrm{Be}}(z)$ is
\begin{gather*}
 P_{\mathrm{Be}}(z)^{-1}\widehat{P}_{\mathrm{Be}}(z) = \mathcal{O}(\log z), \qquad \mbox{if} \quad \alpha \in \mathbb{N}, \\
P_{\mathrm{Be}}(z)^{-1}\widehat{P}_{\mathrm{Be}}(z) = \mathcal{O}(1), \qquad \mbox{if} \quad \alpha \notin \mathbb{N}.
\end{gather*}
\end{itemize}
\subsubsection*{Construction of the local parametrix}
We consider the function
\begin{gather}\label{lol 43}
f(z) = -\frac{1}{4}\tilde{\xi}(z)^{2}, \qquad \mbox{where} \quad \tilde{\xi}(z) = \begin{cases}
\xi(z)-\xi_{+}(t), & \mbox{if } \Im z > 0, \\
\xi(z)-\xi_{-}(t), & \mbox{if } \Im z < 0. \\
\end{cases}
\end{gather}
This a conformal map from $D_{t}$ to a neighbourhood of $0$, and as $z \to t$, we have
\begin{gather}\label{lol 24}
f(z) = k_{1}^{2}(z-t)(1+\mathcal{O}(z-t)), \qquad \mbox{where} \quad k_{1} = 2(t-\overline{b}) \sqrt{\overline{c}-t}.
\end{gather}
To construct the solution $P$ of the above RH problem, it is important to note that $A(-n^{2}f(z))$ remains small as $n \to \infty$ uniformly for $z \in \partial D_{t}$. More precisely, from \eqref{def of F} and \eqref{large z asymptotics Bessel}, as $n \to \infty$ we have
\begin{gather}
F\big({-}n^{2}f(z)\big) = \mathcal{O}(\log n) \mathcal{O}\big(n^{| \alpha|}\big) \mathcal{O}\big(P_{\mathrm{Be}}\big({-}n^{2}f(z)\big)\big) = \mathcal{O}\big(e^{(d+\epsilon)n}\big),
\end{gather}
where $\epsilon >0$ can be chosen arbitrary small but f\/ixed, and $d = \max\limits_{z \in \partial D_{t}}\big|2\sqrt{-f(z)}\big|$. Similarly, we have the estimate $F^{-1}\big({-}n^{2}f(z)\big) = \mathcal{O}\big(e^{(d+\epsilon)n}\big)$. We choose the radius of $D_{t}$ f\/ixed but suf\/f\/iciently small such that $d < \frac{\lambda_{c}(t)}{3}\leq \frac{\lambda}{3}$. This implies from \eqref{Az} that $A\big({-}n^{2}f(z)\big) = \mathcal{O}\big(e^{-\frac{\lambda}{3}n}\big)$ as $n \to \infty$, uniformly for $z \in \partial D_{t}$. The local parametrix is given by
\begin{gather}\label{local param near t, Region 1}
P(z) = E(z) \sigma_{3} \widehat{P}_{\mathrm{Be}}\big({-}n^{2}f(z)\big)\sigma_{3}e^{-n\xi(z)\sigma_{3}} e^{\frac{\pi i \alpha}{2}\tilde\theta(z)\sigma_{3}}(z-t)^{-\frac{\alpha}{2}\sigma_{3}},
\end{gather}
where
\begin{gather}
\tilde\theta(z) = \begin{cases}
+1, & \mbox{if } \Im z >0, \\
-1, & \mbox{if } \Im z <0,
\end{cases}
\end{gather}
and the function $E(z)$ is def\/ined for $z \in D_{t}$ by
\begin{gather}\label{E in local param near t, Region 1}
E(z) = (-1)^{n}P^{(\infty)}(z)(z-t)^{\frac{\alpha}{2}\sigma_{3}}e^{-\frac{\pi i \alpha}{2}\tilde\theta(z)\sigma_{3}}N \big( 2\pi n (-f(z))^{1/2} \big)^{\sigma_{3}/2}.
\end{gather}
It can be checked directly from the jumps for $P^{(\infty)}$ \eqref{lol 26} that $E$ has no jump at all in $D_{t}$. Furthermore, from the behaviour of $P^{(\infty)}(z)$ near $t$ and from~\eqref{lol 24}, one has $E(z) = \mathcal{O}\big((z-t)^{-\frac{1}{2}}\big)$ as $z \to t$. Thus, $t$ is a removable singularity of $E$ and $E$ is analytic in the whole disk $D_{t}$. Since $A(-n^{2}f(z))$ is exponentially small in $n$ uniformly for $z\in \partial D_{t}$, it doesn't contribute to the $n^{-1}$ term in the condition (c) of the RH problem for $P$. Using the large $\zeta$ asymptotics for the Bessel model RH problem given by~\eqref{large z asymptotics Bessel}, we obtain as $n \to \infty$ that
\begin{gather}
P(z)P^{\infty}(z)^{-1} = I + \frac{1}{n(-f(z))^{1/2}}P^{(\infty)}(z)(z-t)^{\frac{\alpha}{2}\sigma_{3}}e^{-\frac{\pi i \alpha}{2}\tilde\theta(z)\sigma_{3}}\sigma_{3}B_{1}\sigma_{3}\nonumber\\
\hphantom{P(z)P^{\infty}(z)^{-1} =}{}\times e^{\frac{\pi i \alpha}{2}\tilde\theta(z)\sigma_{3}}(z-t)^{-\frac{\alpha}{2}\sigma_{3}}P^{(\infty)}(z)^{-1} + \mathcal{O}\big(n^{-2}\big),\label{jump for R on partial D_t}
\end{gather}
uniformly for $z \in \partial D_{t}$, where $B_{1} = \frac{1}{16}\left(\begin{smallmatrix}
-(1+4\alpha^{2}) & -2i \\ -2i & 1+4\alpha^{2}
\end{smallmatrix}\right)$.

\subsection[Local parametrix near $\overline{c}$]{Local parametrix near $\boldsymbol{\overline{c}}$}
Inside a disk $D_{\overline{c}}$ around $\overline{c}$, of radius f\/ixed but smaller than $\delta/3$, we want the local paramet\-rix~$P$ to satisfy the following RH problem.
\subsubsection*{RH problem for $\boldsymbol{P}$}
\begin{itemize}\itemsep=0pt
\item[(a)] $P \colon D_{\overline{c}}\setminus (\mathbb{R} \cup \gamma_{+} \cup \gamma_{-}) \to \mathbb{C}^{2\times 2}$ is analytic.
\item[(b)] $P$ has the following jumps:
\begin{gather*}
 P_{+}(z) = P_{-}(z)\begin{pmatrix}
0 & |z-t|^{\alpha} \\ -|z-t|^{-\alpha} & 0
\end{pmatrix}, \qquad \mbox{if} \quad z \in (-\infty,\overline{c})\cap D_{\overline{c}}, \\
 P_{+}(z) = P_{-}(z)\begin{pmatrix}
1 & |z-t|^{\alpha}e^{2n\xi(z)} \\
0 & 1
\end{pmatrix}, \qquad \mbox{if} \quad z \in (\overline{c},\infty)\cap D_{\overline{c}},\\
 P_{+}(z) = P_{-}(z)\begin{pmatrix}
1 & 0 \\ (z-t)^{-\alpha}e^{-2n \xi(z)} & 1
\end{pmatrix}, \qquad \mbox{if} \quad z \in (\gamma_{+} \cup \gamma_{-})\cap D_{\overline{c}}.
\end{gather*}
\item[(c)] As $n \to \infty$, we have $P(z) = \big(I + \mathcal{O}\big(n^{-1}\big)\big)P^{(\infty)}(z)$ uniformly for $z \in \partial D_{\overline{c}}$.
\item[(d)] As $z$ tends to $\overline{c}$, we have $P(z) = \mathcal{O}(1)$.
\end{itemize}
The solution $P$ of the above RH problem is standard \cite{DKMVZ1} and can be constructed in terms of Airy functions and the associated Airy model RH problem, whose solution is denoted $P_{\mathrm{Ai}}$ and is presented in Appendix \ref{ApA}. Let us f\/irst def\/ine the function
\begin{gather*}
f(z) = \left( -\frac{3}{2}\xi(z) \right)^{2/3}.
\end{gather*}
This is a conformal map from $D_{\overline{c}}$ to a neighbourhood of $0$, and as $z \to \overline{c}$ we have
\begin{gather*}
f(z) = k_{2}^{2/3} (z-\overline{c}) \left[ 1+ \frac{2}{5}k_{3} (z-\overline{c}) + \mathcal{O}\big((z-\overline{c})^{2}\big) \right],
\end{gather*}
with
\begin{gather*}
k_{2} = 2 \frac{\overline{c}-\overline{b}}{\sqrt{\overline{c}-t}}, \qquad
k_{3} = \frac{1}{\overline{c}-\overline{b}}-\frac{1}{2(\overline{c}-t)}.
\end{gather*}
It can be verif\/ied that
\begin{gather*}
P(z) = E(z)P_{\mathrm{Ai}}\big(n^{2/3}f(z)\big)e^{-n\xi(z)\sigma_{3}}(z-t)^{-\frac{\alpha}{2}\sigma_{3}},
\end{gather*}
satisf\/ies the above RH problem, where
\begin{gather*}
E(z) = P^{(\infty)}(z)(z-t)^{\frac{\alpha}{2}\sigma_{3}}N^{-1}f(z)^{\frac{\sigma_{3}}{4}}n^{\frac{\sigma_{3}}{6}}.
\end{gather*}
Again, one can show that $E$ has no jump at all inside $D_{\overline{c}}$ and has a removable singularity at $\overline{c}$, and therefore $E$ is analytic in the whole disk $D_{\overline{c}}$. We will also need explicitly the f\/irst term in the large $n$ expansion of $P(z)P^{(\infty)}(z)^{-1}$ on $\partial D_{\overline{c}}$. As $n \to \infty$, by \eqref{Asymptotics Airy} we have
\begin{gather}
P(z)P^{(\infty)}(z)^{-1} = I + \frac{1}{nf(z)^{3/2}}P^{(\infty)}(z)(z-t)^{\frac{\alpha}{2}\sigma_{3}} A_{1} (z-t)^{-\frac{\alpha}{2}\sigma_{3}}P^{(\infty)}(z)^{-1} \nonumber\\
\hphantom{P(z)P^{(\infty)}(z)^{-1} =}{} + \mathcal{O}\big(n^{-2}\big). \label{jump for R on partial D_c}
\end{gather}
uniformly for $z \in \partial D_{\overline{c}}$, and where $A_{1} = \frac{1}{8} \left(\begin{smallmatrix}
\frac{1}{6} & i \\ i & -\frac{1}{6}
\end{smallmatrix}\right)$,

\subsection[Local parametrix near $\overline{b}$]{Local parametrix near $\boldsymbol{\overline{b}}$}
The local parametrix $P$ in a f\/ixed disk $D_{\overline{b}}$ around $\overline{b}$ can be constructed explicitly. The construction is similar to the one done in \cite{ChCl1} and it is valid for every $\lambda \geq \lambda_{c}$ but is only needed for $\lambda$ close to $\lambda_{c}$.
\subsubsection*{RH problem for $\boldsymbol{P}$}
\begin{itemize}\itemsep=0pt
\item[(a)] $P \colon D_{\overline{b}}\setminus \mathbb{R} \to \mathbb{C}^{2\times 2}$ is analytic.
\item[(b)] $P$ has the following jumps:
\begin{gather*}
P_{+}(z) = P_{-}(z)\begin{pmatrix}
1 & |z-t|^{\alpha}e^{n(\xi_{+}(z)+\xi_{-}(z)-\lambda)} \\
0 & 1
\end{pmatrix}, \qquad \mbox{if} \quad z \in \mathbb{R}\cap D_{\overline{b}}.
\end{gather*}
\item[(c)] As $n \to \infty$, we have $P(z) = (I + o(1))P^{(\infty)}(z)$ uniformly for $z \in \partial D_{\overline{b}}$.
\end{itemize}
Note that for $z \in D_{\overline{b}}$, we have
\begin{gather}\label{lol 31}
\phi(z) = \xi_{+}(z) + \xi_{-}(z) = 2\pi \int_{z}^{t}\tilde{\rho}(w){\rm d}w,
\end{gather}
and $\phi$ is analytic in $D_{\overline{b}}$. The unique solution of the above RH problem is given by
\begin{gather*}
P(z) = P^{(\infty)}(z) \begin{pmatrix}
1 & l(z) \\ 0 & 1,
\end{pmatrix}
\end{gather*}
where
\begin{gather*}
l(z) = \frac{(z-t)^{\alpha}e^{-\pi i \alpha \tilde\theta(z)} e^{-n(\lambda - \lambda_{c})}}{2\pi i} \int_{\mathbb{R}\cap D_{\overline{b}}}\frac{e^{n(\phi(x)-\lambda_{c})}}{x-z}{\rm d}x.
\end{gather*}
From \eqref{lol 28} and \eqref{lol 31}, we have $\phi(\overline{b}) = \lambda_{c}$, $\phi^{\prime}(\overline{b}) = 0$ and $\phi^{\prime\prime}(\overline{b})<0$. Thus, as $n \to \infty$ we have $l(z) = \mathcal{O}(n^{-1/2}e^{-n(\lambda-\lambda_{c})})$ and
\begin{gather*}
P(z) = \big(I + \mathcal{O}\big(n^{-1/2}e^{-n(\lambda-\lambda_{c})}\big)\big)P^{(\infty)}(z), \qquad \mbox{uniformly for} \quad z \in \partial D_{\overline{b}}.
\end{gather*}
\subsection{Small norm RH problem}\label{subsection: small norm region 1}
The f\/inal transformation of the RH analysis is given by
\begin{gather*}
R(z) = \begin{cases}
S(z) P^{(\infty)}(z)^{-1}, & z \in \mathbb{C}\setminus (\overline{D_{t}}\cup \overline{D_{\overline{c}}} \cup \overline{D_{\overline{b}}}), \\
S(z) P(z)^{-1}, & z \in D_{t} \cup D_{\overline{c}} \cup D_{\overline{b}}.
\end{cases}
\end{gather*}
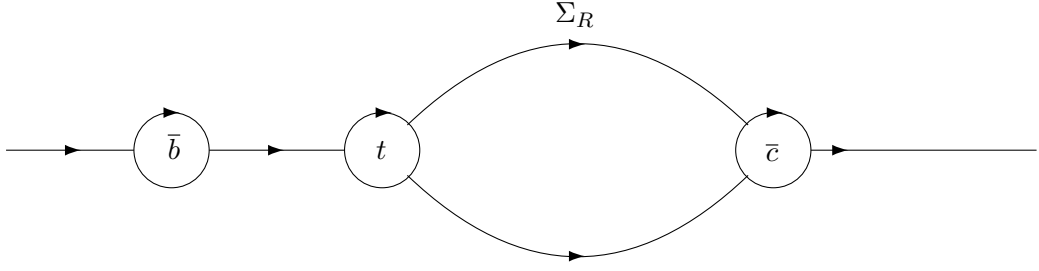
\begin{figure}[t]\centering
 \setlength{\unitlength}{1truemm}
 \begin{picture}(100,38)(28,23)
 \put(81,57){$\Sigma_{R}$}
 \put(115,40){\line(1,0){30}}
 \put(53,40){\line(-1,0){18}}
 \put(25,40){\line(-1,0){17}}
 \put(120,40){\thicklines\vector(1,0){.0001}}
 \put(18,40){\thicklines\vector(1,0){.0001}}
 \put(45,40){\thicklines\vector(1,0){.0001}}
 \put(29.5,38.8){$\overline{b}$} \put(30,40){\circle{10}}\put(31.1,45){\thicklines\vector(1,0){.0001}}
 \put(58,40){\circle{10}} \put(59.3,45){\thicklines\vector(1,0){.0001}}
 \put(57.2,38.8){$t$}
 \put(110,40){\circle{10}} \put(111.3,45){\thicklines\vector(1,0){.0001}}
 \put(109,38.5){$\overline{c}$}
 \qbezier(61.35,43.35)(84,65)(106.65,43.35)
 \put(85,54.1){\thicklines\vector(1,0){.0001}}
 \qbezier(61.35,36.65)(84,15)(106.65,36.65)
 \put(85,25.9){\thicklines\vector(1,0){.0001}}
 \end{picture}
 \vspace{-3mm}

 \caption{Jump contours for the RH problem for $R$. The circles are oriented in clockwise direction. \label{Fig:R region 1}}
\end{figure}
Since $P$ has exactly the same jumps as $S$ inside $D_{t} \cup D_{\overline{c}} \cup D_{\overline{b}}$, and the same behaviour near~$t$ and $\overline{c}$, $R$ has no jumps and is analytic inside these disks, except possibly at $\overline{b}$, $\overline{c}$ and $t$. From the RH problem for $S$, and from the local parametrices around~$\overline{b}$ and~$\overline{c}$, we have that~$S(z)$ and~$P(z)$ are bounded as $z \to \overline{b}$ and as $z \to \overline{c}$. From \eqref{behaviour of S near t}, \eqref{local param near t, Region 1} and \eqref{local behaviour near 0 of P_Be}, as $z \to t$ from outside the lenses, we have
\begin{gather}\label{growth of R near t}
R(z) = \begin{cases}
 \begin{pmatrix}
\mathcal{O}\big(\log (z-t)\big) & \mathcal{O}\big(\log (z-t)\big) \\
\mathcal{O}\big(\log (z-t)\big) & \mathcal{O}\big(\log (z-t)\big)
\end{pmatrix}, & \alpha = 0,\\
\begin{pmatrix}
\mathcal{O}(1) & \mathcal{O}(1) \\
\mathcal{O}(1) & \mathcal{O}(1)
\end{pmatrix}, & \alpha > 0, \\
\begin{pmatrix}
\mathcal{O}\big((z-t)^{\alpha}\big) & \mathcal{O}\big((z-t)^{\alpha}\big) \\
\mathcal{O}\big((z-t)^{\alpha}\big) & \mathcal{O}\big((z-t)^{\alpha}\big)
\end{pmatrix}, & \alpha < 0.
\end{cases}
\end{gather}
By \eqref{growth of R near t}, and since $R(z)$ is bounded as $z \to \overline{b}$ and as $z \to \overline{c}$, the three isolated singularities at $\overline{b}$, $\overline{c}$ and $t$ are removable. On the circles $\partial D_{t}$, $\partial D_{\overline{c}}$ and $\partial D_{\overline{c}}$, we choose the clockwise orientation, as shown in Fig.~\ref{Fig:R region 1}. $R$ satisf\/ies the following RH problem:
\subsubsection*{RH problem for $\boldsymbol{R}$}
\begin{itemize}\itemsep=0pt
\item[(a)] $R \colon \mathbb{C}\setminus \Sigma_{R} \to \mathbb{C}^{2\times 2}$ is analytic, where the contour $\Sigma_{R}$ is shown in Fig.~\ref{Fig:R region 1}.
\item[(b)] The jumps $J_R(z):=R_-^{-1}(z)R_+(z)$ satisfy the following large $n$ asymptotics for $z\in\Sigma_R$:
\begin{gather*}
 J_R(z)=I+\mathcal{O}\big(e^{-cn}\big), \qquad \mbox{uniformly for} \quad z \in (\gamma_{+} \cup \gamma_{-}\cup \mathbb{R})\setminus (\mathcal{S}\cup \overline{D_{t}} \cup \overline{D_{\overline{c}}} \cup \overline{D_{\overline{b}}} ), \nonumber \\
 J_R(z)= I + \mathcal{O}\big(n^{-1}\big), \qquad \mbox{uniformly for} \quad z \in \partial D_{t} \cup \partial D_{\overline{c}}, \nonumber\\
 J_R(z)=I + \mathcal{O}\big(n^{-1/2}e^{-n(\lambda-\lambda_{c})}\big), \qquad \mbox{uniformly for} \quad z \in \partial D_{\overline{b}} \nonumber,
\end{gather*}where $c>0$ is a constant.
\item[(c)] As $z \to \infty$, we have $R(z) = I+\mathcal{O}\big(z^{-1}\big)$.
\end{itemize}
From the standard theory for small-norm RH problems \cite{DKMVZ1}, $R$ exists for all $n$ suf\/f\/iciently large and as $n\to\infty$ we have
\begin{gather}
R(z) = I + \mathcal{O}\big(n^{-1}\big) + \mathcal{O}\big(n^{-1/2}e^{-n(\lambda - \lambda_{c})}\big), \nonumber\\
R^{\prime}(z) = \mathcal{O}\big(n^{-1}\big) + \mathcal{O}\big(n^{-1/2}e^{-n(\lambda - \lambda_{c})}\big),\label{lol 8}
\end{gather}
uniformly for $z \in \mathbb{C}\setminus \Sigma_{R}$, and uniformly for $t$ in compact subsets of $(-1,\infty)$ and for $\lambda \geq \lambda_{c}(t)$. In the rest of this section, we consider the case when $\lambda$ lies in a compact subset of $(\lambda_{c}(t),\infty]$ as $n \to \infty$, i.e., when $\lambda$ is bounded away from $\lambda_{c}(t)$. In this case, the jumps for $R$ on $\partial D_{\overline{b}}$ are exponentially small. On the other hand, the jumps for~$R$ on $\partial D_{t}\cup \partial D_{\overline{c}}$ have a series expansion of the form
\begin{gather*}
J_R(z) = I+\sum_{j=1}^{r} J_{R}^{(j)}(z)n^{-j} + \mathcal{O}\big(n^{-r-1}\big),
\end{gather*}
for any $r \in \mathbb{N}$, where $J_{R}^{(j)}(z) = \mathcal{O}(1)$ as $n \to \infty$ uniformly for $z \in \partial D_{t}\cup \partial D_{\overline{c}}$. Thus, $R$ admits a~series expansion of the form
\begin{gather*}
R(z) = I + \sum_{j=1}^{r} R^{(j)}(z) n^{-j} + \mathcal{O}\big(n^{-r-1}\big), \qquad \mbox{as} \quad n \to \infty,
\end{gather*}
for any $r \in \mathbb{N}$, where $R^{(j)}(z) = \mathcal{O}(1)$ as $n \to \infty$ uniformly for $z \in \mathbb{C}\setminus \Sigma_{R}$. By a perturbative analysis of the RH problem for $R$, the f\/irst term $R^{(1)}(z)$ is given by
\begin{gather*}
R^{(1)}(z) = \frac{1}{2\pi i}\int_{\partial D_{t}\cup \partial D_{\overline{c}}} \frac{J_{R}^{(1)}(w)}{w-z}{\rm d}w.
\end{gather*}
We can evaluate this integral explicitly. When $z$ is outside the disks, we have from~\eqref{jump for R on partial D_t} and~\eqref{jump for R on partial D_c} that
\begin{gather}
 \left(\frac{4}{\overline{c}-t}\right)^{\frac{\alpha}{2}\sigma_{3}} R^{(1)}(z)\left(\frac{4}{\overline{c}-t}\right)^{-\frac{\alpha}{2}\sigma_{3}}\nonumber\\
 = \frac{\sqrt{\overline{c}-t}}{z-t} \frac{4\alpha^{2}-1}{32 k_{1}} \begin{pmatrix}
1 & i \\ i & -1
\end{pmatrix} + \frac{1}{(z-\overline{c})^{2}} \frac{5\sqrt{\overline{c}-t}}{96 k_{2}}\begin{pmatrix}
-1 & i \\ i & 1
\end{pmatrix}\label{R first correction} \\
{} +\frac{1}{z-\overline{c}} \frac{1}{64 \sqrt{\overline{c}-t} k_{2}}\begin{pmatrix}
-8\alpha^{2}+2(\overline{c}-t)k_{3}+3 & \!\! \frac{i}{3}\big( 24(\alpha+2)\alpha - 6(\overline{c}-t)k_{3}+19 \big) \\ \frac{i}{3}\big( 24(\alpha-2)\alpha - 6(\overline{c}-t)k_{3}+19 \big)\!\! & 8\alpha^{2}-2(\overline{c}-t)k_{3}-3
\end{pmatrix}.\nonumber
\end{gather}
\section[RH analysis for $0 < \lambda < \lambda_{c}(t)$]{RH analysis for $\boldsymbol{0 < \lambda < \lambda_{c}(t)}$}
\label{Section: RH analysis region 2}
In this section, we analyse the case when $(t,\lambda)$ is in a compact subset of
\begin{gather*}
\mathcal{R} = \{ (t,\lambda)\colon t \in (-1,1) \mbox{ and } 0<\lambda < \lambda_{c}(t) \}
\end{gather*}
as $n \to \infty$. We def\/ine the function
\begin{gather*}
\tilde{\rho}(z) = \frac{2}{\pi}\frac{\sqrt{z-c}\sqrt{z-b}\sqrt{z-a}}{\sqrt{z-t}},
\end{gather*}
analytic in $\mathbb{C}\setminus ([a,b]\cup[t,c])$ such that $\tilde{\rho}(z) \sim \frac{2}{\pi}z$ as $z \to +\infty$. We will adopt the same approach as in Section \ref{Section: RH analysis region 1}, and we will rewrite the jumps for $T$ in terms of the following two functions:
\begin{gather}\label{def of xi1 and xi2}
\xi_{1}(z) = - \pi \int_{c}^{z}\tilde{\rho}(w){\rm d}w, \qquad \xi_{2}(z) = - \pi \int_{b}^{z}\tilde{\rho}(w){\rm d}w.
\end{gather}
For $\xi_{1}$ the path of integration is chosen to be in $\mathbb{C}\setminus (-\infty,c)$, and for $\xi_{2}$ the path lies in $\mathbb{C}\setminus ((-\infty,b)\cup[t,\infty))$. Therefore, $\xi_{1}$ is analytic in $\mathbb{C}\setminus (-\infty,c)$, satisf\/ies $\xi_{1}(z) < 0$ for $z>c$ and $\xi_{2}$ is analytic in $\mathbb{C}\setminus ((-\infty,b)\cup[t,\infty))$. Similarly to \eqref{xi +}, note that $\xi_{1}(z)-g(z)$ satisf\/ies
\begin{gather}\label{ana cont}
\xi_{1,+}(z) - g_{+}(z) = \xi_{1,-}(z) - g_{-}(z) = \frac{\ell}{2}-z^{2}, \qquad z \in (t,c).
\end{gather}
Analytically continuing $\xi_{1}(z)-g(z)$ in \eqref{ana cont} for $z$ outside $(-\infty,t)$, we have
\begin{gather}\label{xi_1 in terms of g}
\xi_{1}(z) = g(z) + \frac{\ell}{2}-z^{2}, \qquad z \in \mathbb{C}\setminus (-\infty,c].
\end{gather}
It will be useful later to notice the connection formula between $g(z)$, $\xi_{1}(z)$ and $\xi_{2}(z)$ for $z \in (b,t)$:
\begin{gather}\label{connection xi1 and xi2}
g_{+}(z) + g_{-}(z) + \ell - 2z^{2} = \xi_{1,+}(z) + \xi_{1,-}(z) = 2 \xi_{2}(z) + 2\pi \int_{b}^{t}\tilde{\rho}(x){\rm d}x = 2\xi_{2}(z) + \lambda,
\end{gather}
where we have used \eqref{abc lambda}, \eqref{f in proof}, \eqref{f prime in proof} and \eqref{g+ + g-}. Furthermore, by \eqref{var equality} and \eqref{g+ + g-}, we have
\begin{gather*}
\xi_{2,+}(z) + \xi_{2,-}(z) = 0 = g_{+}(z)+g_{-}(z)-2z^{2}+\ell - \lambda, \qquad z \in (a,b),
\end{gather*}
and by \eqref{f prime in proof}, we also have
\begin{gather*}
(\xi_{2,+}(z)+\xi_{2,-}(z))^{\prime} = -2\pi \tilde{\rho}(z) = (g_{+}(z) + g_{-}(z) - 2 z^{2}+\ell - \lambda)^{\prime}, \qquad z<a.
\end{gather*}
Thus, we have the following identity between $\xi_{2}$ and $g$ on $(-\infty,a)$:
\begin{gather}\label{lol 34}
\xi_{2,+}(z)+\xi_{2,-}(z) = g_{+}(z) + g_{-}(z) - 2 z^{2}+\ell - \lambda < 0, \qquad z<a.
\end{gather}
The mass of $\rho$ on the interval $(a,b)$ will play an important role later and will appear in the jumps of the subsequent RH problems, and we denote it as follows:
\begin{gather}\label{Omega def}
\Omega(t,\lambda) = \int_{a}^{b} \rho(x;t,\lambda){\rm d}x.
\end{gather}
We will sometimes omit the dependence of $\Omega(t,\lambda)$ in $t$ and $\lambda$ and simply write $\Omega$ when there is no confusion. For $z \in (a,b)$, by the relation \eqref{g+ - g- 2}, we have this identity
\begin{gather*}
2\xi_{2,+}(z) = 2\pi i \int_{z}^{b} \rho(x){\rm d}x = g_{+}(z)-g_{-}(z) + 2\pi i \Omega - 2\pi i.
\end{gather*}
Also, by Cauchy--Riemann equations, we can show similarly to equation \eqref{lol 30} that there exists an open neighbourhood $W_{1}$ of $(t,c)$ and an open neighbourhood $W_{2}$ of $(a,b)$ such that
\begin{gather}
 \Re \xi_{1}(z) > 0, \qquad \mbox{for} \quad z \in W_{1}\setminus (t,c), \label{lol 32} \\
 \Re \xi_{2}(z) > 0, \qquad \mbox{for} \quad z \in W_{2}\setminus (a,b). \label{lol 33}
\end{gather}
The jumps for $T$ can now be rewritten as
\begin{gather*}
J_{T}(x) = \begin{cases}
\begin{pmatrix}
1 & |x-t|^{\alpha}e^{n(\xi_{2,+}(x)+\xi_{2,-}(x))} \\
0 & 1
\end{pmatrix}, & \mbox{if} \ x<a, \vspace{1mm}\\
\begin{pmatrix}
e^{2\pi i \Omega n}e^{-2n\xi_{2,+}(x)} & |x-t|^{\alpha} \\
0 & \displaystyle e^{-2\pi i \Omega n}e^{2n\xi_{2,+}(x)}
\end{pmatrix}, & \mbox{if} \ a < x < b, \vspace{1mm}\\
\begin{pmatrix}
e^{2\pi i \Omega n} & |x-t|^{\alpha}e^{n(\xi_{1,+}(x)+\xi_{1,-}(x)-\lambda)} \\
0 & \displaystyle e^{-2\pi i \Omega n}
\end{pmatrix}, & \mbox{if} \ b<x<t, \vspace{1mm}\\
\begin{pmatrix}
e^{-2n \xi_{1,+}(x)} & |x-t|^{\alpha} \\ 0 & e^{2n \xi_{1,+}(x)}
\end{pmatrix}, & \mbox{if} \ t<x<c, \vspace{1mm}\\
\begin{pmatrix}
 1 & |x-t|^{\alpha}e^{2n\xi_{1}(x)} \\
 0 & 1
\end{pmatrix}, & \mbox{if} \ c<x.
\end{cases}
\end{gather*}
For $x \in (b,t)$, by \eqref{connection xi1 and xi2}, one can also rewrite $J_{T}(x)$ as
\begin{gather*}
J_{T}(x) = \begin{pmatrix}
e^{2\pi i \Omega n} & |x-t|^{\alpha}e^{2n\xi_{2}(x)} \\
0 & e^{-2\pi i \Omega n}
\end{pmatrix}.
\end{gather*}
\subsection[Second transformation: $T \mapsto S$]{Second transformation: $\boldsymbol{T \mapsto S}$}
We proceed to the opening of the lenses with $\gamma_{+}$, $\gamma_{-}$, $\tilde{\gamma}_{+}$, and $\tilde{\gamma}_{-}$ as shown in Fig.~\ref{fig open lens contours Region 2}, such that $\gamma_{+}\cup \gamma_{-} \subset W_{1}$ and $\tilde\gamma_{+}\cup \tilde\gamma_{-} \subset W_{2}$. The next transformation $S$ is def\/ined by
\begin{gather}\label{lol 17}
S(z) = T(z) \begin{cases}
\begin{pmatrix}
1 & 0 \\ -(z-t)^{-\alpha}e^{-2n \xi_{1}(z)} & 1
\end{pmatrix}, & \mbox{if } z \in \mathcal{I}_{1}, \vspace{1mm}\\
\begin{pmatrix}
1 & 0 \\ (z-t)^{-\alpha}e^{-2n \xi_{1}(z)} & 1
\end{pmatrix}, & \mbox{if } z \in \mathcal{I}_{2}, \vspace{1mm}\\
\begin{pmatrix}
1 & 0 \\ -(z-t)^{-\alpha}e^{-2n \xi_{2}(z)}e^{2\pi i \Omega n} & 1
\end{pmatrix}, & \mbox{if } z \in \widetilde{\mathcal{I}}_{1}, \vspace{1mm}\\
\begin{pmatrix}
1 & 0 \\ (z-t)^{-\alpha}e^{-2n \xi_{2}(z)}e^{-2\pi i \Omega n} & 1
\end{pmatrix}, & \mbox{if } z \in \widetilde{\mathcal{I}}_{2}, \\
I, & \mbox{if } z \mbox{ is outside the lenses},
\end{cases}
\end{gather}
where the sectors $\mathcal{I}_{1}$, $\mathcal{I}_{2}$, $\widetilde{\mathcal{I}}_{1}$, and $\widetilde{\mathcal{I}}_{2}$ are inside the lenses and are shown in Fig.~\ref{fig open lens contours Region 2}.
\begin{figure}[t]\centering
 \setlength{\unitlength}{1truemm}
 \begin{picture}(100,38)(0,20)
 \put(58,40){\line(1,0){70}}
 \put(60,40){\line(-1,0){80}}
 \put(85,39.9){\thicklines\vector(1,0){.0001}}
 \put(120,39.9){\thicklines\vector(1,0){.0001}}
 \put(45,39.9){\thicklines\vector(1,0){.0001}}
 \put(58,40){\thicklines\circle*{1.2}}
 \put(57.2,36.4){$t$}
 \put(110,40){\thicklines\circle*{1.2}}
 \put(109,36.4){$c$}
 \qbezier(58,40)(84,68)(110,40)
 \put(85,54){\thicklines\vector(1,0){.0001}}
 \put(82,55.5){$\gamma_{+}$}
 \qbezier(58,40)(84,12)(110,40)
 \put(85,26){\thicklines\vector(1,0){.0001}}
 \put(82,22){$\gamma_{-}$}
 \put(30,40){\thicklines\circle*{1.2}}
 \put(30,36.4){$b$}
 \put(0,40){\thicklines\circle*{1.2}}
 \put(0,36.4){$a$}
 \qbezier(0,40)(15,63)(30,40)
 \qbezier(0,40)(15,17)(30,40)
 \put(-10,39.9){\thicklines\vector(1,0){.0001}}
 \put(16,39.9){\thicklines\vector(1,0){.0001}}
 \put(16,51.5){\thicklines\vector(1,0){.0001}}
 \put(16,28.5){\thicklines\vector(1,0){.0001}}
 \put(16,53){$\tilde\gamma_{+}$}
 \put(16,25){$\tilde\gamma_{-}$}
 \put(82,45.4){$\mathcal{I}_{1}$}
 \put(82,32){$\mathcal{I}_{2}$}
 \put(14,44.5){$\widetilde{\mathcal{I}}_{1}$}
 \put(14,33){$\widetilde{\mathcal{I}}_{2}$}
 \end{picture}
 \vspace{-3mm}

 \caption{Jump contour for $S$. \label{fig open lens contours Region 2}}
\end{figure}
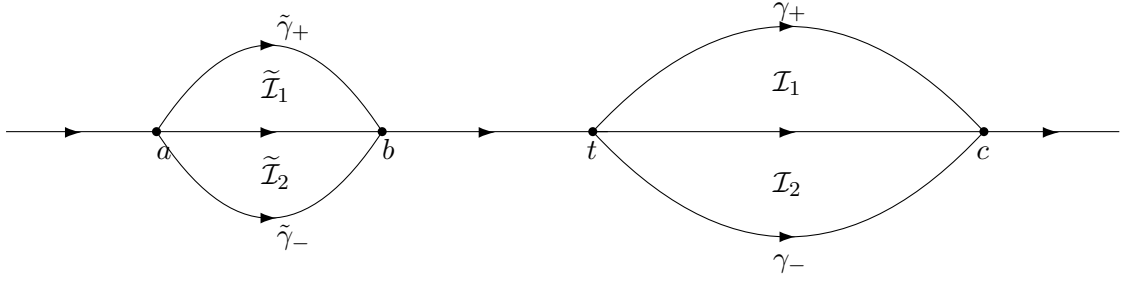
$S$ satisf\/ies the following RH problem.
\subsubsection*{RH problem for $\boldsymbol{S}$}
\begin{itemize}\itemsep=0pt
\item[(a)] $S \colon \mathbb{C}\setminus (\mathbb{R} \cup \gamma_{+} \cup \gamma_{-}\cup \tilde\gamma_{+} \cup \tilde\gamma_{-}) \to \mathbb{C}^{2\times 2}$ is analytic, see Fig.~\ref{fig open lens contours Region 2}.
\item[(b)] $S$ has the following jumps:
\begin{gather*}
 S_{+}(z) = S_{-}(z)\begin{pmatrix}
1 & |z-t|^{\alpha}e^{n(\xi_{2,+}(z)+\xi_{2,-}(z))} \\
0 & 1
\end{pmatrix}, \qquad \mbox{if} \quad z < a, \\
 S_{+}(z) = S_{-}(z)\begin{pmatrix}
e^{2\pi i \Omega n} & |z-t|^{\alpha}e^{n(\xi_{1,+}(z)+\xi_{1,-}(z)-\lambda)} \\
0 & \displaystyle e^{-2\pi i \Omega n}
\end{pmatrix}, \qquad \mbox{if} \quad b<z<t, \\
 S_{+}(z) = S_{-}(z)\begin{pmatrix}
1 & |z-t|^{\alpha}e^{2n\xi_{1}(z)} \\
0 & 1
\end{pmatrix}, \qquad \mbox{if} \quad c < z, \\
 S_{+}(z) = S_{-}(z)\begin{pmatrix}
0 & |z-t|^{\alpha} \\ -|z-t|^{-\alpha} & 0
\end{pmatrix}, \qquad \mbox{if} \quad z \in \mathcal{S}, \\
 S_{+}(z) = S_{-}(z)\begin{pmatrix}
1 & 0 \\ (z-t)^{-\alpha}e^{-2n \xi_{2}(z)}e^{2\pi i \Omega n} & 1
\end{pmatrix}, \qquad \mbox{if} \quad z \in \tilde\gamma_{+}, \\
 S_{+}(z) = S_{-}(z)\begin{pmatrix}
1 & 0 \\ (z-t)^{-\alpha}e^{-2n \xi_{2}(z)}e^{-2\pi i \Omega n} & 1
\end{pmatrix}, \qquad \mbox{if} \quad z \in \tilde\gamma_{-}, \\
 S_{+}(z) = S_{-}(z)\begin{pmatrix}
1 & 0 \\ (z-t)^{-\alpha}e^{-2n \xi_{1}(z)} & 1
\end{pmatrix}, \qquad \mbox{if} \quad z \in \gamma_{+} \cup \gamma_{-}.
\end{gather*}
\item[(c)] As $z \to \infty$, we have $S(z) = I + \mathcal{O}\big(z^{-1}\big)$.

\item[(d)] As $z$ tends to $t$, we have
\begin{gather*}
 S(z) = \begin{cases}
\begin{pmatrix}
\mathcal{O}(1) & \mathcal{O}\big( \log (z-t)\big) \\
\mathcal{O}(1) & \mathcal{O}\big( \log (z-t)\big)
\end{pmatrix}, & z \mbox{outside the lenses}, \\
\begin{pmatrix}
\mathcal{O}\big(\log (z-t)\big) & \mathcal{O}\big(\log (z-t)\big) \\
\mathcal{O}\big(\log (z-t)\big) & \mathcal{O}\big(\log (z-t)\big)
\end{pmatrix}, & z \mbox{ inside the lenses},
\end{cases} \qquad \mbox{if} \quad \alpha = 0, \\
 S(z) = \begin{cases}
\begin{pmatrix}
\mathcal{O}(1) & \mathcal{O}(1) \\
\mathcal{O}(1) & \mathcal{O}(1)
\end{pmatrix}, & z \mbox{ outside the lenses}, \\
\begin{pmatrix}
\mathcal{O}\big((z-t)^{- \alpha }\big) & \mathcal{O}(1) \\
\mathcal{O}\big((z-t)^{- \alpha }\big) & \mathcal{O}(1)
\end{pmatrix}
, & z \mbox{ inside the lenses},
\end{cases} \qquad \mbox{if} \quad \alpha > 0, \\
 S(z) = \begin{pmatrix}
\mathcal{O}(1) & \mathcal{O}\big((z-t)^{ \alpha }\big) \\
\mathcal{O}(1) & \mathcal{O}\big((z-t)^{ \alpha }\big)
\end{pmatrix}
, \qquad \mbox{if} \quad \alpha < 0.
\end{gather*}
As $z$ tends to $a$, $b$ or $c$, we have $S(z) = \mathcal{O}(1)$.
\end{itemize}
From \eqref{lol 34}, \eqref{lol 32}, \eqref{lol 33} and the fact that $\xi_{1}(z)<0$ for $z >c$, the jumps for $S$ on the boundary of the lenses are exponentially close to the identity matrix as $n \to \infty$ and the $(1,2)$ entries of the jumps on~$\mathbb{R}\setminus \mathcal{S}$ are exponentially small as $n \to \infty$, but these convergences are not uniform for~$z$ in a neighbourhood of~$a$, $b$, $t$ and $c$. By ignoring the exponentially small terms in the jumps for~$S$, we are left with a simpler RH problem, whose solution $P^{(\infty)}$ will be a good approximation of~$S$ away from neighbourhoods of $a$, $b$, $t$ and $c$. We construct $P^{(\infty)}$ explicitly in Section~\ref{subsection: global param two cuts}.
\subsection{Global parametrix}\label{subsection: global param two cuts}
If we ignore the exponentially small terms as $n \to \infty$ in the jumps for $S$ and a small neighbourhood of $a$, $b$, $t$ and $c$, we are left with a simpler RH problem, whose solution $P^{(\infty)}$ is called the global parametrix, and is a good approximation of $S$ away from a neighbourhood of $a$, $b$, $t$ and~$c$.
\subsubsection*{RH problem for $\boldsymbol{P^{(\infty)}}$}
\begin{itemize}\itemsep=0pt
\item[(a)] $P^{(\infty)} \colon \mathbb{C}\setminus [a,c] \to \mathbb{C}^{2\times 2}$ is analytic.
\item[(b)] $P^{(\infty)}$ has the following jumps:
\begin{gather*}
 P^{(\infty)}_{+}(z) = P^{(\infty)}_{-}(z)\begin{pmatrix}
0 & |z-t|^{\alpha} \\ -|z-t|^{-\alpha} & 0
\end{pmatrix}, \qquad \mbox{if} \quad z \in \mathcal{S}, \\
 P^{(\infty)}_{+}(z) = P^{(\infty)}_{-}(z)e^{2\pi i \Omega n \sigma_{3}}, \qquad \mbox{if} \quad z \in (b,t).
\end{gather*}
\item[(c)] As $z \to \infty$, we have $P^{(\infty)}(z) = I + \mathcal{O}\big(z^{-1}\big)$.
\item[(d)]
As $z$ tends to $\tilde z \in \{a,b,c\}$, we have $P^{(\infty)}(z) = \mathcal{O}\big((z-\tilde z)^{-1/4}\big)$.

As $z$ tends to $t$, we have $P^{(\infty)}(z) = \mathcal{O}\big((z-t)^{-1/4}\big)(z-t)^{-\frac{\alpha}{2}\sigma_{3}}$.
\end{itemize}
The construction of $P^{(\infty)}$ has been done in similar situations in \cite{DKMVZ2} for $\alpha = 0$ and in \cite{KuijVanLess} for $\alpha \neq 0$. It involves $\theta$-functions and quantities related to a Riemann surface. Let $X$ be the two sheeted Riemann surface of genus one associated to $\sqrt{R(z)}$, with
\begin{gather*}
R(z) = (z-c)(z-t)(z-b)(z-a),
\end{gather*}
and we let $\sqrt{R(z)} \sim z^{2}$ as $z \to \infty$ on the f\/irst sheet. We also def\/ine cycles $A$ and $B$ such that they form a canonical homology basis of $X$. The upper part of the cycle $A$ (the dashed line in Fig.~\ref{fig cycles Riemann surface}) lies on the second sheet, and the lower part lies on the f\/irst sheet. The cycle $B$ surrounds $(a,b)$ in the clockwise direction, and lie in the f\/irst sheet.
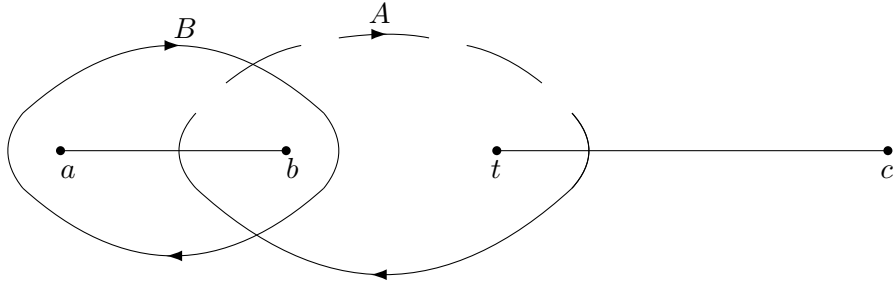
\begin{figure}[t]\centering
 \setlength{\unitlength}{1truemm}
 \begin{picture}(100,40)(0,20)
 \put(58,40){\line(1,0){52}}
 \put(0,40){\line(1,0){30}}
 \put(58,40){\thicklines\circle*{1.2}}
 \put(57.2,36.4){$t$}
 \put(110,40){\thicklines\circle*{1.2}}
 \put(109,36.4){$c$}
 \put(30,40){\thicklines\circle*{1.2}}
 \put(30,36.4){$b$}
 \put(0,40){\thicklines\circle*{1.2}}
 \put(0,36.4){$a$}
 \qbezier(-5,45)(15,63)(35,45)
 \qbezier(-5,35)(15,17)(35,35)
 \qbezier(-5,35)(-9,40)(-5,45)
 \qbezier(35,35)(39,40)(35,45)
 \put(15,55){$B$}
 \put(16,54){\thicklines\vector(1,0){.0001}}
 \put(14,26){\thicklines\vector(-1,0){.0001}}
 \qbezier(68,35)(43,12)(18,35)
 \qbezier(68,35)(72.5,40)(68,45)
 \qbezier(68,35)(72.5,40)(68,45)
 \qbezier(18,35)(13.5,40)(18,45)

 \qbezier(22,49)(27,53)(32,54) 
 \qbezier(37,55)(44,56)(49,55) 
 \qbezier(54,54)(59,53)(64,49) 
 \put(41.2,23.5){\thicklines\vector(-1,0){.0001}}
 \put(43.5,55.5){\thicklines\vector(1,0){.0001}}
 \put(41,57){$A$}
 \end{picture}
 \vspace{-3mm}

 \caption{The cycles $A$ and $B$. The solid line of $A$ is in the f\/irst sheet and the dashed line is in the second sheet. The cycle $B$ lies on the f\/irst sheet. \label{fig cycles Riemann surface}}
\end{figure}
The unique $A$-normalized holomorphic one-form $\omega$ on $X$ is given by
\begin{gather*}
\omega = \frac{c_{0}{\rm d}z}{\sqrt{R(z)}}, \qquad c_{0} = \left( \int_{A} \frac{1}{\sqrt{R(z)}}{\rm d}z \right)^{-1}.
\end{gather*}
By construction $\int_{A} \omega = 1$ and the lattice parameter is given by $\tau = \int_{B} \omega$. A direct calculation shows that
\begin{gather*}
 c_{0} = \left( \int_{b}^{t} \frac{2}{\sqrt{|R(x)|}}{\rm d}x \right)^{-1} \in \mathbb{R}^{+},
\qquad \tau = \int_{a}^{b} \frac{2 i c_{0}}{\sqrt{|R(x)|}}{\rm d}x \in i \mathbb{R}^{+}.
\end{gather*}
The associated $\theta$-function of the third kind $\theta(z) = \theta(z;\tau)$ is given by
\begin{gather*}
\theta(z) = \sum_{m=-\infty}^{\infty} e^{2\pi i mz}e^{\pi i m^{2} \tau}.
\end{gather*}
It is an entire function which satisf\/ies
\begin{gather}\label{lol 49}
\theta(z+1) = \theta(z), \qquad \theta(z+\tau) = e^{-2\pi i z}e^{-\pi i \tau}\theta(z), \qquad \mbox{for all} \quad z \in \mathbb{C}.
\end{gather}
We also need the function
\begin{gather*}
u(z) = \int_{c}^{z}\omega,
\end{gather*}
where the path of integration lies in $\mathbb{C}\setminus [a,c)$. Since $\int_{C}\omega = 0$ for any circle $C$ winding around~$\mathcal{S}$, $u$ is a single valued function for $z \in \mathbb{C}\setminus [a,c)$. For $z$ on the f\/irst sheet, it satisf\/ies
\begin{gather*}
 u_{+}(z) + u_{-}(z) = 0, \qquad z \in (t,c), \\
 u_{+}(z) + u_{-}(z) = 1, \qquad z \in (a,b), \\
 u_{+}(z) - u_{-}(z) = \tau, \qquad z \in (b,t), \\
 \lim_{z\to\infty}u(z) = u_{\infty} \in \mathbb{C}.
\end{gather*}
We def\/ine $\beta(z) = \sqrt[4]{\frac{z-a}{z-b}\frac{z-t}{z-c}}$, such that $\beta(z) \sim 1$ as $z \to \infty$ on the f\/irst sheet. On this f\/irst sheet, it can be verif\/ied that the function $\beta(z)+\beta^{-1}(z)$ never vanishes, while $\beta(z)-\beta^{-1}(z)$ vanishes at a single point $z_{\star}$, given by
\begin{gather*}
z_{\star} = \frac{c-t}{c-t+b-a}b + \frac{b-a}{c-t + b-a} t \in (b,t).
\end{gather*}
This observation will be useful later, because the solution $P^{(\infty)}$ will involve functions of the form $1/\theta(u(z) \pm d)$, where $d = -\frac{1}{2}-\frac{\tau}{2} + \int_{c}^{z_{\star}}\omega$. The function $1/\theta(u(z)+d)$ has no pole on the f\/irst sheet, while the function $1/\theta(u(z)-d)$ has a pole at $z_{\star}$, see \cite{DKMVZ2}. Therefore, the functions $\frac{\beta(z) \pm \beta^{-1}(z)}{\theta(u(z)\pm d)}$ are analytic on the f\/irst sheet.

Before stating the solution $P^{(\infty)}$, we follow \cite{KuijVanLess} and introduce a scalar Szeg\H{o} function $D$ which satisf\/ies
\begin{itemize}\itemsep=0pt
\item[(a)] $D \colon \mathbb{C}\setminus [a,c] \to \mathbb{C}$ is analytic.
\item[(b)] $D$ has the following jumps:
\begin{gather*}
 D_{+}(z)D_{-}(z) = |z-t|^{\alpha}, \qquad \mbox{if}\quad z \in \mathcal{S}, \\
 D_{+}(z) = D_{-}(z)e^{2\pi i \tilde \alpha}, \qquad \mbox{if}\quad z \in (b,t).
\end{gather*}
\item[(c)] As $z \to \infty$, we have $D(z) = D_{\infty} + \mathcal{O}\big(z^{-1}\big)$, $D_{\infty} \neq 0$.
\item[(d)] As $z \to t$, $D(z) = (z-t)^{\frac{\alpha}{2}}d_{t} + o\big((z-t)^{\frac{\alpha}{2}}\big)$, $d_{t} \neq 0$.

As $z$ tends to $\tilde z \in \{a,b,c\}$, $D(z)$ is bounded.
\end{itemize}
The new parameter $\tilde \alpha$ is not arbitrary and is completely determined to ensure the existence of~$D$. The solution is given by
\begin{gather}\label{Szego function}
D(z) = \exp \left( \sqrt{R(z)} \left[ \frac{1}{2\pi i} \int_{\mathcal{S}}\frac{\alpha \log|x-t|}{\sqrt{R(x)}_{+}}\frac{{\rm d}x}{x-z} + \tilde\alpha \int_{b}^{t} \frac{1}{\sqrt{R(x)}}\frac{{\rm d}x}{x-z} \right] \right),
\end{gather}
see \cite{KuijVanLess} for a detailed proof of it, and $\tilde \alpha$ is chosen such that $D(z)$ remains bounded as $z\to \infty$:
\begin{gather*}
\tilde \alpha = - \left( \int_{b}^{t} \frac{1}{\sqrt{R(x)}}{\rm d}x \right)^{-1} \frac{1}{2\pi i}\int_{\mathcal{S}} \frac{\alpha \log|x-t|}{\sqrt{R(x)}_{+}}{\rm d}x.
\end{gather*}
From \eqref{Szego function}, by expanding $\frac{1}{x-z}=-\frac{1}{z}\big(1+\frac{x}{z} + \mathcal{O}\big(z^{-2}\big)\big)$ as $z \to \infty$, we can evaluate $D_{\infty}$. We obtain
\begin{gather*}
D_{\infty} = \exp \left( - \frac{1}{2\pi i} \int_{\mathcal{S}} \frac{\alpha x \log |x-t|}{\sqrt{R(x)}_{+}}{\rm d}x - \tilde\alpha \int_{b}^{t} \frac{x}{\sqrt{R(x)}}{\rm d}x \right).
\end{gather*}
The solution $P^{(\infty)}$ is given by (see \cite[Section 4]{KuijVanLess}, and \cite[Lemma 4.3]{DKMVZ2})
\begin{gather}\label{Pinf}
P^{(\infty)}(z) =
\frac{D_{\infty}^{\sigma_{3}}}{2}\begin{pmatrix}
\beta(z)+\beta^{-1}(z)\Theta_{11}(z) & i(\beta(z)-\beta^{-1}(z))\Theta_{12}(z) \\
-i(\beta(z)-\beta^{-1}(z))\Theta_{21}(z) & \beta(z)+\beta^{-1}(z)\Theta_{22}(z)
\end{pmatrix}D(z)^{-\sigma_{3}},
\end{gather}
where
\begin{gather*}
 \Theta_{11}(z) = \frac{\theta(u_{\infty}+d)\theta(u(z)+d-\Omega n-\tilde\alpha)}{\theta(u_{\infty}+d-\Omega n-\tilde\alpha)\theta(u(z)+d)}, \qquad \Theta_{12}(z) = \frac{\theta(u_{\infty}+d)\theta(u(z)-d+\Omega n+\tilde\alpha)}{\theta(u_{\infty}+d-\Omega n-\tilde\alpha)\theta(u(z)-d)}, \\
 \Theta_{21}(z) = \frac{\theta(u_{\infty}+d)\theta(u(z)-d-\Omega n-\tilde\alpha)}{\theta(u_{\infty}+d+\Omega n+\tilde\alpha)\theta(u(z)-d)}, \qquad \Theta_{22}(z) = \frac{\theta(u_{\infty}+d)\theta(u(z)+d+\Omega n+\tilde\alpha)}{\theta(u_{\infty}+d+\Omega n+\tilde\alpha)\theta(u(z)+d)}.
\end{gather*}
\subsection{Local parametrices}
By the assumption made at the beginning of Section~\ref{Section: RH analysis region 2}, i.e., that $(t,\lambda)$ is in a compact subset of~$\mathcal{R}$ as $n \to \infty$, by Proposition~\ref{prop: equilibrium measure}, there exists $\delta > 0$ independent of $n$ such that
\begin{gather*}
\delta < \min\{ b-a, t-b, c-t \}.
\end{gather*}
We consider some disks $D_{a}$, $D_{b}$, $D_{t}$, $D_{c}$ around $a$, $b$, $t$ and $c$ respectively, such that the radii are smaller than $\delta/3$ but f\/ixed. Inside these disks, we require the local parametrices $P$ to have the same jumps as $S$, and the same behaviour near $a$, $b$, $t$ and $c$. Furthermore, uniformly for $z$ on the boundary of these disks, $P$ satisf\/ies the matching condition
\begin{gather*}
P(z) = \big(I+\mathcal{O}\big(n^{-1}\big)\big)P^{(\infty)}(z), \qquad \mbox{as} \quad n \to \infty.
\end{gather*}
The solution $P$ of the local parametrices around $a$, $b$ and $c$ is constructed in terms of the model Airy RH
problem, and the local parametrix~$P$ around~$t$ in terms of the modif\/ied Bessel model RH problem.
As these constructions do not present any additional technicalities other than the ones in Section~\ref{Section: RH analysis region 1}, and as we will not use them explicitly later, we have decided not to include them in the present article.

\subsection{Small norm RH problem}
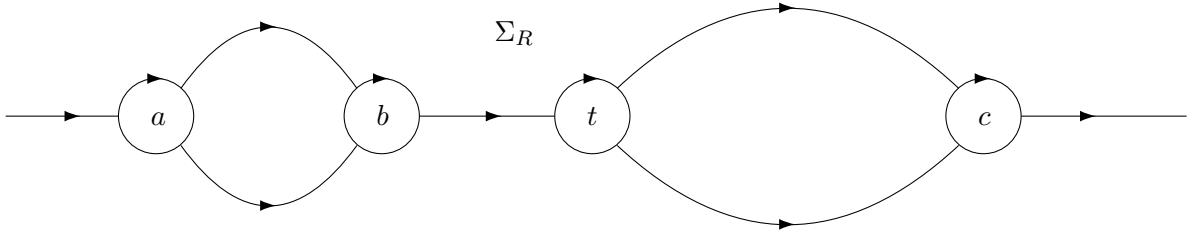
\begin{figure}[t]\centering
 \setlength{\unitlength}{1truemm}
 \begin{picture}(115,32)(0,25)
 \put(115,40){\line(1,0){22}}
 \put(35,40){\line(1,0){18}}
 \put(-5,40){\line(-1,0){15}}
 \put(45,50){$\Sigma_{R}$}
 \put(125,40){\thicklines\vector(1,0){.0001}}
 \put(57.3,38.8){$t$} \put(58,40){\circle{10}}\put(59,45){\thicklines\vector(1,0){.0001}}
 \put(109.3,38.8){$c$} \put(110,40){\circle{10}}\put(111,45){\thicklines\vector(1,0){.0001}}
 \qbezier(61.3,43.7)(84,65)(106.7,43.7)
 \put(85,54.4){\thicklines\vector(1,0){.0001}}
 \qbezier(61.3,36.1)(84,15)(106.7,36.1)
 \put(85,25.6){\thicklines\vector(1,0){.0001}}
 \put(29.3,38.8){$b$} \put(30,40){\circle{10}}\put(31,45){\thicklines\vector(1,0){.0001}}
 \put(-0.7,38.8){$a$} \put(0,40){\circle{10}}\put(1,45){\thicklines\vector(1,0){.0001}}
 \qbezier(3.3,43.7)(15,60)(26.7,43.7)
 \qbezier(3.3,36.1)(15,20)(26.7,36.1)
 \put(-10,40){\thicklines\vector(1,0){.0001}}
 \put(46,40){\thicklines\vector(1,0){.0001}}
 \put(16,51.9){\thicklines\vector(1,0){.0001}}
 \put(16,28.1){\thicklines\vector(1,0){.0001}}
 \end{picture}
 \vspace{-3mm}

 \caption{Jump contour for $R$. The orientation on the circles is clockwise. \label{fig R Region 2}}
\end{figure}
The f\/inal transformation is given by
\begin{gather*}
R(z) = \begin{cases}
S(z) P^{(\infty)}(z)^{-1}, & z \in \mathbb{C}\setminus \overline{D}, \\
S(z) P(z)^{-1}, & z \in D,
\end{cases}
\end{gather*}
where $D = D_{a}\cup D_{b} \cup D_{t} \cup D_{c}$. Since $P$ has exactly the same jumps as $S$ inside $D$, and the same behaviour near $a$, $b$, $t$ and $c$, we show in a similar way as done at the beginning of Section~\ref{subsection: small norm region 1} that $R$ is analytic inside these disks and it satisf\/ies the following RH problem.
\subsubsection*{RH problem for $\boldsymbol{R}$}

\begin{itemize}\itemsep=0pt
\item[(a)] $R \colon \mathbb{C}\setminus \Sigma_{R} \to \mathbb{C}^{2\times 2}$ is analytic, with the contour $\Sigma_{R}$ shown in Fig.~\ref{fig R Region 2}.
\item[(b)] The jumps $J_{R}(z) :=R_{-}^{-1}(z)R_{+}(z)$ satisfy the following large $n$ asymptotics for $z\in \Sigma_{R}$:
\begin{gather*}
 J_{R}(z) = I+\mathcal{O}\big(e^{-cn}\big), \qquad \mbox{uniformly for} \quad z \in (\gamma_{+}\cup\gamma_{-}\cup \tilde\gamma_{+}\cup\tilde\gamma_{-}\cup \mathbb{R})\setminus (\mathcal{S}\cup \overline{D}), \nonumber \\
 J_{R}(z) = I + \mathcal{O}\big(n^{-1}\big), \qquad \mbox{uniformly for}\quad z \in \partial D,
\end{gather*}
where $c>0$ is a constant.
\item[(c)] As $z \to \infty$, we have $R(z) = I+\mathcal{O}\big(z^{-1}\big)$.
\end{itemize}
Again, from standard theory for small-norm RH problems \cite{DKMVZ1}, $R$ exists for suf\/f\/iciently large $n$ and we have
\begin{gather}\label{lol 50}
R(z) = I + \mathcal{O}\big(n^{-1}\big), \qquad R^{\prime}(z) = \mathcal{O}\big(n^{-1}\big),
\end{gather}
uniformly for $z \in \mathbb{C}\setminus \Sigma_{R}$, and uniformly for $(t,\lambda)$ in a compact subset of $\mathcal{R}$.

\section[Asymptotics for the Hankel determinant $H_{n}(v,s,\alpha)$]{Asymptotics for the Hankel determinant $\boldsymbol{H_{n}(v,s,\alpha)}$}\label{Sec_asympHn}
\subsection{Proof of Theorem \ref{Th1}}
In this section, we use the RH analysis done in Section \ref{Section: RH analysis region 1} with $\lambda = +\infty$ (i.e., $s = 0$) and the dif\/ferential identity
\begin{gather}\label{lol 37}
\partial_{t} \log H_{n}\big(\sqrt{2n}t,0,\alpha\big) = 4n U_{1,11},
\end{gather}
which was obtained in \eqref{lol 36}. Inverting the transformations $R \mapsto S \mapsto T \mapsto U$, we obtain for $z$ outside the disks and outside the lenses that
\begin{gather*}
U(z) = e^{-\frac{n\ell}{2}\sigma_{3}}R(z)P^{(\infty)}(z)e^{ng(z)\sigma_{3}}e^{\frac{n\ell}{2}\sigma_{3}}.
\end{gather*}
In particular, the $(1,1)$ entry in the above expression is given by
\begin{gather*}
U_{11}(z) = e^{ng(z)}P_{11}^{(\infty)}(z) \left( R_{11}(z) + R_{12}(z) \frac{P_{21}^{(\infty)}(z)}{P_{11}^{(\infty)}(z)} \right).
\end{gather*}
Thus, by equations \eqref{asymptotics g} and \eqref{asymptotics Pinf}, we have
\begin{gather*}
U_{1,11} = -n \int_{\mathcal{S}}x\rho(x){\rm d}x - \frac{\alpha(\overline{c}-t)}{4} + \frac{R^{(1)}_{1,11}}{n} + \mathcal{O}\big(n^{-2}\big), \qquad \mbox{as} \quad n \to \infty,
\end{gather*}
where the $\mathcal{O}\big(n^{-2}\big)$ is uniform for $t$ in compact subsets of $(-1,\infty)$, and $R^{(1)}_{1,11}$ is the coef\/f\/icient of the $z^{-1}$ term in the large $z$ expansion of $R^{(1)}_{11}(z)$. From~\eqref{R first correction}, it is given by
\begin{gather*}
R^{(1)}_{1,11} = \sqrt{\overline{c}-t} \frac{4\alpha^{2}-1}{32 k_{1}} + \frac{-8\alpha^{2}+2(\overline{c}-t)k_{3}+3}{64 \sqrt{\overline{c}-t} k_{2}}.
\end{gather*}
On the other hand, by using \eqref{lol 35} we can calculate explicitly the f\/irst moment of $\rho$, we have
\begin{gather*}
\int_{t}^{\overline{c}} x \rho(x){\rm d}x = \frac{2}{27}\big(\sqrt{3 + t^2}-t\big)\big(3 + 5 t^2 + 4 t \sqrt{3 + t^2}\big).
\end{gather*}
Thus, we can rewrite the dif\/ferential identity \eqref{lol 37} more explicitly:
\begin{gather}\label{lol 7}
\partial_{t}\log H_{n}\big(\sqrt{2n}t,0,\alpha\big) = u_{1}(t) n^{2} + u_{2}(t,\alpha)n + u_{3}(t,\alpha) + \mathcal{O}\big(n^{-1}\big),
\end{gather}
where
\begin{gather*}
\displaystyle u_{1}(t) = -\frac{8}{27}\big(\sqrt{3+t^{2}}-t\big)\big(3+5t^{2}+4t\sqrt{3+t^{2}}\big), \\
\displaystyle u_{2}(t,\alpha) = -\frac{2\alpha}{3}\big(\sqrt{3+t^{2}}-t\big), \\
\displaystyle u_{3}(t,\alpha) = \frac{\big(\sqrt{3+t^{2}}-t\big) \big(t+\sqrt{3+t^{2}}\big(6\alpha^{2}-1\big)\big)}{12\big(3+t^{2}\big)\big(2t+\sqrt{3+t^{2}}\big)}.
\end{gather*}
Note that for $t>-1$ we have $\int_{0}^{t}u_{1}(x){\rm d}x = C_{1}(t)$, $\int_{0}^{t}u_{2}(x,\alpha){\rm d}x = C_{2}(t,\alpha)$ and $\int_{0}^{t}u_{3}(x,\alpha){\rm d}x = C_{3}(t,\alpha)$, where $C_{1}(t)$, $C_{2}(t,\alpha)$ and $C_{3}(t,\alpha)$ have been def\/ined in \eqref{C_1}, \eqref{C_2} and \eqref{C_3} respectively. Since the $\mathcal{O}\big(n^{-1}\big)$ term in~\eqref{lol 7} is uniform for $t$ in compact subsets of $(-1,\infty)$ and for $\lambda \geq \lambda_{c}(t)$, this gives the result.
\subsection{Proof of Theorem \ref{thm asymptotics for P with (t,lambda)}(1)}
In this section, we again use the RH analysis done in Section \ref{Section: RH analysis region 1}. In order to use the dif\/ferential identity \eqref{diff identity 2}, we need to obtain large $n$ asymptotics for $U$ uniformly on $(-\infty,t)$. This can be achieved by inverting the transformations $R \mapsto S \mapsto T \mapsto U$ in dif\/ferent regions. For $z \in D_{\overline{b}}\setminus \mathbb{R}$, we have
\begin{gather}\label{inv transformation region 1}
U(z) = e^{-\frac{n\ell}{2}\sigma_{3}}R(z)P^{(\infty)}(z)\begin{pmatrix}
1 & l(z) \\ 0 & 1
\end{pmatrix}e^{ng(z)\sigma_{3}}e^{\frac{n\ell}{2}\sigma_{3}}.
\end{gather}
For $z$ outside the lenses and outside the disks, $z \notin \mathbb{R}$, the expression for $U$ in terms of $R$ is
\begin{gather}\label{lol 46}
U(z) = e^{-\frac{n\ell}{2}\sigma_{3}}R(z)P^{(\infty)}(z)e^{ng(z)\sigma_{3}}e^{\frac{n\ell}{2}\sigma_{3}}.
\end{gather}
Thus, from \eqref{inv transformation region 1} and \eqref{lol 46}, for all $z$ such that $z \notin \mathbb{R}$ and $\Re z < t$, $z \notin D_{t}$, we obtain
\begin{gather}\label{lala1}
\big[U^{-1}(z)U^{\prime}(z)\big]_{21}\! =\! e^{n\ell}e^{2ng(z)} \big[P^{(\infty)}(z)^{-1}R(z)^{-1}R^{\prime}(z)P^{(\infty)}(z) \!+\! P^{(\infty)}(z)^{-1}P^{(\infty)}(z)^{\prime}\big]_{21}.\!\!\!
\end{gather}
From \eqref{Pinf Region 1}, as $n \to \infty$ we have the following bounds for the global parametrix:
\begin{gather*}
P^{(\infty)}(z) = \mathcal{O}(1), \qquad P^{(\infty)}(z)^{\prime} = \mathcal{O}(1), \qquad \mbox{uniformly for} \ z \ \mbox{outside the disks}.
\end{gather*}
Furthermore, by using the estimate for $R$ given by \eqref{lol 8}, and by taking the limit $z \to x \in \mathbb{R}$ in equation \eqref{lala1} (the limits from the upper and lower half plane are the same, see Remark \ref{remark: notation + and - not needed}), we obtain
\begin{gather}\label{lala6}
\frac{\tilde{w}(x)}{2\pi i}\left[ U^{-1}(x)U^{\prime}(x) \right]_{21} = |x-t|^{\alpha} e^{n(g_{+}(x)+g_{-}(x)+\ell-V(x))} \mathcal{O}(1), \qquad \mbox{as} \quad n \to \infty,
\end{gather}
where the $\mathcal{O}(1)$ is uniform for $x \in (-\infty,t)$, $x \notin D_{t}$. It is more complicated to obtain similar asymptotics for $\frac{\tilde{w}(x)}{2\pi i}\big[ U^{-1}(x)U^{\prime}(x) \big]_{21}$, uniformly for $x \in (-\infty,t)\cap D_{t}$. We will need the following lemma.
\begin{Lemma}\label{Prop: Bound on the OPs}
As $n \to \infty$, we have
\begin{align}
& U_{+}(x)\begin{pmatrix}
1 \\ 0
\end{pmatrix} = e^{\frac{n\ell}{2}}e^{ng_{+}(x)}e^{-\frac{n\ell}{2}\sigma_{3}}\begin{pmatrix}
\mathcal{O}\big(n^{\frac{1}{2}+\max(\alpha, 0)}\big) \\
\mathcal{O}\big(n^{\frac{1}{2}+\max(\alpha, 0)}\big)
\end{pmatrix}, \label{bound first column U}\\
& U_{+}^{\prime}(x)\begin{pmatrix}
1 \\ 0
\end{pmatrix} = e^{\frac{n\ell}{2}}e^{ng_{+}(x)}e^{-\frac{n\ell}{2}\sigma_{3}}\begin{pmatrix}
\mathcal{O}\big(n^{\frac{5}{2}+\max(\alpha, 0)}\big) \\
\mathcal{O}\big(n^{\frac{5}{2}+\max(\alpha, 0)}\big)
\end{pmatrix}, \label{bound second column U}
\end{align}
uniformly for $x \in (-\infty,t) \cap D_{t}$.
\end{Lemma}
\begin{proof}
For $z \in D_{t}$, $z$ outside the lenses, we have
\begin{gather}\label{lol 10}
U(z)\begin{pmatrix}
1 \\ 0
\end{pmatrix} = e^{\frac{n\ell}{2}}e^{ng(z)}e^{-\frac{n\ell}{2}\sigma_{3}}R(z)P(z)\begin{pmatrix}
1 \\ 0
\end{pmatrix}.
\end{gather}
If furthermore, $\Im z >0$, by \eqref{local param near t, Region 1} and \eqref{Psi explicit} we have
\begin{gather}\label{lol 9}
P(z)\begin{pmatrix}
1 \\ 0
\end{pmatrix} = e^{\frac{\pi i \alpha}{2}}(z-t)^{-\frac{\alpha}{2}}e^{-n\xi(z)}E(z) \begin{pmatrix}
I_{\alpha}(2n\sqrt{-f(z)}) \\ -2\pi i n \sqrt{-f(z)}I_{\alpha}^{\prime}(2n\sqrt{-f(z)})
\end{pmatrix}.
\end{gather}
Let $x \in (-\infty,t)\cap D_{t}$. Note that $\tilde\xi(x) > 0$ ($\tilde\xi$ is def\/ined in \eqref{lol 43}) and thus $\sqrt{-f(x)}_{+} = \frac{1}{2}\tilde\xi(x)$. Inserting \eqref{lol 9} into \eqref{lol 10}, we can take the limit $z\to x$, this gives
\begin{gather}
U_{+}(x)\begin{pmatrix}
1 \\ 0
\end{pmatrix} = (-1)^{n}e^{\frac{n\ell}{2}}e^{ng_{+}(x)}e^{-n\tilde\xi(x)}(t-x)^{-\frac{\alpha}{2}}e^{-\frac{n\ell}{2}\sigma_{3}}\nonumber\\
\hphantom{U_{+}(x)\begin{pmatrix}
1 \\ 0
\end{pmatrix} =}{} \times R(x)E(x)\begin{pmatrix}
I_{\alpha}(n\tilde\xi(x)) \\ - \pi i n \tilde\xi(x) I_{\alpha}^{\prime}(n\tilde\xi(x))
\end{pmatrix}.\label{lol 11}
\end{gather}
Since $E$ is analytic in $D_{t}$, one has from \eqref{E in local param near t, Region 1} that as $n\to \infty$
\begin{gather}\label{bound for E in D_t Region 1}
E(z) = \mathcal{O}(1)n^{\frac{\sigma_{3}}{2}}, \qquad E^{\prime}(z) = \mathcal{O}(1)n^{\frac{\sigma_{3}}{2}}, \qquad \mbox{uniformly for }z \in D_{t}.
\end{gather}
To obtain a uniform bound from \eqref{lol 11}, we distinguish three cases. Let $M>0$ be an arbitrary large but f\/ixed constant and let $m>0$ be an arbitrary small but f\/ixed constant.

\textbf{Case (a)}: $n \tilde \xi(x) \geq M$ as $n \to \infty$.
In this case we need large $\zeta$ asymptotics for~$I_{\alpha}(\zeta)$ and~$I_{\alpha}^{\prime}(\zeta)$. From \eqref{large z asymptotics Bessel}, we have
\begin{gather}\label{lol 12}
I_{\alpha}(\zeta) = \frac{e^{\zeta}}{\sqrt{2\pi\zeta}}\big(1+\mathcal{O}\big(\zeta^{-1}\big)\big), \qquad I_{\alpha}^{\prime}(\zeta) = \frac{e^{\zeta}}{\sqrt{2\pi\zeta}}\big(1+\mathcal{O}\big(\zeta^{-1}\big)\big), \qquad \mbox{as} \quad \zeta \to \infty.
\end{gather}
If we insert \eqref{lol 12} into \eqref{lol 11}, the result follows for Case (a) from \eqref{lol 8}, \eqref{bound for E in D_t Region 1} and from the fact that $(t-x)^{-\frac{\alpha}{2}} = \mathcal{O}(n^{\max( \alpha,0)})$.

\textbf{Case (b)}: $m \leq n \tilde \xi(x) \leq M$ as $n \to \infty$.
In this case we have $I_{\alpha}(n\tilde\xi(x)) \!=\! \mathcal{O}(1)$, $n \tilde\xi(x) I_{\alpha}^{\prime}(n\tilde\xi(x))$ $= \mathcal{O}(1)$, $e^{-n\tilde\xi(x)} = \mathcal{O}(1)$, $(t-x)^{-\frac{\alpha}{2}}=\mathcal{O}(n^{\alpha})$. Again from~\eqref{lol 8} and \eqref{bound for E in D_t Region 1}, we obtain
\begin{gather}\label{lol 13}
U_{+}(x)\begin{pmatrix}
1 \\ 0
\end{pmatrix} = e^{\frac{n\ell}{2}}e^{ng_{+}(x)}e^{-\frac{n\ell}{2}\sigma_{3}}\begin{pmatrix}
\mathcal{O}\big(n^{\frac{1}{2}+ \alpha}\big) \\ \mathcal{O}\big(n^{-\frac{1}{2}+ \alpha}\big)
\end{pmatrix},
\end{gather}
which is even slightly better than \eqref{bound first column U}.

\textbf{Case (c)}: $n \tilde \xi(x) \leq m$ as $n \to \infty$.
From \cite[formula (10.25.2)]{NIST}, we have
\begin{gather*}
I_{\alpha}(\zeta) = \left( \frac{\zeta}{2} \right)^{ \alpha} \left(\frac{1}{\Gamma(1+\alpha)}+\mathcal{O}\big(\zeta^{2}\big)\right), \\ I_{\alpha}^{\prime}(\zeta) = \left( \frac{\zeta}{2} \right)^{\alpha-1} \left(\frac{\alpha}{\Gamma(1+\alpha)}+\mathcal{O}\big(\zeta^{2}\big)\right), \qquad \mbox{as} \quad \zeta \to 0.
\end{gather*}
From the above expansion, we have for Case~(c) that
\begin{gather*}
\frac{I_{\alpha}(n\tilde\xi(x))}{(t-x)^{\frac{\alpha}{2}}} = \mathcal{O}\big(n^{ \alpha}\big), \qquad \frac{n \tilde\xi(x) I_{\alpha}^{\prime}(n\tilde\xi(x))}{(t-x)^{\frac{\alpha}{2}}} = \mathcal{O}\big(n^{ \alpha}\big)
\end{gather*}
and $e^{-n\tilde\xi(x)} = \mathcal{O}(1)$. Thus, from \eqref{lol 8} and \eqref{bound for E in D_t Region 1} we obtain again \eqref{lol 13}, which f\/inishes the proof of \eqref{bound first column U}. We now turn to the proof of \eqref{bound second column U}. From \eqref{lol 11}, we have $U_{+}^{\prime}(x)\left(\begin{smallmatrix}
1 \\ 0
\end{smallmatrix}\right) = \widetilde{U}_{1}(x) + \widetilde{U}_{2}(x)$,
where
\begin{gather*}
 \widetilde{U}_{1}(x) = n(g_{+}^{\prime}(x)-\tilde{\xi}^{\prime}(x))U_{+}(x) \begin{pmatrix}
1 \\ 0
\end{pmatrix} + (-1)^{n}e^{\frac{n\ell}{2}}e^{ng_{+}(x)}e^{-n\tilde{\xi}(x)}(t-x)^{-\frac{\alpha}{2}}e^{-\frac{n\ell}{2}\sigma_{3}} \\
\hphantom{\widetilde{U}_{1}(x) =}{} \times(R^{\prime}(x)E(x) + R(x)E^{\prime}(x)) \begin{pmatrix}
I_{\alpha}(n\tilde\xi(x)) \\ - \pi i n \tilde\xi(x) I_{\alpha}^{\prime}(n\tilde\xi(x))
\end{pmatrix}, \\
 \widetilde{U}_{2}(x) = (-1)^{n}e^{\frac{n\ell}{2}}e^{ng_{+}(x)}e^{-n\tilde{\xi}(x)} e^{-\frac{n\ell}{2}\sigma_{3}} R(x)E(x) \begin{pmatrix}
\displaystyle \left( \frac{I_{\alpha}(n\tilde\xi(x))}{(t-x)^{\frac{\alpha}{2}}} \right)^{\prime} \\
\displaystyle \left( \frac{- \pi i n \tilde\xi(x) I_{\alpha}^{\prime}(n\tilde\xi(x))}{(t-x)^{\frac{\alpha}{2}}} \right)^{\prime}
\end{pmatrix}.
\end{gather*}
The analysis of $\widetilde{U}_{1}(x)$ and $\widetilde{U}_{2}(x)$ can be done very similarly to the f\/irst part of the proof and we do not provide here all the details. From \eqref{lol 14}, one has $g_{+}^{\prime}(x) - \tilde\xi^{\prime}(x) = 2x$ and thus by~\eqref{lol 8}, \eqref{bound first column U} and~\eqref{bound for E in D_t Region 1},
\begin{gather*}
\widetilde{U}_{1}(x) = e^{\frac{n\ell}{2}}e^{ng_{+}(x)}e^{-\frac{n\ell}{2}\sigma_{3}}\begin{pmatrix}
\mathcal{O}\big(n^{\frac{3}{2}+\max(\alpha, 0)}\big) \\
\mathcal{O}\big(n^{\frac{3}{2}+\max(\alpha, 0)}\big)
\end{pmatrix}.
\end{gather*}
Again, by splitting the analysis into the same three cases as in the f\/irst part of the proof, we obtain the estimates
\begin{align*}
& \left( \frac{I_{\alpha}(n\tilde\xi(x))}{(t-x)^{\frac{\alpha}{2}}} \right)^{\prime} = \mathcal{O}\left(n^{2}\left( \frac{I_{\alpha}(n\tilde\xi(x))}{(t-x)^{\frac{\alpha}{2}}} \right)\right), \\
& \left( \frac{-i\pi n \tilde \xi(x) I_{\alpha}^{\prime}(n\xi(x))}{(t-x)^{\frac{\alpha}{2}}} \right)^{\prime}= \mathcal{O}\left(n^{2}\left( \frac{-i\pi n \tilde \xi(x) I_{\alpha}^{\prime}(n\xi(x))}{(t-x)^{\frac{\alpha}{2}}} \right)\right),
\end{align*}
which yields $\widetilde{U}_{2}(x) = e^{\frac{n\ell}{2}}e^{ng_{+}(x)}e^{-\frac{n\ell}{2}\sigma_{3}}\begin{pmatrix}
\mathcal{O}\big(n^{\frac{5}{2}+\max(\alpha,0)}\big) \\
\mathcal{O}\big(n^{\frac{5}{2}+\max(\alpha,0)}\big)
\end{pmatrix}$ and f\/inishes the proof.
\end{proof}

Note that $g_{+}(x)+g_{-}(x)-2x^{2}+\ell$ is continuous on $\mathbb{R}$ and equal to $0$ at $x = t$ by~\eqref{var equality}. Thus, from \eqref{g+ + g-} and \eqref{var inequality} and the fact that $V(x)$ has a jump discontinuity at $x = t$, we have
\begin{gather}\label{lol 66}
\lim_{\substack{x\to t \\ x < t}} g_{+}(x) + g_{-}(x) -V(x)+\ell = - \lambda < -\lambda_{c} < 0.
\end{gather}
Therefore, by using f\/irst Lemma \ref{Prop: Bound on the OPs} and then \eqref{lol 66}, there exists $c \in (0,\lambda_{c})$ such that
\begin{gather}
\frac{\tilde{w}(x)}{2\pi i}\big[U^{-1}(x)U^{\prime}(x)\big]_{21} = |x-t|^{\alpha}e^{n(g_{+}(x)+g_{-}(x)-V(x)+\ell)}\mathcal{O}\big(n^{3+2\max(\alpha,0)}\big)\nonumber\\
\hphantom{\frac{\tilde{w}(x)}{2\pi i}\big[U^{-1}(x)U^{\prime}(x)\big]_{21}}{} = |x-t|^{\alpha}\mathcal{O}\big(e^{-(\lambda-c)n}\big),\label{lol 15}
\end{gather}
as $n \to \infty$ uniformly for $x \in D_{t}\cap (-\infty,t)$. Now, we will split the integral of the dif\/ferential identity \eqref{diff identity 2} into two parts:
\begin{gather*}
 s\partial_{s} \log H_{n}\big(\sqrt{2n}t,s\big) = I_{1}(s)+I_{2}(s), \\
 I_{1}(s) = \int_{(-\infty,t)\setminus D_{t}}\frac{\widetilde{w}(x)}{2\pi i}\big[ U^{-1}(x)U^{\prime}(x) \big]_{21}{\rm d}x, \\
 I_{2}(s) = \int_{(-\infty,t)\cap D_{t}}\frac{\widetilde{w}(x)}{2\pi i}\big[ U^{-1}(x)U^{\prime}(x) \big]_{21}{\rm d}x.
\end{gather*}
The f\/irst integral can be evaluated using \eqref{lala6}. By \eqref{lol 29}, \eqref{lol 48}, \eqref{lol 14} and \eqref{lol 28} (see also the comment just after), we have $g_{+}(\overline{b})+g_{-}(\overline{b})+\ell-V(\overline{b}) = -(\lambda-\lambda_{c})$ and
\begin{gather*}
(g_{+}(x)+g_{-}(x)+\ell-V(x))^{\prime}\big|_{x=\overline{b}} = 0, \qquad
(g_{+}(x)+g_{-}(x)+\ell-V(x)^{\prime\prime}\big|_{x=\overline{b}} < 0.
\end{gather*}
Therefore, we obtain
\begin{gather*}
\big|I_{1}\big(s = e^{-\lambda n}\big)\big| = \mathcal{O}\big(n^{-1/2}e^{-n(\lambda-\lambda_{c})}\big), \qquad \mbox{as} \quad n \to \infty.
\end{gather*}
On the other hand, from \eqref{lol 15}, it immediately follows that
\begin{gather*}
\big|I_{2}\big(s = e^{-\lambda n}\big)\big| = \mathcal{O}\big(e^{-(\lambda-c)n}\big), \qquad \mbox{as} \quad n \to \infty,
\end{gather*}
where $c \in (0,\lambda_{c})$. Therefore, the dif\/ferential identity becomes
\begin{gather*}
\partial_{s} \log H_{n}(v,s,\alpha) \big|_{s = e^{-\lambda n}} = \mathcal{O}\big(n^{-1/2}e^{n\lambda_{c}}\big), \qquad \mbox{as} \quad n \to \infty,
\end{gather*}
where in the above expression the $\mathcal{O}$ term is uniform for $t$ in a compact subset of $(-1,\infty)$ and for $\lambda \geq \lambda_{c}(t)$. Thus, we can integrate it from $s = 0$ to $s = e^{-\lambda n}$, and it gives
\begin{gather*}
\log H_{n}\big(\sqrt{2n}t,e^{-\lambda n},\alpha\big) = \log H_{n}\big(\sqrt{2n}t,0,\alpha\big) + \mathcal{O}\big(n^{-1/2}e^{-n(\lambda-\lambda_{c}(t))}\big), \qquad \mbox{as} \quad n \to \infty,
\end{gather*}
which is the claim \eqref{lala10}.
\subsection{Proof of Theorem \ref{thm asymptotics for P with (t,lambda)}(2)}
In this section we use the RH analysis done in Section \ref{Section: RH analysis region 2}.
\begin{Proposition} \label{prop: convergence of Kn to rho} Let $\mathcal{W} \subset \mathbb{R}$ be an arbitrary small but fixed neighbourhood of the four points $\{a,b,t,c\}$. We have as $n \to \infty$
\begin{gather*}
\frac{\tilde{w}(x)}{2\pi i}\big[ U^{-1}(x)U^{\prime}(x) \big]_{21}-n\rho(x)\chi_{\mathcal{S}}(x) = e^{n(g_{+}(x)+g_{-}(x)+\ell - V(x))}\mathcal{O}(1),
\end{gather*}
uniformly for $x \in \mathbb{R}\setminus \mathcal{W}$.
\end{Proposition}
\begin{proof}
We can assume without loss of generality that the disks of the local parametrices are suf\/f\/iciently small such that $D \subset \mathcal{W}$. Let $z$ be outside the lenses and outside the disks. In this region, by inverting the transformations $R \mapsto S \mapsto T \mapsto U$, we have
\begin{gather*}
U(z) = e^{-\frac{n\ell}{2}\sigma_{3}}R(z)P^{(\infty)}(z) e^{ng(z)\sigma_{3}}e^{\frac{n\ell}{2}\sigma_{3}}.
\end{gather*}
Since the dependence in $n$ of the global parametrix \eqref{Pinf} appears only in the form $n\Omega \in \mathbb{R}$, and as an argument of the $\theta$-function, by the periodicity property \eqref{lol 49}, as $n \to \infty$ we have
\begin{gather*}
P^{(\infty)}(z) = \mathcal{O}(1), \qquad P^{(\infty)}(z)^{\prime} = \mathcal{O}(1), \quad \mbox{uniformly for } z \mbox{ outside the disks}.
\end{gather*}
Therefore, using also the large $n$ asymptotics for $R$ \eqref{lol 50}, we have
\begin{gather}\label{lol 16}
\big[U(z)^{-1}U^{\prime}(z)\big]_{21} = e^{n\ell} e^{2ng(z)}\mathcal{O}(1), \qquad \mbox{as} \quad n \to \infty,
\end{gather}
uniformly for $z$ outside the lenses and outside the disks.
For $x \in \mathbb{R}\setminus (\mathcal{S}\cup \mathcal{W})$, we can take the limit $z\to x$ in~\eqref{lol 16}. As $n \to \infty$, we have
\begin{gather*}
\frac{\tilde{w}(x)}{2\pi i}\big[ U^{-1}(x)U^{\prime}(x) \big]_{21} = |x-t|^{\alpha}e^{n(g_{+}(x)+g_{-}(x)+\ell - V(x))}\mathcal{O}(1),
\end{gather*}
uniformly for $x\in \mathbb{R}\setminus (\mathcal{S}\cup \mathcal{W})$. Now, we consider the case when $z$ is still outside the disks but inside $\mathcal{I}_{1}$, see \eqref{lol 17} and Fig.~\ref{fig open lens contours Region 2}. Inverting the transformations in this region, we get
\begin{gather*}
U(z) = e^{-\frac{n\ell}{2}\sigma_{3}}R(z)P^{(\infty)}(z)\begin{pmatrix}
1 & 0 \\ (z-t)^{-\alpha}e^{-2n\xi_{1}(z)} & 1
\end{pmatrix}e^{ng(z)\sigma_{3}}e^{\frac{n\ell}{2}\sigma_{3}}.
\end{gather*}
Since $P^{(\infty)}(z) = \mathcal{O}(1)$ as $n \to \infty$ uniformly for $z$ in this region, we have
\begin{gather}\label{lol 18}
\big[ U^{-1}(z)U^{\prime}(z) \big]_{21} = (z-t)^{-\alpha}e^{n(2g(z)+\ell)}\big({-}2n \xi_{1}^{\prime}(z)e^{-2n\xi_{1}(z)}+\mathcal{O}(1)\big), \qquad \mbox{as} \quad n \to \infty,\!\!\!\!
\end{gather}
where we have also used $\eqref{lol 32}$ and $\Re \xi_{1,+}(x) =0$ for $x \in (t,c)$. Note that from \eqref{def of xi1 and xi2}, we have $\xi_{1,+}^{\prime}(x) = -\pi i \rho(x)$ for $x \in (t,c)$. Thus, if we let $z \to x \in (t,c)\setminus \mathcal{W}$ in \eqref{lol 18}, from \eqref{var equality} and~\eqref{xi_1 in terms of g}, we have
\begin{gather}\label{lol 19}
\frac{\tilde{w}(x)}{2\pi i}\big[ U^{-1}(x)U^{\prime}(x) \big]_{21} = n\rho(x) \big(1+\mathcal{O}\big(n^{-1}\big)\big), \qquad \mbox{as} \quad n \to \infty,
\end{gather}
where the $\mathcal{O}$ term in the above expression is uniform for $x\in (t,c)\setminus \mathcal{W}$. For $x \in (a,b)\setminus \mathcal{W}$, we can invert the transformations for $z \in \widetilde{\mathcal{I}}_{1}$ and then take the limit $z \to x$. The computations are similar and we obtain the same asymptotics as~\eqref{lol 19}.
\end{proof}

By \eqref{Y definition} and \eqref{Y to U transformation}, note that \eqref{lol 20} can be rewritten as $\int_{\mathbb{R}}\frac{\tilde{w}(x)}{2\pi i}\big[ U^{-1}(x)U^{\prime}(x) \big]_{21}{\rm d}x = n$. Thus, a consequence of Proposition~\ref{prop: convergence of Kn to rho} (by taking $\mathcal{W}$ arbitrarily small) and \eqref{diff identity 2} is that for f\/ixed $t \in (-1,1)$ and f\/ixed $\lambda \in (0,\lambda_{c}(t))$, we have
\begin{gather}\label{lol 21}
\lim_{n\to\infty} \frac{s}{n}\partial_{s} \log H_{n}\big(\sqrt{2n}t,s,\alpha\big)\big|_{s=e^{-\lambda n}} = \lim_{n\to\infty} \int_{-\infty}^{t}\frac{\widetilde{w}(x)}{2\pi in}\big[ U^{-1}(x)U^{\prime}(x) \big]_{21}{\rm d}x = \Omega(t,\lambda).\!\!\!
\end{gather}
A simple change of variables shows that
\begin{gather*}
 \frac{s}{n}\partial_{s} \log H_{n}\big(\sqrt{2n}t,s,\alpha\big)\big|_{s=e^{-\lambda n}} = -\frac{1}{n^{2}} \partial_{\lambda} \log H_{n}\big(\sqrt{2n}t,e^{-\lambda n},\alpha\big).
\end{gather*}
By \eqref{lol 21}, for every $(t,\lambda)$ such that $t \in (-1,1)$ and $\lambda \in (0,\lambda_{c}(t))$, the right-hand side of the above expression converges to $\Omega(t,\lambda)$ as $n \to \infty$. Also, by \eqref{E in terms of Kn} and \eqref{diff identity 1}, we have
\begin{gather*}
\frac{s}{n}\partial_{s} \log H_{n}(v,s,\alpha) = \frac{\mathcal{E}_{n}(v,s,\alpha)}{n}\leq 1.
\end{gather*}
Since the constant function $1$ is integrable on any bounded interval, we can apply Lebesgue's dominated convergence theorem, and we have
\begin{gather*}
\lim_{n\to\infty} \frac{-1}{n^{2}}\int_{0}^{\lambda} \partial_{\tilde\lambda}\log H_{n}\big(\sqrt{2n}t,e^{-\tilde \lambda n},\alpha\big){\rm d}\tilde \lambda = \int_{0}^{\lambda} \Omega\big(t,\tilde \lambda\big){\rm d}\tilde \lambda,
\end{gather*}
which f\/inishes the proof.
\begin{Remark}
As mentioned in Remark \ref{remark : pointwise convergence}, we have indeed only used pointwise convergence for $\lambda \in (0,\lambda_{c}(t))$ of the quantity $ \frac{s}{n}\partial_{s} \log H_{n}\big(\sqrt{2n}t,s,\alpha\big)\big|_{s=e^{-\lambda n}}$ to $\Omega(t,\lambda)$ as $n \to \infty$ and Lebesgue's theorem. The technical RH analysis as $\lambda \to 0$ or $\lambda \to \lambda_{c}(t)$ was thus not needed.
\end{Remark}
\subsection{Direct proof of formula (\ref{lol 52})}
In this section we suppose that $t \in (-1,1)$ and $\lambda \in (0,\lambda_{c}(t))$, but as we will have to integrate in~$\lambda$ over the interval $[0,\lambda_{c}(t)]$, some quantities need also to be def\/ined for $\lambda = 0$ and for $\lambda = \lambda_{c}(t)$. The quantities $\rho(x;t,\lambda)$ and $\ell(t,\lambda)$ refer to~\eqref{lol 45} and~\eqref{Euler-Lagrange constant} if $\lambda \in (0,\lambda_{c}(t))$, to~\eqref{lol 35} and~\eqref{EL constant in region 1} if $\lambda = \lambda_{c}(t)$, and to \eqref{sc law} if $\lambda = 0$. Also, $\Omega(t,\lambda)$ is given by~\eqref{Omega def} for $\lambda \in (0,\lambda_{c}(t))$, and we def\/ine by continuity $\Omega(t,\lambda_{c}(t)) = 0$.
\begin{Lemma}\label{lemma: no riemann surface}For $t\in(-1,1)$ and $\lambda \in (0,\lambda_{c}(t))$, there holds a relation between $\Omega(t,\lambda)$, the density $\rho(x;t,\lambda)$ given by \eqref{lol 45}, and the Euler--Lagrange constant $\ell(t,\lambda)$ given by \eqref{Euler-Lagrange constant}:
\begin{gather}\label{lol 57}
\Omega(t,\lambda) = \partial_{\lambda} \ell(t,\lambda) + \partial_{\lambda}\int_{\mathcal{S}}2x^{2}\rho(x;t,\lambda){\rm d}x + \lambda \partial_{\lambda} \Omega(t,\lambda).
\end{gather}
\end{Lemma}
\begin{proof}
Consider the function
\begin{gather*}
H(x;t,\lambda) = -2 \int_{\mathcal{S}}\log|x-y|\rho(y;t,\lambda){\rm d}y.
\end{gather*}
By the Euler--Lagrange equality \eqref{var equality}, we have
\begin{gather}
 H(x;t,\lambda) = \ell(t,\lambda) - 2x^{2}, \qquad x \in (t,c), \label{lol 53}\\
 H(x;t,\lambda) = \ell(t,\lambda) - 2x^{2} - \lambda, \qquad x \in (a,b).\label{lol 54}
\end{gather}
Thus, by integrating it with respect to $\rho(x;t,\lambda){\rm d}x$, we obtain
\begin{gather*}
\int_{\mathcal{S}}H(x;t,\lambda)\rho(x;t,\lambda){\rm d}x = \ell(t,\lambda) - \lambda \Omega(t,\lambda) - \int_{\mathcal{S}} 2x^{2}\rho(x;t,\lambda){\rm d}x.
\end{gather*}
We will evaluate $\partial_{\lambda}\int_{\mathcal{S}}H(x;t,\lambda)\rho(x;t,\lambda){\rm d}x$ in two dif\/ferent ways. From the above expression, it gives
\begin{gather}\label{lol 55}
\partial_{\lambda}\int_{\mathcal{S}}H(x;t,\lambda)\rho(x;t,\lambda){\rm d}x = \partial_{\lambda}\ell(t,\lambda) - \Omega(t,\lambda) -\lambda \partial_{\lambda}\Omega(t,\lambda) - \partial_{\lambda}\int_{\mathcal{S}} 2x^{2}\rho(x;t,\lambda){\rm d}x.\!\!\!\!\!\!
\end{gather}
On the other hand, by Lebesgue's dominated convergence theorem, and by the symmetry in $x$ and $y$, we have
\begin{gather*}
\partial_{\lambda}\int_{\mathcal{S}}\int_{\mathcal{S}}\log|x-y|\rho(y;t,\lambda)\rho(x;t,\lambda){\rm d}y{\rm d}x = 2 \int_{\mathcal{S}}\int_{\mathcal{S}}\log|x-y|\partial_{\lambda}(\rho(y;t,\lambda))\rho(x;t,\lambda){\rm d}y{\rm d}x.
\end{gather*}
Therefore, by dif\/ferentiating \eqref{lol 53} and \eqref{lol 54} with respect to $\lambda$, we obtain
\begin{gather}\label{lol 56}
\partial_{\lambda}\int_{\mathcal{S}}H(x;t,\lambda)\rho(x;t,\lambda){\rm d}x = 2 \int_{\mathcal{S}} \partial_{\lambda}\left(H(x;t,\lambda)\right)\rho(x;t,\lambda){\rm d}x = 2\partial_{\lambda}\ell(t,\lambda) - 2 \Omega(t,\lambda).\!\!\!\!
\end{gather}
Putting \eqref{lol 55} and \eqref{lol 56} together, we obtain \eqref{lol 57}.
\end{proof}

Let us consider the function
\begin{gather*}
F(t,\lambda) = \ell(t,\lambda) + \int_{\mathcal{S}}2x^{2}\rho(x;t,\lambda){\rm d}x,
\end{gather*}
\begin{Lemma}\label{lemma: Omega as a derivative}
For $t\in(-1,1)$ and $\lambda \in (0,\lambda_{c}(t))$, we have the following relation
\begin{gather}\label{lol 58}
\Omega(t,\lambda) = \partial_{\lambda} \left[ \frac{F(t,\lambda_{c}(t))}{2}\left( \frac{\lambda}{\lambda_{c}(t)} \right)^{2} + \int_{\frac{\lambda}{\lambda_{c}(t)}}^{1}\xi F\big(t,\tfrac{\lambda}{\xi}\big){\rm d}\xi \right].
\end{gather}
\end{Lemma}
\begin{proof}
From a direct calculation and a change of variables, the right-hand side of \eqref{lol 58} is equal to
\begin{gather}\label{lol 59}
\int_{\frac{\lambda}{\lambda_{c}(t)}}^{1} \xi \partial_{\lambda}F\big(t,\tfrac{\lambda}{\xi}\big){\rm d}\xi = \int_{\frac{\lambda}{\lambda_{c}(t)}}^{1} \partial_{u}F(t,u) \big|_{u = \frac{\lambda}{\xi}}{\rm d}\xi = \lambda \int_{\lambda}^{\lambda_{c}(t)} \frac{\partial_{u}F(t,u)}{u^{2}}{\rm d}u.
\end{gather}
By using \eqref{lol 57}, which can be rewritten as $\partial_{\lambda}F(t,\lambda) = \Omega(t,\lambda)-\lambda\partial_{\lambda}\Omega(t,\lambda)$, the right-hand side of \eqref{lol 59} becomes
\begin{gather*}
\lambda \int_{\lambda}^{\lambda_{c}(t)} \frac{\Omega(t,u)-u \partial_{u}\Omega(t,u)}{u^{2}}{\rm d}u = \Omega(t,\lambda),
\end{gather*}
where the last equality is obtained via an integration by parts, and using the identity $\Omega(t,\lambda_{c}(t))$ $= 0$.
\end{proof}

\begin{Lemma}
\begin{gather}\label{lol 60}
-\int_{0}^{\lambda_{c}(t)} \Omega(t,\lambda) {\rm d}\lambda = C_{1}(t)-\frac{\log 3}{2}.
\end{gather}
\end{Lemma}
\begin{proof}
From Lemma \ref{lemma: Omega as a derivative}, we directly obtain that
\begin{gather}\label{lol 61}
-\int_{0}^{\lambda_{c}(t)} \Omega(t,\lambda) {\rm d}\lambda = -\frac{1}{2} (F(t,\lambda_{c}(t))-F(t,0)).
\end{gather}
By \eqref{lol 35} and \eqref{EL constant in region 1}, we obtain
\begin{gather}
F(t,\lambda_{c}(t)) = \frac{3}{2} + 2 \left( \frac{4}{3}t^{2}+\frac{5}{9}t\sqrt{3+t^{2}} \right) + \frac{4t^{3}}{27}\big(\sqrt{3+t^{2}}-t\big)\nonumber\\
\hphantom{F(t,\lambda_{c}(t)) =}{} +2\log \big( 2\big(t+\sqrt{3+t^{2}}\big) \big),\label{lol 62}
\end{gather}
and by \eqref{sc law}, we have
\begin{gather}\label{lol 63}
F(t,0) = \frac{3}{2} + 2\log 2.
\end{gather}
By substituting \eqref{lol 62} and \eqref{lol 63} into \eqref{lol 61}, we obtain \eqref{lol 60}.
\end{proof}

\appendix
\section{Airy model RH problem}\label{ApA}
We consider the following RH problem:

\begin{itemize}\itemsep=0pt
\item[(a)] $P_{\mathrm{Ai}} \colon \mathbb{C} \setminus \Sigma_{A} \rightarrow \mathbb{C}^{2 \times 2}$ is analytic, where $\Sigma_{A}$ is shown in Fig.~\ref{figAiry}.
\item[(b)] $P_{\mathrm{Ai}}$ has the jump relations
\begin{gather*}
P_{\mathrm{Ai},+}(\zeta) = P_{\mathrm{Ai},-}(\zeta) \begin{pmatrix}
0 & 1 \\ -1 & 0
\end{pmatrix}, \qquad \mbox{on} \quad \mathbb{R}^{-}, \\
P_{\mathrm{Ai},+}(\zeta) = P_{\mathrm{Ai},-}(\zeta) \begin{pmatrix}
 1 & 1 \\
 0 & 1
\end{pmatrix}, \qquad \mbox{on} \quad \mathbb{R}^{+}, \\
P_{\mathrm{Ai},+}(\zeta) = P_{\mathrm{Ai},-}(\zeta) \begin{pmatrix}
 1 & 0 \\ 1 & 1
\end{pmatrix}, \qquad \mbox{on} \quad e^{ \frac{2\pi i}{3} } \mathbb{R}^{+} \cup e^{ -\frac{2\pi i}{3} }\mathbb{R}^{+}.
\end{gather*}
\item[(c)] As $\zeta \to \infty$, $z \notin \Sigma_{A}$, we have
\begin{gather}\label{Asymptotics Airy}
P_{\mathrm{Ai}}(\zeta) = \zeta^{-\frac{\sigma_{3}}{4}}N \left( I + \sum_{k=1}^{\infty} A_{k} \zeta^{-3k/2} \right) e^{-\frac{2}{3}\zeta^{3/2}\sigma_{3}},
\end{gather}
where $N = \frac{1}{\sqrt{2}}\left(\begin{smallmatrix}
1 & i \\ i & 1
\end{smallmatrix}\right)$ and $A_{1} = \frac{1}{8} \left(\begin{smallmatrix}
\frac{1}{6} & i \\ i & -\frac{1}{6}
\end{smallmatrix}\right)$.
\end{itemize}
This model RH problem was introduced for the f\/irst time and solved in \cite{DKMVZ1}, and is now well-known. The unique solution of the above RH problem is given in terms of Airy functions, we have
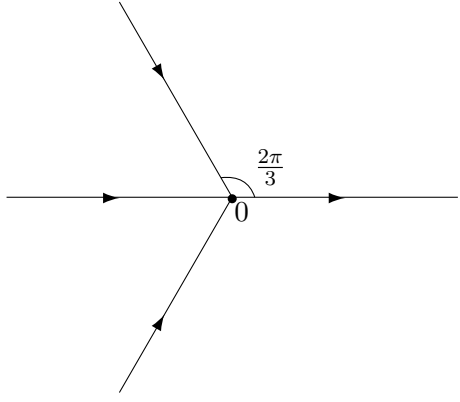
\begin{figure}[t]\centering
 \setlength{\unitlength}{1truemm}
 \begin{picture}(100,55)(-5,10)
 \put(50,40){\line(1,0){30}}
 \put(50,40){\line(-1,0){30}}
 \put(50,39.8){\thicklines\circle*{1.2}}
 \put(50,40){\line(-0.5,0.866){15}}
 \put(50,40){\line(-0.5,-0.866){15}}
 \qbezier(53,40)(52,43)(48.5,42.598)
 \put(53,43){$\frac{2\pi}{3}$}
 \put(50.3,36.8){$0$}
 \put(65,39.9){\thicklines\vector(1,0){.0001}}
 \put(35,39.9){\thicklines\vector(1,0){.0001}}
 \put(41,55.588){\thicklines\vector(0.5,-0.866){.0001}}
 \put(41,24.412){\thicklines\vector(0.5,0.866){.0001}}
 \end{picture}
 \vspace{-3mm}

 \caption{\label{figAiry}The jump contour $\Sigma_{A}$ for $P_{\mathrm{Ai}}(\zeta)$.}
\end{figure}
\begin{gather*}
P_{\mathrm{Ai}}(\zeta) := M_{A} \times \begin{cases}
\begin{pmatrix}
\mbox{Ai}(\zeta) & \mbox{Ai}(\omega^{2}\zeta) \\
\mbox{Ai}^{\prime}(\zeta) & \omega^{2}\mbox{Ai}^{\prime}(\omega^{2}\zeta)
\end{pmatrix}e^{-\frac{\pi i}{6}\sigma_{3}}, & \mbox{for } 0 < \arg \zeta < \frac{2\pi}{3}, \vspace{1mm}\\
\begin{pmatrix}
\mbox{Ai}(\zeta) & \mbox{Ai}(\omega^{2}\zeta) \\
\mbox{Ai}^{\prime}(\zeta) & \omega^{2}\mbox{Ai}^{\prime}(\omega^{2}\zeta)
\end{pmatrix}e^{-\frac{\pi i}{6}\sigma_{3}}\begin{pmatrix}
1 & 0 \\ -1 & 1
\end{pmatrix}, & \mbox{for } \frac{2\pi}{3} < \arg \zeta < \pi, \vspace{1mm}\\
\begin{pmatrix}
\mbox{Ai}(\zeta) & - \omega^{2}\mbox{Ai}(\omega \zeta) \\
\mbox{Ai}^{\prime}(\zeta) & -\mbox{Ai}^{\prime}(\omega \zeta)
\end{pmatrix}e^{-\frac{\pi i}{6}\sigma_{3}}\begin{pmatrix}
1 & 0 \\ 1 & 1
\end{pmatrix}, & \mbox{for } -\pi < \arg \zeta < -\frac{2\pi}{3}, \vspace{1mm}\\
\begin{pmatrix}
\mbox{Ai}(\zeta) & - \omega^{2}\mbox{Ai}(\omega \zeta) \\
\mbox{Ai}^{\prime}(\zeta) & -\mbox{Ai}^{\prime}(\omega \zeta)
\end{pmatrix}e^{-\frac{\pi i}{6}\sigma_{3}}, & \mbox{for } -\frac{2\pi}{3} < \arg \zeta < 0,
\end{cases}
\end{gather*}
with $\omega = e^{\frac{2\pi i}{3}}$, Ai the Airy function and
\begin{gather*}
M_{A} := \sqrt{2 \pi} e^{\frac{\pi i}{6}} \begin{pmatrix}
1 & 0 \\ 0 & -i
\end{pmatrix}.
\end{gather*}

\section{Bessel model RH problem}\label{ApB}
We consider the following RH problem:

\begin{itemize}\itemsep=0pt
\item[(a)] $P_{\mathrm{Be}} \colon \mathbb{C} \setminus \Sigma_{B} \to \mathbb{C}^{2\times 2}$ is analytic, where
$\Sigma_{B}$ is shown in Fig.~\ref{figBessel}.

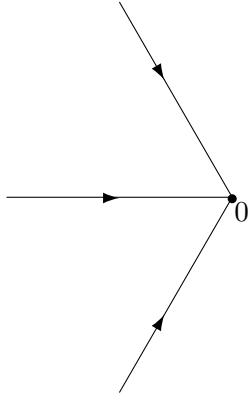
\begin{figure}[t]\centering
 \setlength{\unitlength}{1truemm}
 \begin{picture}(100,55)(-5,10)
 \put(50,40){\line(-1,0){30}}
 \put(50,39.8){\thicklines\circle*{1.2}}
 \put(50,40){\line(-0.5,0.866){15}}
 \put(50,40){\line(-0.5,-0.866){15}}
 \put(50.3,36.8){$0$}
 \put(35,39.9){\thicklines\vector(1,0){.0001}}
 \put(41,55.588){\thicklines\vector(0.5,-0.866){.0001}}
 \put(41,24.412){\thicklines\vector(0.5,0.866){.0001}}
 \end{picture}
 \vspace{-3mm}

 \caption{\label{figBessel}The jump contour $\Sigma_{B}$ for $P_{\mathrm{Be}}(\zeta)$.}
\end{figure}
\item[(b)] $P_{\mathrm{Be}}$ satisf\/ies the jump conditions
\begin{gather}
P_{\mathrm{Be},+}(\zeta) = P_{\mathrm{Be},-}(\zeta) \begin{pmatrix}
0 & 1 \\ -1 & 0
\end{pmatrix}, \qquad \zeta \in \mathbb{R}^{-}, \nonumber\\
P_{\mathrm{Be},+}(\zeta) = P_{\mathrm{Be},-}(\zeta) \begin{pmatrix}
1 & 0 \\ e^{\pi i \alpha} & 1
\end{pmatrix}, \qquad \zeta \in e^{ \frac{2\pi i}{3} } \mathbb{R}^{+},\nonumber \\
P_{\mathrm{Be},+}(\zeta) = P_{\mathrm{Be},-}(\zeta) \begin{pmatrix}
1 & 0 \\ e^{-\pi i \alpha} & 1
\end{pmatrix}, \qquad \zeta \in e^{ -\frac{2\pi i}{3} } \mathbb{R}^{+}. \label{Jump for P_Be}
\end{gather}
\item[(c)] As $\zeta\to \infty$, $\zeta \notin \Sigma_{B}$, we have
\begin{gather}\label{large z asymptotics Bessel}
P_{\mathrm{Be}}(\zeta) = \big( 2\pi \zeta^{\frac{1}{2}} \big)^{-\frac{\sigma_{3}}{2}}N
\left(I+\sum_{k=1}^{\infty} B_{k} \zeta^{-k/2}\right) e^{2\zeta^{\frac{1}{2}}\sigma_{3}},
\end{gather}
where $N = \frac{1}{\sqrt{2}}\left(\begin{smallmatrix}
1 & i \\ i & 1
\end{smallmatrix}\right)$ and $B_{1} = \frac{1}{16}\left(\begin{smallmatrix}
-(1+4\alpha^{2}) & -2i \\ -2i & 1+4\alpha^{2}
\end{smallmatrix}\right)$.
\item[(d)] As $\zeta$ tends to 0, the behaviour of $P_{\mathrm{Be}}(\zeta)$ is
\begin{gather}
 P_{\mathrm{Be}}(\zeta) = \begin{cases}
\begin{pmatrix}
\mathcal{O}(1) & \mathcal{O}(\log \zeta) \\
\mathcal{O}(1) & \mathcal{O}(\log \zeta)
\end{pmatrix}, & |\arg \zeta| < \frac{2\pi}{3}, \\
\begin{pmatrix}
\mathcal{O}(\log \zeta) & \mathcal{O}(\log \zeta) \\
\mathcal{O}(\log \zeta) & \mathcal{O}(\log \zeta)
\end{pmatrix}, & \frac{2\pi}{3}< |\arg \zeta| < \pi,
\end{cases} \qquad \mbox{if} \quad \alpha = 0, \nonumber\\
\displaystyle P_{\mathrm{Be}}(\zeta) = \begin{cases}
\begin{pmatrix}
\mathcal{O}(1) & \mathcal{O}(1) \\
\mathcal{O}(1) & \mathcal{O}(1)
\end{pmatrix}\zeta^{\frac{\alpha}{2}\sigma_{3}}, & |\arg \zeta | < \frac{2\pi}{3}, \\
\begin{pmatrix}
\mathcal{O}\big(\zeta^{-\frac{\alpha}{2}}\big) & \mathcal{O}\big(\zeta^{-\frac{\alpha}{2}}\big) \\
\mathcal{O}\big(\zeta^{-\frac{\alpha}{2}}\big) & \mathcal{O}\big(\zeta^{-\frac{\alpha}{2}}\big)
\end{pmatrix}, & \frac{2\pi}{3}<|\arg \zeta | < \pi,
\end{cases} \qquad \mbox{if} \quad \alpha > 0, \nonumber\\
 P_{\mathrm{Be}}(\zeta) = \begin{pmatrix}
\mathcal{O}\big(\zeta^{\frac{\alpha}{2}}\big) & \mathcal{O}\big(\zeta^{\frac{\alpha}{2}}\big) \\
\mathcal{O}\big(\zeta^{\frac{\alpha}{2}}\big) & \mathcal{O}\big(\zeta^{\frac{\alpha}{2}}\big)
\end{pmatrix}, \qquad \mbox{if} \quad \alpha < 0.\label{local behaviour near 0 of P_Be}
\end{gather}
\end{itemize}
This RH problem was introduced and solved in \cite{KMcLVAV}. Its unique solution is given by
\begin{gather}\label{Psi explicit}
P_{\mathrm{Be}}(\zeta)=
\begin{cases}
\begin{pmatrix}
I_{\alpha}\big(2\zeta^{\frac{1}{2}}\big) & \frac{ i}{\pi} K_{\alpha}\big(2\zeta^{\frac{1}{2}}\big) \\
2\pi i \zeta^{\frac{1}{2}} I_{\alpha}^{\prime}\big(2\zeta^{\frac{1}{2}}\big) & -2 \zeta^{\frac{1}{2}} K_{\alpha}^{\prime}\big(2\zeta^{\frac{1}{2}}\big)
\end{pmatrix}, & |\arg \zeta | < \frac{2\pi}{3}, \vspace{1mm}\\
\begin{pmatrix}
\frac{1}{2} H_{\alpha}^{(1)}\big(2(-\zeta)^{\frac{1}{2}}\big) & \frac{1}{2} H_{\alpha}^{(2)}\big(2(-\zeta)^{\frac{1}{2}}\big) \\
\pi \zeta^{\frac{1}{2}} \big( H_{\alpha}^{(1)} \big)^{\prime} \big(2(-\zeta)^{\frac{1}{2}}\big) & \pi \zeta^{\frac{1}{2}} \big( H_{\alpha}^{(2)} \big)^{\prime} \big(2(-\zeta)^{\frac{1}{2}}\big)
\end{pmatrix}\!e^{\frac{\pi i \alpha}{2}\sigma_{3}},\!\!\!\! & \frac{2\pi}{3} < \arg \zeta < \pi, \vspace{1mm}\\
\begin{pmatrix}
\frac{1}{2} H_{\alpha}^{(2)}\big(2(-\zeta)^{\frac{1}{2}}\big) & -\frac{1}{2} H_{\alpha}^{(1)}\big(2(-\zeta)^{\frac{1}{2}}\big) \\
-\pi \zeta^{\frac{1}{2}} \big( H_{\alpha}^{(2)} \big)^{\prime} \big(2(-\zeta)^{\frac{1}{2}}\big) & \pi \zeta^{\frac{1}{2}} \big( H_{\alpha}^{(1)} \big)^{\prime} \big(2(-\zeta)^{\frac{1}{2}}\big)
\end{pmatrix}\!e^{-\frac{\pi i \alpha}{2}\sigma_{3}},\!\!\!\! & -\pi < \arg \zeta < -\frac{2\pi}{3},
\end{cases}\hspace*{-30mm}
\end{gather}
where $H_{\alpha}^{(1)}$ and $H_{\alpha}^{(2)}$ are the Hankel functions of the f\/irst and second kind, and $I_\alpha$ and $K_\alpha$ are the modif\/ied Bessel functions of the f\/irst and second kind.

\subsection*{Acknowledgements}
C.~Charlier was supported by the European Research Council under the European Union's Seventh Framework Programme (FP/2007/2013)/ ERC Grant Agreement n.~307074. A.~Dea\~{n}o acknowledges f\/inancial support from projects MTM2012-36732-C03-01 and MTM2015-65888-C4-2-P from the Spanish Ministry of Economy and Competitivity. The authors are grateful to A.B.J.~Kuijlaars for sharing a simplif\/ied proof for the f\/irst part of \cite[Proposition~A.1]{ChCl2}. This inspired us to simplify the proof of Lemma~\ref{lemma: no riemann surface}. We also thank T.~Claeys for a careful reading of the introduction and for useful remarks. The authors acknowledge the referees for their careful reading and useful remarks.

\pdfbookmark[1]{References}{ref}
\LastPageEnding

\end{document}